\documentclass[journal,onecolumn,draftclsnofoot,]{IEEEtran}
\usepackage{etoolbox}
\usepackage{enumerate}
\usepackage[colorlinks=true]{hyperref}
 \hypersetup{
     colorlinks=true,
     linkcolor=blue,
     filecolor=blue,
     citecolor = Magenta,      
     urlcolor=cyan,
     }
     
\usepackage[dvipsnames]{xcolor}
\usepackage{graphicx}
\usepackage{bm}
\usepackage{mathtools}
\usepackage{graphicx}
\usepackage{epstopdf}
\usepackage{csquotes}
\usepackage{subfig}
\usepackage{mathtools}
\usepackage{amssymb, amsmath,amsthm, amsfonts}
\usepackage{ifthen}
\usepackage{tikz}
\usepackage{cite}
\usepackage{enumerate}
\usepackage{verbatim}
\usepackage{pifont}
\usepackage{bm}  
\usepackage{stackengine}
\usepackage{bbm}

\allowdisplaybreaks
\usepackage{accents}
\patchcmd{\subsubsection}{\itshape}{\bfseries}{}{}

\usepackage{xpatch}
\makeatletter
\AtBeginDocument{\xpatchcmd{\@thm}{\thm@headpunct{.}}{\thm@headpunct{}}{}{}}
\makeatother

\DeclareMathOperator*{\argmax}{arg\,max}

\DeclareSymbolFont{bbold}{U}{bbold}{m}{n}
\DeclareSymbolFontAlphabet{\mathbbold}{bbold}

\newtheorem{theorem}{Theorem}
\newtheorem{cor}[theorem]{Corollary}
\newtheorem{lemma}[theorem]{Lemma}
\newtheorem{prop}[theorem]{Proposition}

\newtheorem{remark}[theorem]{Remark}
\newtheorem{definition}{Definition}

\newcommand\Item[1][]{%
  \ifx\relax#1\relax  \item \else \item[#1] \fi
  \abovedisplayskip=0pt\abovedisplayshortskip=0pt~\vspace*{-\baselineskip}}

\newcommand{\norm}[1]{\left\Vert#1\right\Vert}
\newcommand{\abs}[1]{\left\vert#1\right\vert}

\newcommand{\bra}[1]{\langle#1|}
\newcommand{\ket}[1]{|#1\rangle}

\newcommand{\cA}{\mathcal{A}}

\newcommand{\cC}{\mathcal{C}}

\newcommand{\cE}{\mathcal{E}}

\newcommand{\cH}{\mathcal{H}}
\newcommand{\cI}{\mathcal{I}}

\newcommand{\cL}{\mathcal{L}}
\newcommand{\cM}{\mathcal{M}}

\newcommand{\cP}{\mathcal{P}}

\newcommand{\cS}{\mathcal{S}}
\newcommand{\cT}{\mathcal{T}}
\newcommand{\cU}{\mathcal{U}}

\newcommand{\cX}{\mathcal{X}}

\newcommand{\EE}{\mathbb{E}}

\newcommand{\HH}{\mathbb{H}}

\newcommand{\NN}{\mathbb{N}}

\newcommand{\PP}{\mathbb{P}}

\newcommand{\RR}{\mathbb{R}}

\newcommand{\tr}[1]{\operatorname{Tr}{\left[#1\right]}}

\newcommand{\ind}{\mathbbm 1}

\newcommand{\X}{\mathcal{X}}
\newcommand{\M}{\mathcal{M}}

\newcommand{\spec}[1]{\mathsf{spec}(#1)}

\newcommand{\qrel}[2]{\mathsf{D}\left(#1\middle\|#2\right)}
\newcommand{\mrel}[2]{\mathsf{D}_\mathcal{M}\left(#1\middle\|#2\right)}

\newcommand{\sanddiv}[3]{\tilde{\mathsf{D}}_{#3}\left(#1\middle\|#2\right)}
\newcommand{\sanddivvar}[4]{\tilde{\mathsf{D}}_{#3}\left(#1\middle\|#2;#4\right)}
\newcommand{\petzdiv}[3]{\bar{\mathsf{D}}_{#3}\left(#1\middle\|#2\right)}

\newcommand{\trightarrow}[1]{\overset{#1}{\longrightarrow}}

\newcommand\blfootnote[1]{%
	\begingroup
	\renewcommand\thefootnote{}\footnote{#1}%
	\addtocounter{footnote}{-1}%
	\endgroup
}

\definecolor{cblue}{rgb}{0.16, 0.32, 0.75}

\def\h2{\tilde h}

\def\hm1{\hat h_{-1}}

\title{Limit Distribution Theory for Quantum Divergences}
\author{Sreejith Sreekumar and Mario Berta}

\begin{document}
\maketitle
\vspace{-1 cm}
\begin{abstract}
Estimation of quantum relative entropy and its R\'{e}nyi generalizations is a fundamental statistical task in quantum information theory, physics, and beyond. While several estimators of these divergences have been proposed in the literature  along with their computational complexities explored, a limit distribution theory which characterizes the asymptotic fluctuations of the estimation error is still premature. As our main contribution, we characterize these   asymptotic distributions  in terms of Fr\'{e}chet derivatives of elementary operator-valued functions. We achieve this by leveraging an operator version of Taylor's theorem and identifying the regularity conditions needed.  
As an application of our results, we consider an estimator of quantum relative entropy based on  Pauli tomography of quantum states and show that the resulting asymptotic distribution is a centered normal, with its variance characterized in terms of the Pauli operators and states. We  utilize the knowledge of the aforementioned limit distribution to obtain asymptotic performance guarantees for a multi-hypothesis testing problem. 
\end{abstract}
\begin{IEEEkeywords}
Quantum divergences, limit distribution, divergence estimation, Fr\'{e}chet derivative, hypothesis testing
\end{IEEEkeywords}
\section{Introduction}
\blfootnote{S. Sreekumar (email: sreekumar@physik.rwth-aachen.de) and M. Berta are with the  Institute for Quantum Information at RWTH Aachen University, Germany. }
Estimation of a quantum state, also known as quantum state tomography, is an important problem  in quantum information theory, physics, and quantum machine learning, see e.g.,  \cite{Smithey-1993,Blume-Kohout_2010,Cramer-2010,Christandl-2011,HHJWY16,wright2016learn,OW16,Yuen2023improvedsample}.    In several applications, however, the quantity of interest may not be the entire state, but only a functional of it. Quantum divergences such as  quantum relative entropy \cite{Umegaki-62} 
and its R\'{e}nyi generalizations \cite{Renyi-60,Petz1985Quasi-entropiesAlgebra,Petz1986Quasi-entropiesSystems,muller2013quantum,wilde2014strong} form an important class of such functionals. They play a central role  in quantum information theory both in terms of characterizing fundamental limits as well as  applications, e.g.,  see the books   \cite{wilde2017quantum, tomamichel2015quantum, Hayashi-book-2016}. For instance, the quantum relative entropy characterizes the error-exponent in asymmetric binary quantum hypothesis testing \cite{Hiai-Petz-1991} and the Petz-R\'{e}nyi  divergence quantifies the exponent in quantum Chernoff bounds\cite{Nussbaum-Szkola-2006,Audenaert-2008}.   Owing to their significance, several estimators  of these measures have been proposed recently in the literature  and their performance investigated in terms of benchmarks such as copy and query complexity (see \textit{Related work} section below). However, a limit distribution theory which characterizes the asymptotic distribution of estimation error  is largely unexplored.

Here, we seek a limit distribution theory for the aforementioned  quantum divergences. Given   two quantum states $\rho$ and $\sigma$ with corresponding estimators  $\rho_n$ and $\sigma_n$, respectively, and a divergence $\mathsf{D}(\rho,\sigma)$, we want to identify the scaling  rate $r_n$ (or convergence rate $r_n^{-1}$) and the limiting variable $Z$ such that the following convergence in distribution (weak convergence) holds: 
\begin{align}
    r_n\big(\mathsf{D}(\rho_n,\sigma_n)-\mathsf{D}(\rho,\sigma)\big) \trightarrow{w} Z. \notag
\end{align}
Of interest is also the scenario where only one state, say $\rho$ or $\sigma$, is estimated and the other is known. Characterization of such limit distributions  have several  potential applications in quantum statistics and machine learning such as constructing confidence intervals for quantum hypothesis testing, asymptotic analysis of quantum algorithms, and quantum statistics (see \cite{Okano-2024, SGK-IT-2023,Pietrzak-2015} for some classical applications).

While limit distributions fully quantify the asymptotic performance, deriving such results for  estimators of quantum divergences are  challenging on account of two reasons. Firstly,   limit distributions need not always exist, as is well-known for relative entropy in the classical setting. Secondly, the non-commutative framework of quantum theory makes the analysis more involved. For tackling the first challenge, we use an operator version of Taylor's expansion with remainder and ascertain primitive  conditions for the existence of limits. The technical core of our contribution entails determining conditions that allow interchange of limiting operations on trace functionals of Fr\'{e}chet derivatives that appear in such an expansion. For handling  issues arising due to non-commutativity, we use appropriate integral expressions for operator functions and dual formulations for divergences.

Applying the aforementioned method  to quantum relative entropy, we  establish the following convergence in distribution (Theorem \ref{Thm:quantrelent-limdist}) when $r_n(\rho_n-\rho) \trightarrow{w} L_1$ and $r_n(\sigma_n-\sigma) \trightarrow{w} L_2$  for $\rho \neq \sigma$: 
\begin{align}
    r_n \big(\qrel{\rho_n}{\sigma_n}-\qrel{\rho}{\sigma}\big) \trightarrow{w}  \tr{L_1 (\log \rho-\log \sigma )-\rho D[\log \sigma](L_2)}. \notag
\end{align}
Here\footnote{Throughout, we consider logarithms to the base $e$.}, $D[f(A)](B)$ denotes the first-order Fr\'{e}chet derivative (see Section \ref{Sec:Prelims} for precise definitions) of an operator-valued function $f$ at operator $A$ in the direction of operator $B$, and $L_1$ (resp. $L_2$) denotes the weak limit of the estimator $\rho_n$ (resp. $\sigma_n$), appropriately centered and scaled. Analogous to the classical case, a faster convergence rate  is achieved when $\rho=\sigma$ with the limit characterized in terms of second-order derivatives. 
We then consider two prominent  quantum R\'{e}nyi divergences, namely, the Petz-R\'{e}nyi \cite{Petz1985Quasi-entropiesAlgebra} and the minimal or sandwiched R\'{e}nyi divergences  of order $\alpha$ \cite{muller2013quantum,wilde2014strong}. 
While the former admits  an analysis similar to Theorem \ref{Thm:quantrelent-limdist},   for the latter, we  consider a variational form based on dual expressions  and derive the limits by applying Taylor's theorem to an intermediate quantity which is easier to analyze.  
We also use a similar approach to derive the limit distribution for measured relative entropy estimation under a certain assumption on the uniqueness of the optimal measurement. 

As an application of our limit distribution results, we  characterize the asymptotic distribution of an estimator of quantum relative entropy based on Pauli tomography of  quantum states  $\rho,\sigma$. Specifically, we show  that 
\begin{align}
    \sqrt{n} \big(\qrel{\hat \rho_n}{\hat \sigma_n}-\qrel{\rho}{\sigma}\mspace{-2 mu}\big) \trightarrow{w}  W , \notag
\end{align}
where  $W$ is  a  centered Gaussian variable with a variance that depends on the states and Pauli operators. We then use this result to obtain performance guarantees for a multi-hypothesis testing problem for determining the quantum relative entropy between an unknown state $\rho$ and a known state $\sigma$. Assuming that identical copies of the unknown state are available for measurement, we first perform tomography  to obtain an estimate of $\rho$ and 
 then use the knowledge of the Gaussian limit  to design a test statistic (decision rule) that achieves any desired error level for appropriately chosen thresholds. The test statistic achieves the same performance even when the number of hypotheses scales at a sufficiently slow rate with the number of measurements. Such tests have potential applications to  auditing of  quantum differential privacy \cite{Hirche-2023}, as considered in \cite{SGK-ISIT-2023,SGK-IT-2023} for the classical case.

\subsection{Related Work}
Statistical analysis of estimators of classical information measures and divergences has been an active area of research over the past few decades. The relevant literature pertains broadly to showing consistency, quantifying convergence rates  of estimators (or equivalently sample complexity), and characterizing their  limiting distributions. Consistency and/or convergence rates for various estimators of $f$-divergences, which subsumes entropy and mutual information as special cases, have been studied in \cite{Paninski-2003,Wang-2005,Leonenko-2008,Perez-2008,kandasamy2015nonparametric, Singh-Poczos-2016,Noshad-2017,belghazi2018,Moon-2018,Nguyen-2010,berrett2019efficient,Yanjun-2020,SS-2021-aistats,sreekumar2021neural}. 
Limit distributions for several $f$-divergence estimators  such as those based on kernel density estimates, k-nearest neighbour methods, and plug-in methods have been established recently  \cite{Antos-Ioannis-2001,Zhang-2012,kandasamy2015nonparametric,Ioannis-Skoularidou-2016,Moon-2014,Goldfeld-Kato-2020,SGK-ISIT-2023,SGK-IT-2023}, while corresponding results  for R\'{e}nyi divergences have been studied in \cite{Pietrzak-2015}. Limit distribution theory has also been explored extensively in the optimal transport literature for the class of Wasserstein distances \cite{Barrio-1999,DBGU-2005,Sommerfeld2016Inference,del-barrio-loubes-2019,Tameling-2019,delBarrio-2022,MBWW-2022,HKSM-22}, as well as their regularized versions  \cite{Bigot-2019,Mena-2019,Klatt-2020,HLP-20,Goldfeld2020limit_wass,GX-21,delBarrio-2022,Gonzalez-2022,Sadhu-2022,GKNR-smooth-p-2022,GKRS-2022-entropicOT,GKRS-22}.

In the quantum setting,  computational complexities of various estimators of quantum information measures have been investigated under different input models \cite{Li2017QuantumEstimation,Subasi2019EntanglementCircuit,Jayadev-2020, Subramanian2019QuantumStates,Gilyen2019DistributionalWorld,Gur2021SublinearEntropy,chen2021variational,rethinasamy2021estimating,Wang2022QuantumEntropies,Wang2022NewDistances,Gilyen2022ImprovedEstimation,Wang2022QuantumSystems,
 Wang2023QuantumEstimation,GPSW-2024,quek2022multivariate}. Specifically,
 \cite{Jayadev-2020} established   copy complexity bounds characterizing the optimal dimension dependence for quantum R\'{e}nyi entropy estimation when independent copies of the state are available for measurement. \cite{Subramanian2019QuantumStates,Gur2021SublinearEntropy} considered entropy estimation under a quantum query model, which assumes  access to an oracle that prepares the input quantum state. 
 For limit distributional results in the quantum setting, the
asymptotic distribution for spectrum estimation of a quantum state based on the empirical Young's diagram (EYD) algorithm \cite{Alicki-1988,Keyl-Werner-2001} was determined in  \cite{Tracy-widom-1999,Houdre-Xu-2013}. 
 However, to the best of our knowledge, a limit distribution theory  for quantum divergences has  not been explored before.  
Here,  we study  this aspect focusing mostly on finite dimensional quantum systems (except in Section \ref{Sec:infdim} where we treat quantum relative entropy between density operators on an infinite-dimensional separable Hilbert space).

\subsection{Paper Organization}
The rest of the paper is organized as follows. Section \ref{Sec:Prelims} introduces the notation and preliminary concepts required for stating our results. The main results on limit distributions of quantum divergences are  presented in Section \ref{Sec:Mainres}. The applications, namely limit distributions of quantum relative entropy for  tomographic estimators and performance guarantees for a multi-hypothesis testing problem,  are discussed in Section \ref{Sec:applic}. This is followed by concluding remarks with avenues for future research in Section \ref{Sec:conclusion}. The proofs of the main results and applications are furnished in Section \ref{Sec:Proof} while those of the auxiliary lemmas are provided in the Appendix. 
\section{Preliminaries}\label{Sec:Prelims}
\subsection{Notation} \label{Sec:notation}
For most part, we consider a finite dimensional complex Hilbert space $\HH_d$ of dimension $d$. Denote  the set of linear operators from $\HH_d$ to $\HH_d$ by $\cL(\HH_d)$. Without loss of generality, we  identify $\HH_d$ and $\cL(\HH_d)$ with $\mathbb{C}^d$ and $\mathbb{C}^{d \times d}$, respectively. Denote the set of all $d \times d$  Hermitian, positive semi-definite, positive definite, and unitary operators by   $\cH_d$,   $\cP_d$, $\cP^+_d$ and $\cU_d$, respectively. Let $\cS_d $ denote the set of density operators, i.e., the set of elements of $ \cP_d$ with unit trace, and $\cS^+_d$ be its subset with strictly positive eigenvalues. We use $[A,B]:=AB-BA$ to represent the commutator of two operators $A$ and $B$.   $\tr{\cdot}$ and $\norm{\cdot}_p$ for $p \geq 1$ signifies the trace operation and Schatten $p$-norm, respectively. The notation $\leq$ denotes the L\"{o}wner partial order in the context of operators, i.e., for $A,B \in \cH_d$, $A \leq B$ means that $B-A \in \cP_d$. $\ind_{\cX}$ denotes indicator of a set $\cX$ and $I$ denotes the identity operator on $\mathbb{H}_d$. For linear operators $A,B$, $A \ll B$ designates that the support of $A$ is contained in that of $B$, and $A \ll \gg B$ means that $B \ll A \ll B$. 
$A^{-1}$ stands for the generalized (Moore–Penrose) inverse of an operator $A$. For $a,b \in \RR$, $a \vee b :=\max\{a,b\}$ and $a \wedge b :=\min\{a,b\}$. Lastly, we use $a \lesssim_{x} b$ to denote that $a \leq c_x b$ for some constant $c_x>0$ which depends only on $x$. 
\subsection{Weak convergence of Random Density Operators}\label{Sec:weakconv}
Let $(\Omega, \cA,\PP)$ be a sufficiently rich probability space on which all random  elements are defined. A sequence of random elements $(X_n)_{n \in \NN}$ taking values in a  topological space $\mathfrak{S}$  converges weakly to a random element $X$ (taking values in $\mathfrak{S}$) if $\EE[f(X_n)] \rightarrow \EE[f(X)]$ for all bounded continuous functions $f:\mathfrak{S} \rightarrow \RR$. This is denoted by $X_n\trightarrow{w} X$. Here, the random element of interest is a random density operator (or operators), which  is a Borel-measurable mapping from $\Omega$ to the space of density operators, $\cS_d$. The weak limit of a random density operator is unique if it exists (see e.g. \cite{AVDV-book}). Since density operators are trace-class, the natural ambient Banach space for weak convergence in the infinite-dimensional setting is the trace class $\mathfrak S_1(\HH)=\{A:\operatorname{Tr}|A|<\infty\}$, equipped with the trace norm $\|A\|_1=\operatorname{Tr}|A|$. In finite dimensions, we may take $\mathfrak{S}=\cL(\HH_d)$ equipped with any norm since all norms are equivalent.
\subsection{Fr\'{e}chet Differentiability}
\begin{definition}[Fr\'{e}chet differentiability, see e.g. \cite{Bhatia-book}]\label{Def:Frdiff}
For an open set $\cX \subseteq \RR$, let  $f: \cX \rightarrow \bar \RR $ and $A \in \cH_d$ with its spectrum  $\spec{A} \subseteq \cX$. Then, $f$ is called (Fr\'{e}chet) differentiable at $A$ if there exists a linear map $D[f(A)]:\cH_d \rightarrow \cH_d$  such that for all $H \in \cH_d$ such that $\spec{A+H} \subseteq \cX$,
\begin{align}
    \norm{f(A+H)-f(A)-D[f(A)](H)}=o(\norm{H}), \label{eq:frechetderdef}
\end{align}
where $\norm{\,\cdot\,}$ denotes the standard operator norm. 
Then, $D[f(A)]$ is called the (Fr\'{e}chet) derivative of $f$ at $A$ and $D[f(A)](H)$ is the directional derivative of $f$ at $A$ in the direction $H$.  The derivative of $f$ induces a map from $\cH_d$ into $\cL(\cH_d)$ given by $D[f]: A \rightarrow D[f(A)]$. If this map is also differentiable at $A$, then $f$ is said to be twice differentiable at $A$ with the corresponding second-order derivative given by a bilinear map $D^2[f(A)]:\cH_d \times \cH_d \rightarrow \cH_d$.
\end{definition}
If $f$ is differentiable at $A$, then
\begin{align}
    D[f(A)](H)=\frac{d}{dt} f(A+tH) \big|_{t=0},~\forall ~H \in \cH_d. \notag
\end{align}
The chain rule holds: the composition of two differentiable maps $f$ and $g$ is differentiable and $D [g \circ f(A)]=D[g\big(f(A)\big)]  D[f(A)]$. Also, we have the product rule: for two  differentiable maps $f$ and $g$ and  $h=fg$,
\begin{align}
    D[h(A)](H)=f(A)D[g(A)](H)+D[f(A)](H)g(A).\notag
\end{align}
Finally, we will frequently use that
\begin{align}
    D[A^{-1}](H)=-A^{-1}HA^{-1}. \label{eq:derinvop}
\end{align}
\subsection{Bochner Integrability}\label{Sec:Bochnerint}
We need the concept of Bochner-integrability \cite{Bochner-1933} in the proofs of our main results, which we briefly mention. 
Let $(\mathfrak{X},\Sigma,\mu)$ be a measure space and $\mathfrak{B}$ be a Banach space. A function $f:\mathfrak{X} \rightarrow \mathfrak{B}$ is said to be integrable (in the sense of Bochner)  if there exists a sequence of simple functions $g_n$ such that $g_n \rightarrow f$, $\mu$-a.e., and 
\begin{align}
    \lim_{n \rightarrow \infty}\int_{\mathfrak{X}} \norm{f-g_n}_{\mathfrak{B}} d \mu =0, \notag
\end{align}
where $\norm{\cdot}_{\mathfrak{B}}$ denotes the Banach space norm. A  Bochner-measurable function $f$ is integrable if and only if 
\begin{align}
\int_{\mathfrak{X}}\norm{f}_{\mathfrak{B}} d \mu <\infty.\label{eq:Bochnerintegcond}    
\end{align}
 Moreover, if $f$ is integrable, then 
\begin{align}
\norm{\int_{\mathfrak{X}}f d \mu}_{\mathfrak{B}} \leq  \int_{\mathfrak{X}}\norm{f}_{\mathfrak{B}} d \mu. \label{eq:Bochnerintegcor}
\end{align}
\subsection{Quantum Information Measures} \label{Sec:quantinfmeas}
The von Neumann entropy of a density operator $\rho \in \cS_d$  is
\begin{align}
    \mathsf{H}(\rho):=-\tr{\rho \log \rho},\notag
\end{align}
For density operators $\rho,\sigma \in \cS_d $, the quantum relative entropy \cite{Umegaki-62} is 
\begin{align}
\qrel{\rho}{\sigma} :=\begin{cases}
    \tr{\rho \big(\log \rho-\log \sigma\big)}, &\quad \mbox{if } \rho \ll \sigma,\\
    \infty, & \quad \mbox{otherwise}.
\end{cases} \notag 
\end{align}
From the above two definitions, it follows that 
\begin{align}
    \mathsf{H}(\rho)= \log d - \qrel{\rho}{\pi_d}, \label{eq:ent-qrent-rel}
\end{align}
where $\pi_d=I/d$ is the maximally mixed state. 
By some abuse of notation, the classical relative entropy or Kullback-Leibler (KL) divergence \cite{Kullback1951OnSufficiency} between two probability distributions $\mathsf{P},\mathsf{Q} $  on a discrete alphabet $\cX$ is 
\begin{align}
\qrel{\mathsf{P}}{\mathsf{Q}}:=\sum_{x \in \X} \mathsf{P}(x) \log \frac{\mathsf{P}(x)}{\mathsf{Q}(x)},\notag    
\end{align}
if $\mathsf{P} \ll \mathsf{Q}$, and $\infty$ otherwise.

For $\alpha \in (0,1) \cup (1,\infty)$, the Petz-R\'{e}nyi divergence \cite{Petz1985Quasi-entropiesAlgebra} between $\rho,\sigma \in \cS_d$ is 
\begin{align} 
\petzdiv{\rho}{\sigma}{\alpha}:=\begin{cases}
    \frac{1}{\alpha-1} \log \tr{\rho^{\alpha}\sigma^{1-\alpha}}, & \mbox{ if } \rho \ll \sigma \mbox{ or }  \rho \not\perp \sigma \mbox{ for } \alpha \in (0,1),\\
\infty, &\mbox{ otherwise}.    
\end{cases}\label{eq:Petz-Renyi-div}
\end{align}
$\petzdiv{\rho}{\sigma}{\alpha}$ satisfies the data-processing inequality for $\alpha \in (0,2]$. 
For $\alpha \in (0,1) \cup (1,\infty)$, the sandwiched R\'{e}nyi divergence \cite{muller2013quantum,wilde2014strong} between $\rho,\sigma \in \cS_d$ is
\begin{align}
    \sanddiv{\rho}{\sigma}{\alpha}:=\begin{cases} \frac{\alpha}{\alpha-1} \log \norm{\rho^{\frac 12} \sigma^{\frac{\bar \alpha}{\alpha}} \rho^{\frac 12}}_{\alpha}, & \mbox{ if } \rho \ll \sigma \mbox{ or }  \rho \not\perp \sigma \mbox{ for } \alpha \in (0,1),\\
\infty, &\mbox{ otherwise}.
    \end{cases} \label{eq:sand-Renyi-div}
\end{align}
$\sanddiv{\rho}{\sigma}{\alpha}$ satisfies data-processing inequality for $\alpha \geq 1/2$. Also, note that  $ \sanddiv{\rho}{\sigma}{1/2}=-\log F(\rho,\sigma)$, where  $F(\rho,\sigma):=\norm{\sqrt{\rho}\sqrt{\sigma}}_1^2=\big(\tr{\abs{\sqrt{\rho}\sqrt{\sigma}}}\big)^2$ denotes the fidelity \cite{Uhl76,Jozsa-1994} between $\rho$ and $\sigma$. The max-divergence between states $\rho$ and $\sigma$ is
\begin{align}
    \mathsf{D}_{\max}(\rho\|\sigma):=\lim_{\alpha \rightarrow \infty} \sanddiv{\rho}{\sigma}{\alpha}=\inf\big\{\lambda: \rho \leq e^{\lambda} \sigma \big\}. \label{eq:rendivinft} 
\end{align}
This divergence is the unique quantum generalization of the classical R\'{e}nyi divergence of infinite order that satisfies the data-processing inequality. 
For further details about the aforementioned information measures, see the books \cite{wilde2017quantum,tomamichel2015quantum}.

In the next section, we derive limit distribution for estimators of the aforementioned information measures.

\section{Main Results}\label{Sec:Mainres}
Let $\rho_n$ and $\sigma_n$ be random density operators such that $\rho_n \trightarrow{w} \rho$ and $\sigma_n \trightarrow{w} \sigma$ in \textit{trace norm} (see Section \ref{Sec:weakconv} for definitions). Since all norms are equivalent in finite dimensions, the choice of trace norm does not incur any loss of generality.  Further, let $(r_n)_{n \in \NN}$ denote a diverging positive sequence. In the following,  
 \textit{null} and \textit{alternative} refers to the scenarios $\rho=\sigma$ and    $\rho \neq \sigma$, respectively, while \textit{two-sample} signifies that both  $\rho$ and $\sigma$ are estimated. 

\subsection{Quantum Relative Entropy}

\begin{theorem}
[Limit distribution for quantum relative entropy]\label{Thm:quantrelent-limdist}
 Let   
 $\rho_n \ll \sigma_n  \ll \sigma$ and $\rho_n \ll \rho \ll  \sigma$.  The following hold: 
\begin{enumerate}[(i)]
 \item (Two-sample alternative)
 If $\big(r_n(\rho_n-\rho), r_n(\sigma_n-\sigma)\big) \trightarrow{w} (L_1,L_2)$, then
 \begin{align}
    r_n \big(\qrel{\rho_n}{\sigma_n}-\qrel{\rho}{\sigma}\big) \trightarrow{w}  \tr{L_1 (\log \rho-\log \sigma )-\rho D[\log \sigma](L_2)}.\label{eq:qrel-twosample-alt}
\end{align} 

\item (Two-sample null) 
If $\rho=\sigma$ and 
 $\big(r_n(\rho_n-\rho),r_n(\sigma_n-\rho)\big) \trightarrow{w} (L_1,L_2)$,  then 
    \begin{align}
    r_n^2 \qrel{\rho_n}{\sigma_n} &\trightarrow{w}   \tr{L_1 D[\log \rho](L_1-L_2)}  +\tr{\frac{\rho}{2}D^2[\log \rho](L_1-L_2,L_1-L_2)}.\label{eq:qrel-twosample-null}
\end{align}
\end{enumerate}
\end{theorem}

The proof of Theorem \ref{Thm:quantrelent-limdist} is presented in Section \ref{Thm:quantrelent-limdist-proof}. The key idea relies on applying an operator version of Taylor's theorem to the function,  $(x,y) \mapsto x (\log x-\log y)$, and showing that the  remainder terms (e.g., second and higher order terms in the alternative case) vanish under the conditions stated in the theorem. At a technical level, the arguments use uniform integrability (see Section \ref{Sec:Proof})  of the remainder terms (guaranteed under the assumptions) to justify interchange of limits, trace, and integrals. We note that the 
regularity conditions in Theorem \ref{Thm:quantrelent-limdist} are same as that of \cite[Theorem 1]{SGK-IT-2023} specialized to the discrete setting. 
Also, observe  that analogous to the classical case, the limits depend on whether $\rho=\sigma$ (null)  or $\rho \neq \sigma$ (alternative), and that the convergence rate is faster in the former.

\begin{remark}[One-sample null and alternative]
The one-sample case refers to the setting when  $\sigma_n=\rho$ (null) or $\sigma_n=\sigma$ (alternative) for all $n \in \NN$, i.e., when only $\rho$ is approximated by $\rho_n$. 
In this case, 
the respective limits can be obtained by letting $L_2=0$ in   \eqref{eq:qrel-twosample-alt} and \eqref{eq:qrel-twosample-null}.
\end{remark}
Simplified expressions for the limit variables in Theorem \ref{Thm:quantrelent-limdist} exist  when all relevant density operators commute, as stated in the following corollary.
\begin{cor}[Commutative case] \label{Cor:qrel-comm}
 If all operators in Theorem \ref{Thm:quantrelent-limdist} commute, then
\begin{align}
        r_n \big(\qrel{\rho_n}{\sigma_n}-\qrel{\rho}{\sigma}\big) \trightarrow{w}  \tr{L_1 (\log \rho-\log \sigma )-L_2\rho \sigma^{-1}}.\label{eq:qrel-twosample-alt-comm}
\end{align}
Additionally, when $\rho=\sigma$, then 
 \begin{align}
    r_n^2 \qrel{\rho_n}{\sigma_n} \trightarrow{w}   \frac 12  \tr{\big(L_1-L_2\big)^2 \rho^{-1}}.\label{eq:qrel-twosample-null-comm}
\end{align}
\end{cor}
Equations \eqref{eq:qrel-twosample-alt-comm} and \eqref{eq:qrel-twosample-null-comm} recovers \cite[Theorem 2]{SGK-IT-2023} specialized to the discrete setting with finite support, which is the classical analogue of Theorem \ref{Thm:quantrelent-limdist} pertaining to KL divergence. The RHS of \eqref{eq:qrel-twosample-null-comm} is reminiscent of the expression for $\chi^2$ divergence and can be interepreted as a weighted $L^2$ norm between the limits $L_1$ and $L_2$  (see e.g., \cite{Gao-Rouze}).

A class of divergences intermediate between the classical and quantum  relative entropy are the  measured relative entropies \cite{Donald1986,Hiai-Petz-1991,Bennett-1993,Rains-2001}, which equals the largest KL divergence between the output probability distributions induced by a set of measurements (quantum to classical channel) on two quantum states. In Appendix \ref{Sec:limdist-mrel}, we characterize the  distributional limits for an estimator of this quantity under a uniqueness assumption on the optimal measurement. 

Specializing Theorem \ref{Thm:quantrelent-limdist} to  von Neumann entropy leads to the following result.
\begin{cor}[Limit distribution for von Neumann entropy]
    \label{Thm:quant-vnent-limdist}
  Let $\rho_n \ll \rho$.  If $r_n(\rho_n-\rho) \trightarrow{w} L$, then
 \begin{align}
    r_n \big(\mathsf{H}(\rho_n)-\mathsf{H}(\rho)\big) \trightarrow{w} -\tr{L \log \rho}. \label{eq:vnent-onesample-alt}
\end{align}
\end{cor}
\begin{proof}
The claim follows from \eqref{eq:ent-qrent-rel} and \eqref{eq:qrel-twosample-alt} with $L_1=L$ and $L_2=0$ by noting that the regularity conditions in Part $(i)$ of Theorem \ref{Thm:quantrelent-limdist} are satisfied with $\sigma_n=\sigma=\pi_d$ for all $n \in \NN$.
\end{proof}
It is well-known that $\mathsf{H}(\rho)=\mathsf{H}(\lambda)$, where  $\lambda \in \cH_d$ denotes the  diagonal operator comprising of the eigenvalues of $\rho$ arranged in non-increasing order. In other words,  $\mathsf{H}(\rho)$ equals the Shannon entropy of the probability distribution composed of the eigenvalues of $\rho$. An unbiased estimator of the spectrum of a quantum state is given by the EYD algorithm \cite{Alicki-1988,Keyl-Werner-2001} that outputs a Young's diagram as its estimate. In \cite[Theorem 3.1]{Houdre-Xu-2013}, the limit distribution of this estimator with a scaling rate $n^{1/2}$ is characterized in terms of a $d$-dimensional Brownian functional $B=(B_1,\ldots, B_d)$. From  Corollary \ref{Thm:quant-vnent-limdist} and the discussion above, it then follows that the asymptotic distribution of the EYD algorithm based estimator of $\mathsf{H}(\rho)$ is governed by \eqref{eq:vnent-onesample-alt}  with $r_n=n^{1/2}$, $L=\mathrm{diag}(B)$ and $\rho=\lambda$, where $\mathrm{diag}(\cdot)$ denotes the operation of representing a vector as the diagonal elements of a matrix.

\subsection{Quantum R\'{e}nyi Divergence} 
In contrast to the classical case, there is no single definition of quantum  R\'{e}nyi divergence known that satisfies all natural properties desired of a quantum information measure for all $\alpha$ (see \cite{tomamichel2015quantum}). Among the infinite possibilities, the most  important ones with direct operational significance  are the Petz-R\'{e}nyi and sandwiched R\'{e}nyi divergences, which we consider below.

\noindent
\textbf{Petz-R\'{e}nyi divergence:} We first derive distributional limits for Petz-R\'{e}nyi divergence estimators.
\begin{theorem}[Limit distribution for Petz-R\'{e}nyi divergence]\label{Thm:Petzrelent-limdist}
Let $\alpha \in (0,1) \cup (1,2]$ and $\bar \alpha=1-\alpha$. Suppose   
 $\rho_n \ll \sigma_n  \ll \sigma$ and $\rho_n \ll \rho \ll  \sigma$.  
Then, the following hold:
\begin{enumerate}[(i)]

 \item  (Two-sample alternative)
If $\big(r_n(\rho_n-\rho), r_n(\sigma_n-\sigma)\big) \trightarrow{w} (L_1,L_2)$, then 
 \begin{align}
    r_n \big(\petzdiv{\rho_n}{\sigma_n}{\alpha}-\petzdiv{\rho}{\sigma}{\alpha}\big) \trightarrow{w}  \frac{\tr{\sigma^{\bar \alpha} D[\rho^{\alpha}](L_1)}+\tr{\rho^{\alpha}D[\sigma^{\bar \alpha}](L_2)}}{(\alpha-1)\tr{\rho^{\alpha}\sigma^{\bar \alpha}}}. \label{eq:petzrelent-twosample-alt}
\end{align}
  \item  (Two-sample null) If $\rho=\sigma$ and  $\big(r_n(\rho_n-\rho), r_n(\sigma_n-\rho)\big) \trightarrow{w} (L_1,L_2)$, then 
 \begin{align}
    r_n^2 \petzdiv{\rho_n}{\sigma_n}{\alpha} \trightarrow{w}  \frac{\tr{\rho^{\bar \alpha} D^2[\rho^{\alpha}](L_1,L_1)+\rho^{\alpha}D^2[\rho^{\bar \alpha}](L_2,L_2)+2 D[\rho^{\alpha}](L_1)D[\rho^{\bar \alpha}](L_2)}}{2(\alpha-1)}.\label{eq:petzrelent-twosample-null}
\end{align}
   
\end{enumerate}
\end{theorem}
The proof of Theorem \ref{Thm:Petzrelent-limdist} is given in Section \ref{Thm:Petzrelent-limdist-proof}, and utilizes a similar approach as Theorem \ref{Thm:quantrelent-limdist}. As  in the case of quantum relative entropy, the limits and the scaling rate differ in the null and alternative. Observe that we consider Petz-R\'{e}nyi divergence of order  less than two since it does not satisfy the data-processing inequality above this value. 
In the commutative case, the expressions in \eqref{eq:petzrelent-twosample-null} and \eqref{eq:petzrelent-twosample-alt} simplify further, and in particular, we deduce the  limit distributions for estimators of classical R\'{e}nyi divergences as stated below (see \cite{Pietrzak-2015} for a more general result when the dimension scales with $n$).  
\begin{cor}[Commutative case]
  If all operators in Theorem \ref{Thm:Petzrelent-limdist} commute, then 
\begin{align}
    r_n \big(\petzdiv{\rho_n}{\sigma_n}{\alpha}-\petzdiv{\rho}{\sigma}{\alpha}\mspace{-3 mu}\big) \trightarrow{w} \frac{\alpha\tr{L_1\sigma^{\bar \alpha} \rho^{-\bar \alpha}}+\bar \alpha \tr{L_2\rho^{\alpha}\sigma^{-\alpha}}}{(\alpha-1)\tr{\rho^{\alpha}\sigma^{\bar \alpha}}} . \label{eq:petzrelent-twosample-alt-comm}
\end{align}
Moreover, if  $\rho=\sigma$, then
\begin{align}
       r_n^2 \petzdiv{\rho_n}{\sigma_n}{\alpha} \trightarrow{w}  \frac{\alpha }{2}\tr{(L_1-L_2)^2\rho^{-1}}  .\label{eq:petzrelent-twosample-null-comm} 
\end{align}
\end{cor}
\noindent
\textbf{Sandwiched R\'{e}nyi divergence:}
We next consider the minimal or sandwiched R\'{e}nyi divergence of order $\alpha$.  
\begin{theorem}[Limit distribution for sandwiched R\'{e}nyi divergence]\label{Thm:sandrelent-limdist}
Let $\alpha \in [0.5,1) \cup (1,\infty)$ and  $\bar \alpha=1-\alpha$. Suppose   
 $\rho_n \ll \sigma_n  \ll \sigma$ and $\rho_n \ll \rho \ll  \sigma$.  
 If $\big(r_n(\rho_n-\rho), r_n(\sigma_n-\sigma)\big) \trightarrow{w} (L_1,L_2)$, then 
 \begin{align}
   & r_n \big(\sanddiv{\rho_n}{\sigma_n}{\alpha}-\sanddiv{\rho}{\sigma}{\alpha}\big) \notag \\
    &\trightarrow{w} \frac{\alpha}{\alpha-1} \frac{\tr{\Big(D[\rho^{\frac 12}](L_1) \sigma^{\frac{\bar \alpha}{\alpha}} \rho^{\frac 12}+ \rho^{\frac 12} \sigma^{\frac{\bar \alpha}{\alpha}} D[\rho^{\frac 12}](L_1)  +\rho^{\frac 12} D[\sigma^{\frac{\bar \alpha}{\alpha}}](L_2) \rho^{\frac 12}\Big) \big(\rho^{\frac 12} \sigma^{\frac{\bar \alpha}{\alpha}}\rho^{\frac 12} \big)^{\alpha-1}}}{\big\|\rho^{\frac 12 }\sigma^{\frac{\bar \alpha}{\alpha}}\rho^{\frac 12 }\big\|_{\alpha}^{\alpha}}.\label{eq:sandrelent-twosample-alt}
\end{align}
If all operators  above commute, then \eqref{eq:sandrelent-twosample-alt} simplifies to \eqref{eq:petzrelent-twosample-alt-comm}.
\end{theorem}
The proof of Theorem \ref{Thm:sandrelent-limdist} is given in Section \ref{Thm:sandrelent-limdist-proof}. Different from the approach used in the previous results, 
we compute this limit by recasting the term  $\|\rho^{ 1/2} \sigma^{\frac{\bar \alpha}{\alpha}} \rho^{ 1/2}\|_{\alpha}$ in \eqref{eq:sand-Renyi-div} as a maximization using dual expressions for Schatten norms.  However,  establishing the desired limit with the new expression is more involved on account of the additional maximization involved. To this end, we consider upper and lower bounds (without the maximization) and show that they coincide  with the expression in \eqref{eq:sandrelent-twosample-alt} asymptotically, thus establishing  the claim.    We mention here that the right-hand side (RHS) of \eqref{eq:sandrelent-twosample-alt} vanishes when $\rho=\sigma$, showing that correct scaling rate in the null setting for a non-degenerate limit should be $r_n^2$. However, the above technique does not lead to a proof of this claim due to non-matching upper and lower limits.

As a corollary of Theorem \ref{Thm:sandrelent-limdist}, we characterize the limit distributions for  fidelity and max-divergence.
\begin{cor}[Fidelity and max-divergence]\label{cor:limdistfidel}
Let $\rho_n \ll \sigma_n  \ll \sigma$ and $\rho_n \ll \rho \ll  \sigma$. If $\big(r_n(\rho_n-\rho), r_n(\sigma_n-\sigma)\big) \trightarrow{w} (L_1,L_2)$, then
  \begin{align}
   &  r_n \big(F(\rho_n,\sigma_n)-F(\rho,\sigma)\big) \notag \\
    &\trightarrow{w} \big(F(\rho,\sigma)\big)^{\frac 12}\tr{\Big(D[\rho^{\frac 12}](L_1) \sigma \rho^{\frac 12}+ \rho^{\frac 12} \sigma D[\rho^{\frac 12}](L_1)  +\rho^{\frac 12} L_2 \rho^{\frac 12}\Big) \big(\rho^{\frac 12} \sigma\rho^{\frac 12} \big)^{-\frac 12}}, \notag 
  \end{align}
  and 
  \begin{align}
& r_n \big(\mathsf{D}_{\max}(\rho_n\|\sigma_n)-\mathsf{D}_{\max}(\rho\|\sigma)\big) \notag \\
    &\trightarrow{w} e^{-\mathsf{D}_{\max}(\rho\|\sigma)} \lambda_{\max,\,P_{\max}} \left(D[\rho^{\frac 12}](L_1) \sigma^{-1} \rho^{\frac 12}+ \rho^{\frac 12} \sigma^{-1} D[\rho^{\frac 12}](L_1)  -\rho^{\frac 12} \sigma^{-1}L_2\sigma^{-1} \rho^{\frac 12} \right),\notag
\end{align}
where   $\lambda_{\max,P}(A)$ is the maximum eigenvalue of $PAP$ restricted to the support of $P$ and  $P_{\max}:=P_{\max}\big(\rho^{\frac 12}\sigma^{-1}\rho^{\frac 12}\big)$  denotes the  eigenprojection corresponding to  the maximal eigenvalue of $\rho^{\frac 12}\sigma^{-1}\rho^{\frac 12}$.
\end{cor}
\subsection{Generalization to Infinite-dimensional Quantum Systems}\label{Sec:infdim}
Here, we consider a generalization of Theorem \ref{Thm:quantrelent-limdist} to  infinite-dimensional quantum systems when the underlying Hilbert space $\HH$ is separable.  For a density operator $\rho\geq0$ with $P_\rho$ denoting its support projection, let
$\mathfrak{S}_{1,\rho^{-1}}$ be the space of operators $A$ satisfying
$A=A P_\rho$ for which the densely defined product $A\rho^{-1}$ admits a (necessarily unique) trace-class extension to $\HH$. We denote this extension again by $A\rho^{-1}$.  We equip this space with
\[
   \|A\|_{1,\rho^{-1}}:=\|A\rho^{-1}\|_1 .
\]
Then, $\big(\mathfrak{S}_{1,\rho^{-1}},\|\cdot\|_{1,\rho^{-1}}\big)$ is a separable Banach space. Equivalently,
$\mathfrak{S}_{1,\rho^{-1}}=\{X\rho:X\in\mathfrak{S}_1,\ X=XP_\rho\}$, and
$X\mapsto X\rho$ is an isometric isomorphism from the closed subspace
$\{X\in\mathfrak{S}_1:X=XP_\rho\}$ onto $\mathfrak{S}_{1,\rho^{-1}}$.
We use the corresponding norm topology when speaking of random elements and weak convergence in distribution. The product space
$\mathfrak{S}_{1,\rho^{-1}}\times\mathfrak{S}_{1,\sigma^{-1}}$ is equipped with the norm
$\|(A,B)\|:=\|A\|_{1,\rho^{-1}}\vee\|B\|_{1,\sigma^{-1}}$ and the weak convergence of random elements taking values in this space is denoted as $(A_n,B_n) \trightarrow{(w)} (A,B)$ in $\norm{\,\cdot\,}_{1,\rho^{-1}}\vee  \norm{\,\cdot\,}_{1,\sigma^{-1}}$.
Given the above, the following result (see Section \ref{Thm:quantrelent-limdist-inf-proof} for proof) provides sufficient conditions under which   limit distribution for quantum relative entropy  exists. 
\begin{theorem}
[Quantum relative entropy: Infinite dimensional case]\label{Thm:quantrelent-limdist-inf} 
 Let   
 $\rho_n \ll \sigma_n  \ll \sigma$ and $\rho_n \ll \rho \ll  \sigma$ be such that $\qrel{\rho_n}{\sigma_n}<\infty$ a.s. for all $n \in \NN$ and $\qrel{\rho}{\sigma}<\infty$. 
Then, the  following hold:
\begin{enumerate}[(i)]
   \item (Two-sample alternative)
 Suppose  there exists a constant $c$ satisfying  $\PP\big(\big\|\rho_n \sigma_n^{-1}\big\|_{\infty}>c\big) \rightarrow 0$. If  $\big(r_n(\rho_n-\rho), r_n(\sigma_n-\sigma)\big) \trightarrow{w} (L_1,L_2)$ in  $\norm{\,\cdot\,}_{1,\rho^{-1}}\vee  \norm{\,\cdot\,}_{1,\sigma^{-1}}$, then \eqref{eq:qrel-twosample-alt} holds.
   \item (Two-sample null)
Let $\rho=\sigma$. If 
 $\big(r_n(\rho_n-\rho),r_n(\sigma_n-\rho)\big) \trightarrow{w} (L_1,L_2)$ in  $\norm{\,\cdot\,}_{1,\rho^{-1}} \vee \norm{\,\cdot\,}_{1,\rho^{-1}}$, then
  \eqref{eq:qrel-twosample-null} holds.
\end{enumerate}

\end{theorem}
We briefly discuss the regularity assumptions in the above theorem. In the infinite dimensional case, $\qrel{\rho}{\sigma}$ can be unbounded even if the support conditions $\rho \ll \sigma$ are satisfied. This necessitates the finiteness assumption on the quantum relative entropies above. 
The requirement $\PP\big(\big\|\rho_n \sigma_n^{-1}\big\|_{\infty}>c\big) \rightarrow 0$ imposes a stochastic boundedness assumption on the operator $\rho_n \sigma_n^{-1}$, which  reduces to a probabilistic uniform boundedness assumption on the likelihood ratio in the commuting case.  In the proof,  this regularity condition suffices to establish the Bochner-integrable domination of the terms appearing in the Taylor's expansion of the quantum relative entropy. 
\section{Application}\label{Sec:applic}
Limit theorems for classical divergences have several applications in statistics, computational science and biology such as constructing confidence intervals for hypothesis testing\cite{Okano-2024},  auditing of differential privacy\cite{SGK-IT-2023}, and biological data analysis \cite{Pietrzak-2015}. Here, we consider an application of Theorem \ref{Thm:quantrelent-limdist} in establishing performance guarantees for the problem of testing for the quantum relative entropy between unknown states\footnote{Since the states under the hypotheses are unknown, the problem is different from  testing with known states and computing the quantum relative entropy to decide on the hypotheses.}.   
The relevant multi-hypothesis testing problem  can be  formulated as\footnote{Note that here $(\rho_i,\sigma_i)$, $1 \leq i \leq m $, denote pairs of quantum states and not random density operators as was used until now.}
  \begin{align}
    H_{i}&: \epsilon_{i} < \qrel{\rho_i}{\sigma_i} 
 \leq  \epsilon_{i+1}, \label{eq:hyptest} 
\end{align}   
where $\epsilon_i \geq 0$ satisfy $\epsilon_{i+1} > \epsilon_i$ for $i \in \cI = \{0,\ldots,m-1\} $.  
 We are interested in the setting where approximately $nd^2$ identical copies of the unknown states are available for the tester. The goal then is to design a test  $\cT_n=\{M_i^{(n)}, i \in \cI\}$ with $M_i^{(n)} \geq 0$ for all $i$, and $\sum_{i \in \cI} M_i^{(n)}=I$ that achieves a specified performance, i.e., an $m$-outcome positive operator-valued measure (POVM)  with index set $\cI$ (see Appendix \ref{Sec:limdist-mrel} for further details). Denoting the original hypothesis by $H$ and the test outcome by $\hat H$, the performance of $\cT_n$ is quantified by the  error probabilities
 \begin{align}
    \alpha_{i,n}(\cT_n,\rho_i^{\otimes n},\sigma_i^{\otimes n}) &:= \PP(\hat H \neq i|H=i)=\operatorname{Tr}\Bigg[(\rho_i^{\otimes n}\otimes \sigma_i^{\otimes n})\sum_{j \neq i}M_j^{(n)}\Bigg]. \notag
 \end{align}
 A test $\cT_n$ is said to achieve level $\tau$ if $\alpha_{i,n}(\cT_n,\rho_i^{\otimes n},\sigma_i^{\otimes n}) \leq \tau$ for every $i \in \cI$. A sequence of tests $(\cT_n)_{n \in \NN}$ is asymptotically said to achieve level $\tau$ if $\limsup_{n \rightarrow \infty} \alpha_{i,n}(\cT_n,\rho_i^{\otimes n},\sigma_i^{\otimes n}) \leq \tau$ for every $i \in \cI$.

A pertinent approach to realize a hypothesis test  is to first perform tomography of the states to obtain estimates, $\hat \rho_n, \hat \sigma_n$,  
and then compute the relative entropy between them.  
A standard class of tests (motivated from the Neyman-Pearson theorem) then decides in favor of $H_i$ if $t_{i,n} < \qrel{\hat \rho_n}{\hat \sigma_n} \leq t_{i+1,n}$, where $t_{i,n}$ for $0 \leq i \leq m$ are critical values chosen according to the desired level $\tau_i \in (0,1]$ for $i^{th}$ error probability. Each such test (statistic) $T_n$ induces a POVM indexed by $\cI$, denoted by $\cT_n^{\mathrm{tom}}\big(\{t_i\}_{i \in \cI}\big)$, for which we will use the shorthand $\cT_n^{\mathrm{tom}}$. Let $\alpha_{i,n}\big(\cT_n^{\mathrm{tom}},\rho_i^{\otimes n},\sigma_i^{\otimes n}\big)=\PP(T_n \neq i| H= i)$ denote the error probability for the test statistic $T_n$ given the $i^{th}$ hypothesis is true.

 To obtain concrete performance guarantees for the aforementioned  hypothesis test, 
 we  consider a specific tomographic estimator for density operators based on Pauli measurements.  This can be considered as a quantum analogue of the classical plug-in estimator based on empirical probability distributions. The choice of Pauli measurements is mainly due to simplicity of presentation, and our approach will extend to other tomographic schemes that rely on estimating the coefficients in an operator basis expansion using measurements on independent copies of quantum states. In the following, we first describe the estimator and characterize its limiting distribution, which will then be used to construct the test statistic for \eqref{eq:hyptest}. 
\subsection{Tomographic Estimator of Quantum States}
Let $d=2^N$ for some integer $N$, and $\{\gamma_j\}_{j=0}^{d^2-1}$ denote the set of  
multi-qubit ($N$-qubit)  Pauli operators constructed as the $N$-fold tensor product of standard Pauli operators acting on a qubit. Specifically,  $ \gamma_j=\bigotimes_{i=1}^N\gamma_{j,i}$ with $\gamma_{j,i} \in \{R_k\}_{k=0}^3$, where $\{R_k\}_{k=0}^3$ denotes the  single-qubit Pauli basis with the following representations in the standard basis: 
 \begin{align}
R_0=\Bigg[ {\begin{array}{cc}
   1 & 0 \\[2 pt]
   0 & 1 \\
  \end{array} } \Bigg], \quad  R_1=\Bigg[ {\begin{array}{cc}
   0 & 1 \\[2 pt]
   1 & 0 \\
  \end{array} } \Bigg], \quad  R_2=\Bigg[ {\begin{array}{cc}
   0 & -i \\[2 pt]
   i & 0 \\
  \end{array} } \Bigg], \quad  R_3=\Bigg[ {\begin{array}{cc}
   1 & 0 \\[2 pt]
   0 & -1 \\
  \end{array} }  \Bigg]. \notag
 \end{align}
 We may take $\gamma_0=I$.  
The multi-qubit Pauli operators are Hermitian and form an orthogonal operator basis for the real vector space $\cH_d$ with respect to the Hilbert-Schmidt inner product. Consequently, any multi-qubit density operator $\rho$ can be written as
\begin{align}
   \rho=\frac 1d \left(I+\sum_{j=1}^{d^2-1} s_{j}(\rho) \gamma_j\right), \label{eq:exppaulibasis} 
\end{align}
 with $s_{j}(\rho)=\tr{\rho \gamma_j}$. Note that  $\gamma_j$, for $1 \leq j \leq d^2-1$, are traceless  and have eigenvalues $\pm 1$. Moreover, for any $\bm{s}=(s_1,\ldots,s_{d^2-1}) \in \RR^{d^2-1}$, the operator  $\frac 1d \left(I+\sum_{j=1}^{d^2-1} s_{j} \gamma_j\right)$ is Hermitian with unit trace; hence, is a valid density operator provided $\bm{s}$ is such that  positive semi-definiteness holds.   
 Let $P_{j}^+$ and  $P_{j}^-$ denote the projections onto the eigenspace of $\gamma_j$ corresponding to the eigenvalue $+1$  and  $-1$ , respectively. 
 Then
 \begin{align}
   s_{j}(\rho)=s_{j}^+(\rho)-s_{j}^-(\rho), \notag  
 \end{align}
 where $s_{j}^+(\rho):=\tr{\rho P_{j}^+}$ and $s_{j}^-(\rho):=\tr{\rho P_{j}^-}=1-s_{j}^+(\rho)$.

Assume that identical copies of $\rho$ and $\sigma$ are available as desired, on which measurements using Pauli operators can be performed and the outcomes recorded. Let $O_{k}(j,\rho)$ denote the $k^{th}$ measurement outcome using $\gamma_j$ on $\rho$. 
Denoting by $\ind_{\cA}$ the indicator of the event $\cA$ and by $P_{\cS_d}$ the  projection (in the sense of Hilbert projection theorem) onto the closed convex set $\cS_d$, a  tomographic estimator of $\rho$ and $\sigma$  is then given by
\begin{subequations}
 \begin{align}
    \hat \rho_n= \ind_{\bar \rho_n \geq 0} \mspace{2 mu} \bar \rho_n + \ind_{\bar \rho_n \ngeq 0 } \mspace{2 mu} P_{\cS_d} (\bar \rho_n) , \label{eq:tomogestrho} 
        \end{align}
 \mbox{and }  \\      
    \begin{align}
    \hat \sigma_n= \frac{I}{nd} +\left(1-\frac 1n \right) \left(\ind_{\bar \sigma_n \geq 0} \mspace{2 mu} \bar \sigma_n +  \ind_{\bar \sigma_n \ngeq 0 } \mspace{2 mu} P_{\cS_d} (\bar \sigma_n)\right), \label{eq:tomogestsigma}
\end{align}   
\end{subequations}
respectively, where 
\begin{align}
  \bar \rho_n &:=  \frac {1}{d} \left(I+ \sum_{j=1}^{d^2-1} \hat s_{j}^{(n)}(\rho) \gamma_j\right), \notag \\
  \hat s^{(n)}_{j}(\rho) &:= \frac 1n \sum_{k=1}^n \ind_{O_{k}(j,\rho) =+1}-\ind_{O_{k}(j,\rho) =-1 },  \quad j \neq 0, \notag \\
  \hat{\mathbf{s}}^{(n)}(\rho)&:=\big(\hat s^{(n)}_1(\rho),\ldots, \hat s^{(n)}_{d^2-1}(\rho)\big), \notag
\end{align}
and $\bar \sigma_n$, $ \hat s^{(n)}_{j}(\sigma)$, and $\hat{\mathbf{s}}^{(n)}(\sigma)$ are defined analogously with $\rho$ replaced by $\sigma$ in the above expressions.  
 It follows from the above discussion that $\hat \rho_n,\hat \sigma_n \in \cS_d$ for all $n \in \NN$. Note that the extra term (negligible asymptotically) $ I/nd$ ensures that $\hat \sigma_n>0$ so that   $\hat \rho_n \ll \hat \sigma_n$  and $\qrel{\hat \rho_n}{\hat \sigma_n}$ is finite. Also, observe that to construct $\hat \rho_n$ and $\hat \sigma_n$, we need $n (d^2-1)$ independent copies, each of $\rho$ and $\sigma$, available for measurement.

Let $N(c,v^2)$ denote the one-dimensional normal distribution with mean $c$ and variance $v^2$. The following result shows that the limit distribution for estimators of quantum relative entropy based on Pauli tomography is Gaussian.
\begin{prop}[Limit distribution for tomographic estimator] \label{prop:limdisttomest}
Let $\rho,\sigma>0$. 
Then
\begin{subequations}
 \begin{align}
   & \sqrt{n} \big(\qrel{\hat \rho_n}{\sigma}-\qrel{\rho}{\sigma}\big) \trightarrow{w}  W_1 \sim N\big(0,v_1^2(\rho,\sigma)\big) , \label{eq:qreltom-onesample-alt} \\
   &  \sqrt{n} \big(\qrel{\hat \rho_n}{\hat \sigma_n}-\qrel{\rho}{\sigma}\big) \trightarrow{w}  W_2 \sim N\big(0,v_2^2(\rho, \sigma)\big), \label{eq:qreltom-twosample-alt}
    \end{align}
    \end{subequations}
    where 
        \begin{align}
    v_1^2(\rho,\sigma)&:=  \sum_{j=1}^{d^2-1}\frac{4s_{j}^+(\rho)s_{j}^-(\rho)}{d^2}\tr{\gamma_j(\log \rho-\log \sigma )}^2, \notag \\
v_2^2(\rho,\sigma)&:=v_1^2(\rho,\sigma)+\sum_{j=1}^{d^2-1} \frac{4s_{j}^+(\sigma)s_{j}^-(\sigma)}{d^2}\tr{\rho D[\log \sigma](\gamma_j) }^2. \notag
    \end{align}
\end{prop}
The proof of Proposition \ref{prop:limdisttomest} is given in Section \ref{Sec:prop:limdisttomest-proof} and follows by an application of Theorem \ref{Thm:quantrelent-limdist}. The main ingredient of the proof is to show that $\big(\sqrt{n}(\hat \rho_n-\rho),\sqrt{n}(\hat \sigma_n-\sigma)\big) \trightarrow{w} (L_{\rho},L_{\sigma}) $, where  $L_{\rho}:=\sum_{j=1}^{d^2-1}  \gamma_j Z_{j}(\rho)$ and $ Z_{j}(\rho) \sim N\big(0,4s_{j}^+(\rho)s_{j}^-(\rho)/d^2\big) $. The claim then follows from \eqref{eq:qrel-twosample-alt} by  noting that all relevant regularity conditions are satisfied. 
\subsection{Performance Guarantees for  Multi-hypothesis Testing}
 For simplicity of presentation, we will 
assume that $\sigma_i=\sigma$ for all $i \in \cI$ with $\sigma$ known for the test in \eqref{eq:hyptest}. Such a scenario arises, for instance, 
when testing for  the mixedness of an unknown state $\rho$ with respect to the maximally mixed state, $\sigma=\pi_d$. 
Also, for $\tau \in [0,1]$, let 
\begin{align}
    Q^{-1}(\tau)=\inf\left\{z\in \RR:(2 \pi)^{-1/2} \int_{z}^{\infty} e^{-u^2/2}du \leq \tau\right\}, \notag
\end{align}
be the inverse complimentary cumulative distribution function of the standard normal distribution $N(0,1)$. 
The following proposition provides a test statistic for the multi-hypothesis testing problem in \eqref{eq:hyptest} by utilizing the limit distribution for quantum relative entropy and characterizes its error probabilities. 

\begin{prop}[Performance of multi-hypothesis testing]\label{prop:HTperf} 
 Let $\tau \in (0,1]$, and $\rho_i,\sigma>0$ for $i\in\cI$ 
 satisfy the hypotheses in \eqref{eq:hyptest} for $\sigma_i=\sigma$ therein. Let $\hat D_n=\qrel{\hat \rho_n}{\sigma}$ with $\hat \rho_n$ given in \eqref{eq:tomogestrho}. Then, 
 the test statistic 
 \begin{align}
   T_n=\sum_{i \in \cI} i\ind_{\hat D_n \in \cL_{i,n}(c)} \mbox{ with } \cL_{i,n}(c):=(\epsilon_i+c n^{-1/2}, \epsilon_{i+1}+cn^{-1/2}),  \notag
 \end{align}
   asymptotically achieves a level $\tau$ provided 
   \begin{align}
      c\geq2d(Q^{-1}(\tau)\vee0)|\log b|, \notag
   \end{align}
where $b$ denotes the minimum of the 
eigenvalues of $\rho_i$ and $\sigma$ over all $i\in\cI$.
\end{prop}

The proof of Proposition \ref{prop:HTperf} (see  Section \ref{prop:HTperf-proof}) follows by an application of Proposition \ref{prop:limdisttomest} and Portmanteau theorem \cite[Theorem 2.1]{Billingsley-99}. The  threshold $c$ achieving a desired asymptotic level $\tau$ is determined by utilizing the knowledge that $\sqrt{n}\big(\hat D_n-\qrel{\rho_i}{\sigma}) $ converges in distribution to a centered normal under hypothesis $i$, whose   variance $v_1^2(\rho_i,\sigma)$ can be uniformly bounded for $\rho_i,\sigma $ with $i \in \cI$.
\section{Concluding Remarks} \label{Sec:conclusion}
    This paper studied limit distributions for a certain class of estimators of important quantum divergences such as quantum relative entropy,  its R\'{e}nyi generalizations, and measured relative entropy. Taking recourse to an operator version of Taylor's theorem, the limit distributions are characterized in terms of trace functionals of first or second-order Fr\'{e}chet derivatives of elementary functions. These functions simplify in the commuting case and coincide with previously known expressions in the classical case. We  employed the derived results to show that the asymptotic distribution of an estimator of quantum relative entropy based on Pauli tomography of states is normal. We  then utilized this knowledge to propose a test statistic for a multi-hypothesis testing problem and characterized its asymptotic performance.

   Looking forward, 
    several  open questions remain. One pertinent question concerns the rate of convergence of the empirical distribution of the divergence to its limit in the flavor of classical  Berry-Esseen theorem. Any progress in this direction would be extremely useful to understand the non-asymptotic behavior of such estimators.  Also, accounting  for the case of infinite dimensional quantum systems would be a natural extension. In Section \ref{Sec:infdim}, we  treated the case of quantum relative entropy for density operators  on a separable Hilbert space. We believe that our approach can be extended to more general scenarios with appropriate technical modifications to ensure uniform integrability of  terms that appear in a Taylor's expansion. However, this is beyond the scope of the current article.   Of interest further is  to understand the asymptotic and non-asymptotic behavior of other classes of estimators such as those based on variational methods, for which the techniques used here may not be directly applicable.  Lastly, it would also be beneficial to study the statistical behaviour of estimators of other quantum divergences not considered here such as quantum $\chi^2$ divergence \cite{Temme-2010} and geometric R\'{e}nyi divergence \cite{Matsu-2018,Fang2019Geometric}.  

\section*{Acknowledgement}  
SS and MB acknowledges support from the Excellence Cluster - Matter and Light for Quantum Computing (ML4Q).  MB acknowledges funding from the European Research Council (ERC Grant Agreement No. 948139).
\section{Proofs}\label{Sec:Proof}
The following technical lemma will be handy for our purposes. Its proof is  given in Appendix \ref{sec:app-lem:multprojtr-proof}.
\begin{lemma}[Properties of trace-class self-adjoint operators] \label{lem:multprojtr}
Let $\HH$ be a separable Hilbert space. Then, the following hold:
\begin{enumerate}[(i)]
\item Suppose $A$ and $B$ are self-adjoint (Hermitian) operators such that  $AB$ is  trace-class. Let $P$ be an orthogonal projection (i.e., $0 \leq P=P^2$)  satisfying $A \ll P$. Then, $\tr{AB}=\tr{PAPBP}$.  
    \item Let $A,B,C$ be trace-class self-adjoint operators such that $B \leq A \leq C$. Then, $\norm{A}_1 \leq \norm{B}_1+\norm{C}_1$.
\end{enumerate}
\end{lemma}
\noindent We next proceed with the proofs of the main results.
\subsection{Proof of Theorem \ref{Thm:quantrelent-limdist}} \label{Thm:quantrelent-limdist-proof}
For $A_1,A_2 >0$, let 
\begin{align}
  f(A_1,A_2)=A_1(\log A_1-\log A_2). \notag   
\end{align}
 Consider the integral representation 
\begin{align}
   \log A&=\int_{0}^{\infty} \left( \frac{1}{(\tau+1)I}-\frac{1}{\tau I+A}\right) d\tau \notag \\
   &=\int_{0}^{1} \left( \frac{1}{(\tau+1)I}-\frac{1}{\tau I+A}\right) d\tau  +\int_{1}^{\infty} \left( \frac{1}{(\tau+1)I}-\frac{1}{\tau I+A}\right) d\tau, \label{eq:logint}
\end{align}
for $A > 0$. We have $D\big[(\tau I+A)^{-1}\big](H)=-(\tau I+A)^{-1}H(\tau I+A)^{-1}$ by \eqref{eq:derinvop}  
and
\begin{align}
    D[\log A](H)= \int_{0}^{\infty} (\tau I+A)^{-1}H(\tau I+A)^{-1}  d\tau.  \notag
\end{align}
Applying the above, we obtain via the chain rule and product rule for Fr\'{e}chet derivatives that
\begin{align*}
&D^{(1,0)}[f(A_1,A_2)](H)=A_1 \int_{0}^{\infty} \left(\tau I+A_1\right)^{-1} H \left(\tau I+A_1\right)^{-1} d\tau+ H\left(\log A_1-\log A_2\right),  \\
&D^{(0,1)}[f(A_1,A_2)](H)=-A_1 \int_{0}^{\infty} \left(\tau I+A_2\right)^{-1} H \left(\tau I+A_2\right)^{-1} d\tau, \\
&D^{(2,0)}[f(A_1,A_2)](H_1,H_2)= H_1 \int_{0}^{\infty} \left(\tau I+A_1\right)^{-1} H_2 \left(\tau I+A_1\right)^{-1} d\tau  \\
& \qquad \qquad\qquad \qquad \qquad\qquad \quad +H_2 \int_{0}^{\infty} \left(\tau I+A_1\right)^{-1} H_1 \left(\tau I+A_1\right)^{-1} d\tau  \\
& \qquad \qquad\qquad \qquad \qquad\qquad \quad-A_1 \int_{0}^{\infty} \left(\tau I+A_1\right)^{-1} H_1 \left(\tau I+A_1\right)^{-1} H_2  \left(\tau I+A_1\right)^{-1} d\tau  \\
& \qquad \qquad\qquad \qquad \qquad\qquad \quad-A_1 \int_{0}^{\infty} \left(\tau I+A_1\right)^{-1} H_2\left(\tau I+A_1\right)^{-1} H_1  \left(\tau I+A_1\right)^{-1}   d\tau,  \\
&D^{(0,2)}[f(A_1,A_2)](H_1,H_2)= A_1 \int_{0}^{\infty} \left(\tau I+A_2\right)^{-1} H_1 \left(\tau I+A_2\right)^{-1} H_2 \left(\tau I+A_2\right)^{-1}d\tau  \\
& \qquad \qquad\qquad \qquad \qquad\qquad \quad  +A_1 \int_{0}^{\infty} \left(\tau I+A_2\right)^{-1} H_2 \left(\tau I+A_2\right)^{-1} H_1 \left(\tau I+A_2\right)^{-1}d\tau,  \\
&D^{(1,1_+)}[f(A_1,A_2)](H_1,H_2)=- H_1 \int_{0}^{\infty} \left(\tau I+A_2\right)^{-1} H_2 \left(\tau I+A_2\right)^{-1} d\tau,  \\
&D^{(1_+,1)}[f(A_1,A_2)](H_1,H_2)=- H_2 \int_{0}^{\infty} \left(\tau I+A_2\right)^{-1} H_1 \left(\tau I+A_2\right)^{-1} d\tau, 
\end{align*}
    where the notation $D^{(1,1_+)}$ means that the order of differentiation is with respect to first coordinate followed by the second, and vice versa for $D^{(1_+,1)}$.

 Note that $f(A_1,A_2)$ is continuously twice differentiable function from $\cP^+_d \times \cP^+_d$ to $\cL(\HH_d)$. Hence, applying the operator version of multivariate Taylor's theorem (see e.g. \cite{Bhatia-book}), we obtain for $A_1,A_2,$ $B_1,B_2>0$ that
\begin{align}
        f(B_1,B_2)&=f(A_1,A_2)+D^{(1,0)}[f(A_1,A_2)](B_1-A_1)+D^{(0,1)}[f(A_1,A_2)](B_2-A_2)\notag \\
       & \quad  +\int_{0}^1 (1-t) D^{(2,0)}[f((1-t)A_1+tB_1,(1-t)A_2+tB_2)](B_1-A_1,B_1-A_1)\, dt \notag \\
       & \quad  +\int_{0}^1 (1-t) D^{(0,2)}[f((1-t)A_1+tB_1,(1-t)A_2+tB_2)](B_2-A_2,B_2-A_2)\, dt \notag \\
       & \quad  +\int_{0}^1 (1-t) D^{(1,1_+)}[f((1-t)A_1+tB_1,(1-t)A_2+tB_2)](B_1-A_1,B_2-A_2)\, dt \notag \\
       & \quad  +\int_{0}^1 (1-t) D^{(1_+,1)}[f((1-t)A_1+tB_1,(1-t)A_2+tB_2)](B_2-A_2,B_1-A_1)\, dt.\label{eq:taylor-op}
    \end{align}
 \noindent
\textbf{Two-sample null:} Assume first that $\rho_n,\sigma_n,\rho>0$. Setting $B_1=\rho_n$, $B_2=\sigma_n$ and $A_1=A_2=\rho$ in   \eqref{eq:taylor-op}, and defining 
\begin{align}
v(\rho_n,\rho,\tau,t):=\big(\tau I+(1-t)\rho+t\rho_n\big)^{-1}, \notag 
\end{align}
 we have
 \begin{align}
 &    f(\rho_n,\sigma_n) \notag \\
 &=\rho \int_{0}^{\infty} \left(\tau I+\rho\right)^{-1} (\rho_n-\rho) \left(\tau I+\rho\right)^{-1} d\tau -\rho \int_{0}^{\infty} \left(\tau I+\rho\right)^{-1} (\sigma_n-\rho) \left(\tau I+\rho\right)^{-1} d\tau \notag \\
     & \quad +2\int_{0}^1 (1-t)  (\rho_n-\rho) \int_{0}^{\infty} v(\rho_n,\rho,\tau,t) (\rho_n-\rho) v(\rho_n,\rho,\tau,t) d\tau dt  \notag \\
&  \quad  -2\int_{0}^1 (1-t) \left((1-t) \rho+t\rho_n \right) \int_{0}^{\infty} v(\rho_n,\rho,\tau,t) (\rho_n-\rho)  v(\rho_n,\rho,\tau,t) (\rho_n-\rho)  v(\rho_n,\rho,\tau,t) d\tau  dt \notag \\
&  \quad +2\int_{0}^1 (1-t)\big((1-t)\rho+t\rho_n\big) \int_{0}^{\infty} v(\sigma_n,\rho,\tau,t) (\sigma_n-\rho)  v(\sigma_n,\rho,\tau,t) (\sigma_n-\rho) v(\sigma_n,\rho,\tau,t)d\tau dt \notag \\
&\quad-2 \int_{0}^1 (1-t) (\rho_n-\rho) \int_{0}^{\infty} v(\sigma_n,\rho,\tau,t) (\sigma_n-\rho)  v(\sigma_n,\rho,\tau,t) d\tau dt. \label{eq:Tayexp-twosamp-null}
 \end{align}
 To extend the validity of the above equation to $0 \leq \rho_n \ll \sigma_n \geq 0$, we consider $f\big(\tilde \rho_n(\epsilon),\tilde \sigma_n(\epsilon)\big)$ with 
 \begin{align}
   &  \tilde \rho_n(\epsilon):=\rho_n+\epsilon \pi_d, \notag \\
    & \tilde \sigma_n(\epsilon):=\sigma_n+\epsilon \pi_d, \notag 
 \end{align}
 for
 $\epsilon>0$,  and   take limit $\epsilon \rightarrow 0$. Then, the desired expression  follows since $f\big(\tilde \rho_n(\epsilon),\tilde \sigma_n(\epsilon)\big)$ is continuous in $\epsilon$ for $\rho_n \ll \sigma_n$, and  uniform integrability conditions which allows for interchange of limits and integral hold for $0 <\epsilon \leq 1$. We illustrate the latter condition for some of the terms above. 
By using H\"{o}lder's inequality for Schatten-norms  and $\norm{A}_p \geq \norm{A}_q$ for a linear operator $A$ and $1 \leq p \leq q \leq \infty$, we have 
 \begin{align}
 \norm{\left(\tau I+\rho\right)^{-1} (\tilde \rho_n(\epsilon)-\rho) \left(\tau I+\rho\right)^{-1} }_1 &\leq \norm{\left(\tau I+\rho\right)^{-1}}_{\infty} \norm{(\tilde \rho_n(\epsilon)-\rho)\left(\tau I+\rho\right)^{-1}}_1 \notag \\
 & \leq \norm{\left(\tau I+\rho\right)^{-1}}_{\infty}^2 \norm{(\tilde \rho_n(\epsilon)-\rho)}_1. \label{eq:applyHolderineq}
 \end{align}  
Hence, for all $0 <\epsilon \leq 1 $,
\begin{align}
  &\int_{0}^{\infty} \norm{\left(\tau I+\rho\right)^{-1} (\tilde \rho_n(\epsilon)-\rho) \left(\tau I+\rho\right)^{-1} }_1 d\tau \leq \norm{\tilde \rho_n(\epsilon)-\rho}_1 \int_{0}^{\infty} \norm{\left(\tau I+\rho\right)^{-1}}_{\infty}^2  d\tau \notag \\
    & \qquad \qquad\qquad \qquad\qquad \qquad\qquad \qquad\qquad ~\quad \lesssim_{\rho}  \norm{\rho_n-\rho}_1+1, \notag     
\end{align}
where the final inequality follows because of the finiteness of the integral on account of $\rho >0$. 
Then, using \eqref{eq:Bochnerintegcond} and  \eqref{eq:Bochnerintegcor}, we obtain  
\begin{align}
  & \norm{\int_{0}^{\infty} \left(\tau I+\rho\right)^{-1} (\tilde \rho_n(\epsilon)-\rho) \left(\tau I+\rho\right)^{-1} }_1 d\tau \leq \norm{\tilde \rho_n(\epsilon)-\rho}_1 \int_{0}^{\infty} \norm{\left(\tau I+\rho\right)^{-1}d\tau}_{\infty}^2   \notag \\
    & \qquad \qquad\qquad \qquad\qquad \qquad\qquad \qquad\qquad ~\quad \lesssim_{\rho}  \norm{\rho_n-\rho}_1+1. \notag    
\end{align}
Thus, the LHS is uniformly integrable independent of $\epsilon$.

Similarly, 
    \begin{align}
   &\norm{(1-t) (\tilde \rho_n(\epsilon)-\rho) \int_{0}^{\infty} v\big(\tilde \rho_n(\epsilon),\rho,\tau,t\big) (\tilde \rho_n(\epsilon) -\rho) v\big(\tilde \rho_n(\epsilon),\rho,\tau,t\big) d\tau}_1 \notag \\
& \qquad \leq (1-t)\norm{  \tilde \rho_n(\epsilon) -\rho}_1^2 \int_{0}^{\infty} \norm{(\tau I+(1-t)\rho)^{-1}}_{\infty}^2 d\tau \notag\\
&\qquad \lesssim_{\rho}(1-t)\norm{ \tilde \rho_n(\epsilon) -\rho}_1^2 (1-t)^{-1}
\notag \\
&\qquad \lesssim  \norm{ \rho_n-\rho}_1^2+1, \notag \\
&\big\|(1-t) \big((1-t) \rho+t\tilde \rho_n(\epsilon) \big) \int_{0}^{\infty} v\big(\tilde \rho_n(\epsilon),\rho,\tau,t\big) (\tilde \rho_n(\epsilon)-\rho) v\big(\tilde \rho_n(\epsilon),\rho,\tau,t\big) \notag \\
 & \qquad \qquad  \qquad \qquad  \qquad \qquad \qquad \qquad \qquad \qquad \qquad  (\tilde \rho_n(\epsilon)-\rho)v\big(\tilde \rho_n(\epsilon),\rho,\tau,t\big) d\tau\big\|_1 \notag \\
 & \qquad \leq  (1-t)\norm{\tilde \rho_n(\epsilon) -\rho }_1^2 \int_{0}^{\infty} \norm{(\tau I+(1-t)\rho)^{-1}}_{\infty}^2 d\tau \notag \\
 &\qquad 
 \lesssim_{\rho} \norm{ \rho_n-\rho}_1^2+1. \notag
\end{align}
The uniform integrability for the remaining terms also follows  via analogous steps.

 Let $g_n:=   r_n^2 \qrel{\rho_n}{\sigma_n}$. 
 Multiplying  by $r_n^2$ and taking trace in \eqref{eq:Tayexp-twosamp-null}, we obtain
\begin{align}
g_n& =  2\int_{0}^1 \mspace{-2 mu} (1-t)  \operatorname{Tr} \bigg[r_n(\rho_n-\rho) \int_{0}^{\infty} \mspace{-4 mu}v(\rho_n,\rho,\tau,t) r_n(\rho_n-\rho) v(\rho_n,\rho,\tau,t)  d\tau\bigg] dt \notag \\
& \quad    -2\int_{0}^1 (1-t) \operatorname{Tr} \bigg[\left((1-t) \rho+t\rho_n \right) \int_{0}^{\infty} v(\rho_n,\rho,\tau,t) r_n(\rho_n-\rho) \notag \\
&\qquad  \qquad \qquad \qquad \qquad  \qquad \qquad \qquad \qquad v(\rho_n,\rho,\tau,t) r_n(\rho_n-\rho)  v(\rho_n,\rho,\tau,t) d\tau \bigg]  dt \notag \\
&  \quad +2\int_{0}^1 (1-t)\operatorname{Tr} \bigg[\big((1-t)\rho+t\rho_n\big) \int_{0}^{\infty} v(\sigma_n,\rho,\tau,t) r_n(\sigma_n-\rho) \notag \\
&\qquad  \qquad \qquad \qquad \qquad  \qquad \qquad \qquad \qquad   v(\sigma_n,\rho,\tau,t) r_n(\sigma_n-\rho)  v(\sigma_n,\rho,\tau,t)d\tau \bigg]dt \notag \\
& \quad -2 \int_{0}^1 (1-t)\operatorname{Tr} \bigg[r_n(\rho_n-\rho)\int_{0}^{\infty} v(\sigma_n,\rho,\tau,t) r_n(\sigma_n-\rho)  v(\sigma_n,\rho,\tau,t) d\tau \bigg]dt, \label{eq:limittwosampnull}
\end{align}
where we used that the first two terms above vanish. To see this for the first term,  note that
\begin{align}
   &\int_{0}^{\infty} \norm{\rho\left(\tau I+\rho\right)^{-1} (\rho_n-\rho) \left(\tau I+\rho\right)^{-1}}_1 d\tau \notag \\
   &=\int_{0}^{1} \norm{\rho\left(\tau I+\rho\right)^{-1} (\rho_n-\rho) \left(\tau I+\rho\right)^{-1}}_1 d\tau +\int_{1}^{\infty} \norm{\rho\left(\tau I+\rho\right)^{-1} (\rho_n-\rho) \left(\tau I+\rho\right)^{-1}}_1 d\tau\notag \\
 & \leq  \norm{\rho_n-\rho}_1 \norm{\rho^{-1}}_1+\norm{\rho (\rho_n-\rho) }_1\int_{1}^{\infty} \tau^{-2} d\tau\notag \\
 & \lesssim  \norm{ \rho^{-1}}_{\infty} +\int_{1}^{\infty} \tau^{-2} d\tau\notag \\
   &<\infty. \label{eq:trnormfin1}
\end{align}
Similarly, 
\begin{align}
   \int_{0}^{\infty} \norm{\rho(\rho_n-\rho)\left(\tau I+\rho\right)^{-2}}_1 d\tau <\infty. \label{eq:trnormfin2}
\end{align}
Hence, $\rho\left(\tau I+\rho\right)^{-1} (\rho_n-\rho) \left(\tau I+\rho\right)^{-1}$ and $\rho(\rho_n-\rho)\left(\tau I+\rho\right)^{-2}$ are  integrable functions (with respect to Lebesgue measure on $(0,\infty))$. Then, we have
\begin{align}
\tr{ \int_{0}^{\infty} \rho\left(\tau I+\rho\right)^{-1} (\rho_n-\rho) \left(\tau I+\rho\right)^{-1} d\tau } &\stackrel{(a)}{=}\int_{0}^{\infty}\tr{  \rho\left(\tau I+\rho\right)^{-1} (\rho_n-\rho) \left(\tau I+\rho\right)^{-1}  } d\tau\notag \\
&\stackrel{(b)}{=}\int_{0}^{\infty} \tr{   \rho(\rho_n-\rho)\left(\tau I+\rho\right)^{-2}  } d\tau\notag \\
&\stackrel{(c)}{=} \tr{   \rho(\rho_n-\rho)\int_{0}^{\infty}\left(\tau I+\rho\right)^{-2} d\tau } \notag \\
&=\tr{   \rho(\rho_n-\rho) \rho^{-1}} \notag \\
&=\tr{  \rho_n-\rho }\stackrel{(d)}{=}0, \label{eq:firstordtmzero}
\end{align}
where
    \begin{enumerate}[(a)]
        \item follows from \eqref{eq:trnormfin1} by the fact that  $\tr{\cdot}$ is continuous linear (hence bounded) functional on the normed space $\cL(\HH_d)$ (with $\norm{\cdot}_1$ norm); 
        \item uses that $[\rho,\left(\tau I+\rho\right)^{-1}]=0$ and the cylic property of $\tr{\cdot}$;
        \item follows via the same argument as in $(a)$ using   \eqref{eq:trnormfin2};
        \item is because $\tr{\rho_n}=\tr{\rho}=1$ on account of $\rho_n$ and $\rho$ being density operators. 
    \end{enumerate}
Likewise, it can be shown that 
\begin{align}
    \tr{ \int_{0}^{\infty} \rho\left(\tau I+\rho\right)^{-1} (\sigma_n-\rho) \left(\tau I+\rho\right)^{-1} d\tau }=0. \label{eq:tracevanishqrel}
\end{align}

We next analyze the limit of the expression in \eqref{eq:limittwosampnull}. To prove the  weak convergence of $g_n$ to its desired limit, it suffices to show that for every subsequence of $(g_n)_{n \in \NN}$, there exists a further subsequence along which the sequence converges to a  unique weak limit (see e.g., \cite[Theorem 2.6]{Billingsley-99}).  We refer to this as the \textit{subsequence argument}.
Let  
\begin{subequations}\label{eq:unifbochinttermqrel}
\begin{align}
   &p_n(r_n,t):=(1-t) r_n(\rho_n-\rho) \int_{0}^{\infty} v(\rho_n,\rho,\tau,t) r_n(\rho_n-\rho) v(\rho_n,\rho,\tau,t) d\tau,  \\
 &  q_n(r_n,t):=(1-t) \big((1-t) \rho+t\rho_n \big) \int_{0}^{\infty} v(\rho_n,\rho,\tau,t) r_n (\rho_n-\rho) v(\rho_n,\rho,\tau,t) r_n(\rho_n-\rho)v(\rho_n,\rho,\tau,t) d\tau,  \\
  &  \tilde p_n(r_n,t):=(1-t)\big((1-t)\rho+t\rho_n\big) \int_{0}^{\infty} \mspace{-5mu}v(\sigma_n,\rho,\tau,t) r_n(\sigma_n-\rho) v(\sigma_n,\rho,\tau,t) r_n(\sigma_n-\rho)  v(\sigma_n,\rho,\tau,t) d\tau, \\
    &\tilde q_n(r_n,t):=(1-t)r_n(\rho_n-\rho)  \int_{0}^{\infty} v(\sigma_n,\rho,\tau,t) r_n(\sigma_n-\rho)  v(\sigma_n,\rho,\tau,t) d\tau. 
\end{align} 
\end{subequations}
Then, the following bounds hold by using H\"{o}lder's inequality for Schatten norms:
\begin{subequations}\label{eq:boundsunifint}
\begin{align}
   &\norm{ p_n(r_n,t)}_1 
      \leq \norm{  r_n(\rho_n -\rho)}_1^2 \int_{0}^{\infty} \norm{v(\rho_n,\rho,\tau,t)}_{\infty}^2 d\tau \label{eq:boundsunifinteq1}, \\
     & \norm{q_n(r_n,t)}_1 
\leq \norm{  r_n(\rho_n -\rho )}_1^2 \int_{0}^{\infty} \norm{v(\rho_n,\rho,\tau,t)}_{\infty}^3 d\tau,  \label{eq:boundsunifinteq2}\\
&\norm{\tilde p_n(r_n,t)}_1 \leq   \norm{r_n(\sigma_n-\rho)}_1^2 \int_{0}^{\infty} \norm{v(\sigma_n,\rho,\tau,t)}_{\infty}^3 d\tau,   \\
&\norm{\tilde q_n(r_n,t)}_1  \leq\norm{r_n(\rho_n-\rho)}_1 \norm{r_n(\sigma_n-\rho)}_1 \int_{0}^{\infty} \norm{v(\sigma_n,\rho,\tau,t)}_{\infty}^2 d\tau.  
\end{align}
\end{subequations}

To show the aforementioned claim of unique weak limit, consider any subsequence $(n_k)_{k \in \NN}$. Then, $\big((r_{n_{k}}(\rho_{n_{k}}-\rho),r_{n_{k}}(\sigma_{n_{k}}-\rho)\big) \trightarrow{w} (L_1,L_2)$ in $\norm{\cdot}_1$ since every subsequence of weakly convergent sequence has the same weak limit. Hence, due to  separability of $\cL(\HH_d)$ (for finite $d$), by Skorokhods representation theorem (see e.g. \cite{AVDV-book}), there exists a further subsequence  $(n_{k_j})_{j \in \NN}$ such that 
$\big((r_{n_{k_j}}(\rho_{n_{k_j}}-\rho),r_{n_{k_j}}(\sigma_{n_{k_j}}-\rho)\big) \rightarrow (L_1,L_2)$ in $\norm{\cdot}_1$ almost surely (a.s.).     
Then, since $\sigma_{n_{k_j}} \rightarrow \rho$,  we have  that $(1-t)\rho+t\sigma_{n_{k_j}} \geq c \rho$  for a constant $0<c<1$ (which depends on the realization $\sigma_{n_{k_j}}$) and  sufficiently large $j$. This implies that the integrals in \eqref{eq:boundsunifint} are finite. For instance, the integral in \eqref{eq:boundsunifinteq1} is $O(1)$ as $\tau$ approaches zero   and $O(\tau^{-2})$ as $\tau$ tends to $\infty$. Hence,   we obtain
\begin{subequations} \label{eq:boundsunifint2}
\begin{align}
 &\norm{ p_{n_{k_j}}(r_{n_{k_j}},t)}_1 
      \lesssim_{\rho} \norm{  r_{n_{k_j}}(\rho_{n_{k_j}} -\rho)}_1^2 \label{eq:boundsunifint2-1} , \\
     & \norm{q_{n_{k_j}}(r_{n_{k_j}},t)}_1 
\lesssim_{\rho} \norm{  r_{n_{k_j}}(\rho_{n_{k_j}} -\rho )}_1^2, \label{eq:boundsunifint2-2}\\
&\norm{\tilde p_{n_{k_j}}(r_{n_{k_j}},t)}_1 \lesssim_{\rho}   \norm{r_{n_{k_j}}(\sigma_{n_{k_j}}-\rho)}_1^2 , \label{eq:boundsunifint2-3}  \\
&\norm{\tilde q_{n_{k_j}}(r_{n_{k_j}},t)}_1  \lesssim_{\rho} \norm{r_{n_{k_j}}(\rho_{n_{k_j}}-\rho)}_1 \norm{r_{n_{k_j}}(\sigma_{n_{k_j}}-\rho)}_1. \label{eq:boundsunifint2-4} 
\end{align}
\end{subequations}
Next, recall that $\big((r_{n_{k_j}}(\rho_{n_{k_j}}-\rho),r_{n_{k_j}}(\sigma_{n_{k_j}}-\rho)\big) \rightarrow (L_1,L_2)$ (a.s.) in the space of linear operators with bounded trace norm implies that $\big(r_{n_{k_j}}(\rho_{n_{k_j}}-\rho)\big)_{j \in \NN}$ and $\big(r_{n_{k_j}}(\sigma_{n_{k_j}}-\rho)\big)_{j \in \NN}$ have  bounded trace norms along the chosen sample path. Combined with \eqref{eq:boundsunifint} and \eqref{eq:boundsunifint2}, this gives a single integrable dominating function of $t\in[0,1]$ for the four Bochner-valued remainder integrands, independent of $j$ for all sufficiently large $j$, which we refer to as uniform (Bochner) integrability. Then, Bochner dominated convergence justifies the interchange of limit and integral. Since \eqref{eq:boundsunifint} and \eqref{eq:boundsunifint2}   shows that $\big(p_{n_{k_j}}(r_{n_{k_j}},t)\big)_{j \in \NN}$, $\big(q_{n_{k_j}}(r_{n_{k_j}},t)\big)_{j \in \NN}$, $(\tilde p_{n_{k_j}}(r_{n_{k_j}},t))_{j \in \NN}$ and $(\tilde q_{n_{k_j}}(r_{n_{k_j}},t))_{j \in \NN}$ are uniformly integrable sequences,   taking limits $j \rightarrow \infty $ in \eqref{eq:limittwosampnull},  interchanging limits and integral and noting that $\rho_{n_{k_j}} \rightarrow \rho$, $\sigma_{n_{k_j}} \rightarrow \rho$ a.s. and $\norm{\cdot}_1$,  yields
 \begin{align}
    g_{n_{k_j}} &\rightarrow    \operatorname{Tr} \bigg[L_1 \int_{0}^{\infty} \left(\tau I+\rho\right)^{-1} L_1\left(\tau I+\rho\right)^{-1} d\tau-\rho \int_{0}^{\infty} \left(\tau I+\rho\right)^{-1} L_1  \left(\tau I+\rho\right)^{-1} L_1 \left(\tau I+\rho\right)^{-1} d\tau \bigg]  \notag \\
    &   \qquad+\operatorname{Tr} \bigg[\rho \int_{0}^{\infty} \left(\tau I+\rho\right)^{-1} L_2  \left(\tau I+\rho\right)^{-1} L_2 \left(\tau I+\rho\right)^{-1}d\tau-L_1 \int_{0}^{\infty} \left(\tau I+\rho\right)^{-1} L_2 \left(\tau I+\rho\right)^{-1} d\tau \bigg] \notag \\
&\quad=\tr{L_1 D[\log \rho](L_1-L_2)}  +\tr{\frac{\rho}{2}D^2[\log \rho](L_1-L_2,L_1-L_2)}. \notag
 \end{align}
     Hence, every  subsequence $(g_{n_k})_{k \in \NN}$ has a further subsequence $(g_{n_{k_j}})_{j \in \NN}$ with the same unique limit which implies \eqref{eq:qrel-twosample-null}.

\medskip

    If  $[\rho_n,\rho]=0$, then $[r_n(\rho_n-\rho), \rho]=0$. Consider a subsequence  $r_{n_j}(\rho_{n_j}-\rho) \rightarrow L_1$ a.s. Since the commutator is a continuous linear functional of its individual arguments, we have 
    \begin{align}
        [L_1, \rho]=\lim_{j \rightarrow \infty}[r_{n_j}(\rho_{n_j}-\rho), \rho]=0, \mbox{ a.s.} \label{eq:prfcommutlim}
    \end{align}
 Hence, $L_1$ and $\rho$ commutes. The proof of  $[L_1,L_2]=[L_2,\rho]=0$ under the conditions $[\rho_n,\rho]=[\sigma_n,\rho]=[\sigma_n,\rho_n]=0$ follow similarly. Under this scenario, the expression in the RHS of \eqref{eq:qrel-twosample-null} simplifies to 
\begin{align}
&\operatorname{Tr} \bigg[L_1^2 \int_{0}^{\infty} \left(\tau I+\rho\right)^{-2}  d\tau\bigg]   - \operatorname{Tr} \bigg[L_1^2\rho \int_{0}^{\infty} \left(\tau I+\rho\right)^{-3}  d\tau \bigg] +\operatorname{Tr} \bigg[L_2^2\rho \int_{0}^{\infty} \left(\tau I+\rho\right)^{-3} d\tau \bigg]  \notag \\
& \quad  -\operatorname{Tr} \bigg[L_1 L_2\int_{0}^{\infty} \left(\tau I+\rho\right)^{-2}  d\tau \bigg] \notag \\
&=\frac 12 \operatorname{Tr} \big[L_1^2 \rho^{-1}\big] +\frac 12 \operatorname{Tr} \big[L_2^2 \rho^{-1}\big]-\operatorname{Tr} \big[L_1L_2 \rho^{-1}\big] \notag \\
&=\frac 12 \operatorname{Tr} \big[(L_1-L_2)^2 \rho^{-1}\big]. \notag 
\end{align}

Finally, consider the case $ \rho \geq 0$. Note that  the left-hand side (LHS) and RHS of \eqref{eq:qrel-twosample-null} is invariant to restricting the space to the support of $\rho$. To see this,   let $P_{\rho}=\sum_{i=1}^r \ket{e_i}\bra{e_i}$ be the projector onto the eigenspace pertaining to the non-zero eigenvalues of $\rho$, where $r$ is the rank of $\rho$ and $(\ket{e_i})_{i=1}^r$ are the corresponding orthonormal eigenvectors. 
Setting $\tilde \rho=P_{\rho} \rho P_{\rho}$, $\tilde \rho_n=P_{\rho} \rho_n P_{\rho}$, $\tilde \sigma_n=P_{\rho} \sigma_n P_{\rho}$,   
 and noting that $\rho_n \ll \sigma_n \ll \rho \ll \gg P_{\rho}$, it follows from Lemma \ref{lem:multprojtr}$(i)$ that 
\begin{align}
    \qrel{\rho_n}{\sigma_n}&=\tr{\rho_n\log \rho_n}-\tr{\rho_n \log \sigma_n} \notag \\
    &=\tr{ P_{\rho}\rho_n P_{\rho}\log  \rho_n P_{\rho}}-\tr{ P_{\rho}\rho_n P_{\rho} \log  \sigma_n P_{\rho}} \notag \\
     &=\tr{\tilde \rho_n\log \tilde \rho_n}-\tr{\tilde \rho_n  \log \tilde \sigma_n }. \notag
\end{align}
Note that $\tilde \rho >0$ and $\tilde \rho_n, \tilde \sigma_n \geq 0$ are density operators.

Next, to see that the RHS of \eqref{eq:qrel-twosample-null} is invariant to restricting to support of $\rho$, we first note that the  support of $L_1$ and $L_2$ is contained in that of $\rho$. To show this, notice that for every $n \in \NN$,  $\rho_n-\rho \ll \rho$ and $\sigma_n-\rho \ll \rho$  because $\rho_n,\sigma_n \ll \rho$ by assumption. By Portmanteau's theorem \cite[Theorem 1.3.4 $(vii)$]{AVDV-book}, since $r_n(\rho_n-\rho) \trightarrow{w} L_1$, we have 
\begin{align}
    \liminf \EE[f\big(r_n(\rho_n-\rho)\big)] \geq f(L_1), \label{eq:portmantthm}
\end{align}
for every bounded Lipschitz continuous (w.r.t. to trace norm) non-negative $f$. Let $P_{\rho}^{\perp}$ denote the projector onto the orthogonal complement of the support of $\rho$.  Applying \eqref{eq:portmantthm} to the bounded Lipschitz continuous function $f_{M}(L)=\norm{P_{\rho}^{\perp} LP_{\rho}^{\perp}}_1 \wedge M$ on the space of trace-class operators, where $M >0$, we obtain
\begin{align}
0= \liminf r_n\EE\left[\norm{P_{\rho}^{\perp}(\rho_n-\rho)P_{\rho}^{\perp}}_1\right] \geq r_n \norm{P_{\rho}^{\perp}L_1P_{\rho}^{\perp}}_1 \wedge M \geq 0. \notag 
\end{align}
Since $r_n $ is positive and the above equation has to hold for every $M$, taking limit $M \rightarrow \infty$ implies that  $\norm{P_{\rho}^{\perp}L_1 P_{\rho}^{\perp}}_1=0$. Hence, the support of $L_1$ is contained in that of $\rho$. Similar claim also holds for $L_2$. 
 Thus, the RHS is also invariant to replacing all operators by their sandwiched versions obtained by left and right multiplying with $P_{\rho}$. Finally, observe that  
 $\big(r_n(\rho_n-\rho),r_n(\sigma_n-\rho)\big) \trightarrow{w} (L_1,L_2)$ implies $\big(r_n(\tilde \rho_n-\tilde \rho),r_n(\tilde \sigma_n-\tilde \rho)\big) \trightarrow{w} (P_{\rho} L_1 P_{\rho},P_{\rho} L_2 P_{\rho})$ by an application of Slutsky's theorem \cite{AVDV-book}. Hence, the previous proof applies  and the claim follows. 
\medskip

\noindent
\textbf{Two-sample alternative:} Assume first that $\rho_n,\sigma_n,\rho, \sigma >0$. Setting $B_1=\rho_n$, $B_2=\sigma_n$, $A_1=\rho$ and $A_2=\sigma \neq \rho$ in   \eqref{eq:taylor-op}, we have
 \begin{align}
&f(\rho_n,\sigma_n)\notag \\
&=f(\rho,\sigma)+\rho \int_{0}^{\infty} \left(\tau I+\rho\right)^{-1} (\rho_n-\rho) \left(\tau I+\rho\right)^{-1} d\tau  -\rho \int_{0}^{\infty} \mspace{-8 mu} \left(\tau I+\sigma\right)^{-1} \mspace{-4 mu}(\sigma_n-\sigma) \mspace{-3 mu}\left(\tau I+\sigma\right)^{-1} d\tau  \notag \\
     & \quad  + (\rho_n-\rho)\left(\log \rho-\log \sigma\right) + 2\int_{0}^1  (1-t) \mspace{-1 mu} (\rho_n-\rho) \int_{0}^{\infty}  v(\rho_n,\rho,\tau,t)  (\rho_n-\rho) v(\rho_n,\rho,\tau,t) d\tau dt   \notag \\
& \quad -2\int_{0}^1 (1-t) \left((1-t) \rho+t\rho_n \right)\int_{0}^{\infty} v(\rho_n,\rho,\tau,t) (\rho_n-\rho) v(\rho_n,\rho,\tau,t) (\rho_n-\rho) v(\rho_n,\rho,\tau,t) d\tau  dt  \notag \\
& \quad  + \mspace{-2 mu} 2\int_{0}^1 \mspace{-2 mu} (1-t)\big((1-t)\rho+t\rho_n\big) \mspace{-2 mu}\int_{0}^{\infty}\mspace{-4 mu}v(\sigma_n,\sigma,\tau,t)   (\sigma_n-\sigma) v(\sigma_n,\sigma,\tau,t)  (\sigma_n-\sigma)  v(\sigma_n,\sigma,\tau,t) d\tau dt\notag \\
&\quad-2 \int_{0}^1  (1-t)(\rho_n-\rho) \int_{0}^{\infty} v(\sigma_n,\sigma,\tau,t)  (\sigma_n-\sigma)  v(\sigma_n,\sigma,\tau,t)  d\tau dt. \label{eq:taylexpqrelalt}
 \end{align}
 This equation extends to $0 \leq \rho_n \ll \sigma_n \geq 0$ via similar arguments in Part $(i)$.  Multiplying by $r_n$ and taking trace, we obtain
 \begin{align}
g_n&:=r_n\big(\qrel{\rho_n}{\sigma_n}-\qrel{\rho}{\sigma}\big) \notag \\
 &=\tr{r_n(\rho_n-\rho)\left(\log \rho-\log \sigma\right)-\rho \mspace{-2 mu}\int_{0}^{\infty} \mspace{-6 mu}\left(\tau I+\sigma\right)^{-1} r_n(\sigma_n-\sigma) \left(\tau I+\sigma\right)^{-1} d\tau   }\mspace{-4 mu}\notag \\
 &\quad +\mspace{-4 mu}\operatorname{Tr}\bigg[2\int_{0}^1 (1-t)   r_n^{\frac 12}(\rho_n-\rho) \int_{0}^{\infty} v(\rho_n,\rho,\tau,t)  r_n^{\frac 12}(\rho_n-\rho) v(\rho_n,\rho,\tau,t) d\tau \bigg] dt  \notag \\
     & \quad -2\int_{0}^1(1-t)\operatorname{Tr}\bigg[ \left((1-t) \rho+t\rho_n \right) \int_{0}^{\infty}v(\rho_n,\rho,\tau,t) r_n^{\frac 12}(\rho_n-\rho) v(\rho_n,\rho,\tau,t)  r_n^{\frac 12}(\rho_n-\rho)\notag \\
     & \qquad \qquad \qquad \qquad \qquad \qquad \qquad \qquad \qquad \qquad \qquad \qquad \qquad \qquad \qquad \qquad v(\rho_n,\rho,\tau,t) d\tau \bigg]  dt \notag \\  
& \quad~~   +2\int_{0}^1 (1-t)\operatorname{Tr}\bigg[\big((1-t)\rho+t\rho_n\big)  \int_{0}^{\infty} v(\sigma_n,\sigma,\tau,t) r_n^{\frac 12}(\sigma_n-\sigma) v(\sigma_n,\sigma,\tau,t) r_n^{\frac 12}(\sigma_n-\sigma)  \notag \\
&\qquad \qquad \qquad \qquad \qquad \qquad \qquad \qquad \qquad \qquad \qquad \qquad \qquad \qquad \qquad \qquad v(\sigma_n,\sigma,\tau,t)d\tau \bigg] dt \notag \\
&\quad~~-2 \int_{0}^1 (1-t)\operatorname{Tr}\bigg[r_n^{\frac 12}(\rho_n-\rho) \int_{0}^{\infty} v(\sigma_n,\sigma,\tau,t)r_n^{\frac 12}(\sigma_n-\sigma) v(\sigma_n,\sigma,\tau,t) d\tau \bigg] dt, \label{eq:bndtrqrelalt}
 \end{align}
 where we used \eqref{eq:firstordtmzero}. 
 
 Let
 \begin{subequations}\label{eq:bochintaltqrel}
  \begin{align}
  \bar p_n(r_n,t)&:=(1-t)\big((1-t)\rho+t\rho_n\big) \mspace{-5 mu}\int_{0}^{\infty} \mspace{-7 mu}v(\sigma_n,\sigma,\tau,t)  r_n^{\frac 12}(\sigma_n-\sigma)v(\sigma_n,\sigma,\tau,t)   r_n^{\frac 12}(\sigma_n-\sigma)  v(\sigma_n,\sigma,\tau,t) d\tau,  \label{eq:bochintaltqrel1}  \\  
\bar q_n(r_n,t)&:=(1-t) r_n^{\frac 12}(\rho_n-\rho) \int_{0}^{\infty} v(\sigma_n,\sigma,\tau,t) r_n^{\frac 12}(\sigma_n-\sigma)  v(\sigma_n,\sigma,\tau,t) d\tau. \label{eq:bochintaltqrel2} 
 \end{align}    
 \end{subequations}
 Then, via steps akin to \eqref{eq:boundsunifint}, we have
    \begin{align}
\norm{\bar p_n(r_n,t)}_1& \leq \norm{r_n^{\frac 12}(\sigma_n-\sigma)}_1^2 \int_{0}^{\infty} \norm{v(\sigma_n,\sigma,\tau,t)}_{\infty}^3 d\tau, \notag \\
 \norm{\bar q_n(r_n,t)}_1& \leq \norm{r_n^{\frac 12}(\rho_n-\rho)}_1\norm{r_n^{\frac 12}(\sigma_n-\sigma)}_1  \int_{0}^{\infty} \norm{v(\sigma_n,\sigma,\tau,t)}_{\infty}^2 d\tau.\notag   
\end{align}
As in Part $(i)$, for any subsequence $(n_k)_{k \in \NN}$, consider a further subsequence  $(n_{k_j})_{j \in \NN}$ such that  
$\big((r_{n_{k_j}}(\rho_{n_{k_j}}-\rho),r_{n_{k_j}}(\sigma_{n_{k_j}}-\sigma)\big) \rightarrow (L_1,L_2)$  in $\norm{\cdot}_1$ a.s. Noting that there exists a constant $0 <c<1$ such that  $(1-t)\sigma+t\sigma_{n_{k_j}} \geq c \sigma$  for  sufficiently large $j$, we have  
\begin{subequations}\label{eq:unifbint}
\begin{align}
\norm{\bar p_{n_{k_j}}(r_{n_{k_j}},t)}_1& \lesssim_{\sigma} \norm{r_{n_{k_j}}^{\frac 12}(\sigma_{n_{k_j}}-\sigma)}_1^2,  \\
 \norm{\bar q_{n_{k_j}}(r_{n_{k_j}},t)}_1& \lesssim_{\sigma}  \norm{r_{n_{k_j}}^{\frac 12}(\rho_{n_{k_j}}-\rho)}_1\norm{r_{n_{k_j}}^{\frac 12}(\sigma_{n_{k_j}}-\sigma)}_1 . 
\end{align}
\end{subequations}
The  above equations and \eqref{eq:boundsunifint2} subsequently implies that  $\big(p_{n_{k_j}}(r_{n_{k_j}},t)\big)_{j \in \NN}$, $\big(q_{n_{k_j}}(r_{n_{k_j}},t)\big)_{j \in \NN}$, $\big(\bar p_{n_{k_j}}(r_{n_{k_j}},t)\big)_{j \in \NN}$ and $\big(\bar q_{n_{k_j}}(r_{n_{k_j}},t)\big)_{j \in \NN}$ are uniformly integrable. Moreover, $\big((r_{n_{k_j}}^{1/2}(\rho_{n_{k_j}}-\rho),$ $r_{n_{k_j}}^{1/2}(\sigma_{n_{k_j}}-\sigma)\big) \rightarrow (0,0)$. Taking limits $j \rightarrow \infty $ and interchanging limits and integral  yields
\begin{align}
g_{n_{k_j}}&\rightarrow \tr{L_1\left(\log \rho-\log \sigma\right) -\rho \int_{0}^{\infty} \left(\tau I+\sigma\right)^{-1} L_2 \left(\tau I+\sigma\right)^{-1} d\tau} \notag \\
&=\tr{L_1\left(\log \rho-\log \sigma\right)} -\tr{\rho D[\log \sigma](L_2)}. \notag
\end{align}
This implies \eqref{eq:qrel-twosample-alt-comm} via the subsequence argument mentioned in Part $(i)$.

\medskip
In the commutative case, we observe similar to \eqref{eq:prfcommutlim} that  $[L_1,L_2]=[L_1,\rho]=[L_1,\sigma]=[L_2,\rho]=[L_2,\sigma]=0$ when $[\rho_n,\rho]=[\sigma_n,\sigma]=[\rho,\sigma]=[\sigma_n,\rho_n]=[\sigma_n,\rho]=[\sigma,\rho_n]=0$. In this case, the above limit simplifies as
\begin{align}
  &\tr{\rho \int_{0}^{\infty} \left(\tau I+\sigma\right)^{-1} L_2 \left(\tau I+\sigma\right)^{-1} d\tau} =  \tr{L_2 \rho \sigma^{-1}}. \notag 
\end{align}

Finally, consider the case $0 \leq \rho \ll \sigma \geq 0$ 
and $P_{\rho}$ be the projector onto support of $\rho$ as defined in Part $(i)$. Since $\rho_n  \ll  \rho \ll \gg  P_{\rho} $,
we have $P_{\rho} \rho_n=\rho_n P_{\rho}=\rho_n $ and by cylicity of trace
\begin{align}
\qrel{\rho_n}{\sigma_n}-\qrel{\rho}{\sigma}&=\tr{\rho_n (\log \rho_n-\log \sigma_n)}-\tr{\rho (\log \rho-\log \sigma)}  \notag \\
&=\tr{P_{\rho}\rho_n (\log \rho_n-\log \sigma_n)P_{\rho}}-\tr{P_{\rho}\rho (\log \rho-\log \sigma)P_{\rho}}.\notag
\end{align}
Since $\rho_n,\sigma_n,\rho \ll \sigma$, we may assume without loss of generality that $\sigma>0$  by restricting the underlying Hilbert space to the support of $\sigma$. Note that $P_{\rho} L_1=L_1$ and $P_{\rho} \rho=\rho$  due to $\rho_n,L_1 \ll \rho$. Consider the function $f(A_1,A_2)=P_{\rho}A_1(\log A_1-\log A_2)P_{\rho}$. Note that $f(A_1,A_2)$ is continuously twice differentiable at  $(A_1,A_2)$ such that $0 \leq A_1 \ll \rho$ and $A_2>0$. Then,  applying the operator version of Taylor's theorem at $(\rho,\sigma)$   and following similar steps as above (for the case $\rho,\sigma>0$) yields 
\begin{align}
    r_n\big(\qrel{\rho_n}{\sigma_n}-\qrel{\rho}{\sigma}\big)  \trightarrow{w} &\tr{P_{\rho}L_1\left(\log \rho-\log \sigma\right) P_{\rho}} -\tr{P_{\rho}\rho D[\log \sigma](L_2)P_{\rho}} \notag \\
    &=\tr{L_1\left(\log \rho-\log \sigma\right) } -\tr{\rho D[\log \sigma](L_2)}, \notag
\end{align}
provided $(P_{\rho} p_n(r_n,t)P_{\rho})_{n \in \NN}$, $(P_{\rho} q_n(r_n,t)P_{\rho})_{n \in \NN}$, $(P_{\rho}\bar p_n(r_n,t)P_{\rho})_{n \in \NN}$ and $(P_{\rho}\bar q_n(r_n,t)P_{\rho})_{n \in \NN}$ (see \eqref{eq:boundsunifint2} and \eqref{eq:bochintaltqrel}) are uniformly integrable along the subsequence $(n_{k_j})_{j \in \NN}$ such that  
$\big((r_{n_{k_j}}(\rho_{n_{k_j}}-\rho),r_{n_{k_j}}(\sigma_{n_{k_j}}-\sigma)\big) \rightarrow (L_1,L_2)$  in $\norm{\cdot}_1$ a.s.  The uniform integrability of the last two subsequences follow via analogous arguments leading to \eqref{eq:unifbint} since $\sigma>0$. Moreover, we have
\begin{align}
 P_{\rho} p_n(r_n,t) P_{\rho}&:=  P_{\rho}  (1-t) r_n(\rho_n-\rho)\int_{0}^{\infty}  v(\rho_n,\rho,\tau,t)   r_n(\rho_n-\rho) v(\rho_n,\rho,\tau,t)  d\tau  P_{\rho} \notag \\
 &=(1-t) r_n(\rho_n-\rho)\int_{0}^{\infty}  \bar v(\rho_n,\rho,\tau,t)   r_n(\rho_n-\rho) \bar v(\rho_n,\rho,\tau,t)  d\tau, \notag
\end{align}
where $\bar v(\rho_n,\rho,\tau,t):=P_{\rho}\big(\tau I+(1-t)\rho+t\rho_n\big)^{-1}P_{\rho} \ll \rho $. The last equality follows since the operators coincide on the support $\rho$ and, otherwise act as zero operator  on the kernel of $\rho$. 
Then, via similar arguments leading to \eqref{eq:boundsunifint2-1} and \eqref{eq:boundsunifint2-1}, we have
\begin{align}
   &\norm{ P_{\rho} p_{n_{k_j}}(r_{n_{k_j}},t) P_{\rho}}_1 
      \leq \norm{  r_{n_{k_j}}(\rho_{n_{k_j}} -\rho)}_1^2 \int_{0}^{\infty} \norm{\bar v(\rho_{n_{k_j}},\rho,\tau,t)}_{\infty}^2 d\tau \lesssim_{\rho} \norm{  r_{n_{k_j}}(\rho_{n_{k_j}} -\rho)}_1^2, \notag \\
     & \norm{ P_{\rho} q_{n_{k_j}}(r_{n_{k_j}},t) P_{\rho}}_1 
\leq \norm{  r_{n_{k_j}}(\rho_{n_{k_j}} -\rho )}_1^2 \int_{0}^{\infty} \norm{\bar v(\rho_{n_{k_j}},\rho,\tau,t)}_{\infty}^3 d\tau \lesssim_{\rho} \norm{  r_{n_{k_j}}(\rho_{n_{k_j}} -\rho)}_1^2, \notag
\end{align}
from which the uniform integrability of the first two sequences follow. This completes the proof.

\subsection{Proof of Theorem \ref{Thm:Petzrelent-limdist}} \label{Thm:Petzrelent-limdist-proof}
\noindent
In the following, we will assume without loss of generality that $\rho,\sigma>0$. The proofs for extending  to  the general case $\rho \ll \sigma$ follows via similar arguments as given in the proof of Theorem \ref{Thm:quantrelent-limdist}.  
We first prove Part $(i)$.

\noindent
\textbf{Two-sample alternative:} 
Consider the case $\alpha \in (0,1)$,   $\bar \alpha:=1-\alpha $,  and $Q_{\alpha}(\rho,\sigma):=\tr{\rho^{\alpha}\sigma^{\bar \alpha}}$. We will initially show that
\begin{align}
  r_n\big(Q_{\alpha}(\rho_n,\sigma_n)-Q_{\alpha}(\rho,\sigma)\big) \trightarrow{w} & \tr{\sigma^{\bar \alpha} D[\rho^{\alpha}](L_1)}+\tr{\rho^{\alpha}D[\sigma^{\bar \alpha}](L_2)} \notag \\
&= c_{\alpha} \tr{\int_0^{\infty} \tau^{\alpha} (\tau I+\rho)^{-1}L_1 (\tau I+\rho)^{-1} d\tau \sigma^{\bar \alpha}} \notag \\
& \qquad + c_{\bar \alpha} \tr{\rho^{\alpha} \int_0^{\infty} \tau^{\bar \alpha} (\tau I+\sigma)^{-1}L_2 (\tau I+\sigma)^{-1} d\tau}, \notag
\end{align}
where $c_{\alpha}=\sin(\pi \alpha)/\pi$. Then, applying the functional delta method \cite[Theorem 3.9.4]{AVDV-book} with $\phi(x)= \log x/(\alpha-1)$ at $x=Q_{\alpha}(\rho,\sigma)$ leads to 
\begin{align}
   r_n \big(\petzdiv{\rho_n}{\sigma_n}{\alpha}-\petzdiv{\rho}{\sigma}{\alpha}\big) \trightarrow{w} \frac{\tr{\sigma^{\bar \alpha} D[\rho^{\alpha}](L_1)}+\tr{\rho^{\alpha}D[\sigma^{1-\alpha}](L_2)}}{(\alpha-1)\tr{\rho^{\alpha}\sigma^{\bar \alpha}}}, \label{eq:qalphlim}
\end{align}
as claimed in \eqref{eq:petzrelent-twosample-alt}.

To show the above, we  compute the Fr\'{e}chet derivatives of the operator-valued function $f(A_1,A_2)=A_1^{\alpha}A_2^{1-\alpha}$. Using the integral representation \cite[Lemma 2.8]{Carlen-2010} 
\begin{align}
A^{\alpha}=c_{\alpha}\int_{0}^{\infty} \tau ^{\alpha} \left(\frac{1}{\tau I}-\frac{1}{\tau I+A}\right), ~\alpha \in (0,1),\notag
\end{align}
 we have via chain and product rule for Fr\'{e}chet derivatives that
\begin{align*}
&D^{(1,0)}[f(A_1,A_2)](H)=c_{\alpha} \int_{0}^{\infty} \tau^{\alpha} (\tau I+A_1)^{-1} H (\tau I+A_1)^{-1}  d\tau A_2^{\bar \alpha},  \\
&D^{(0,1)}[f(A_1,A_2)](H)=c_{\bar \alpha} A_1^{\alpha} \int_{0}^{\infty} \tau^{\bar \alpha} (\tau I+A_2)^{-1} H (\tau I+A_2)^{-1}  d\tau,  \\
&D^{(2,0)}[f(A_1,A_2)](H_1,H_2)= -c_{\alpha} \int_0^{\infty} \tau^{\alpha}  (\tau I+A_1)^{-1}H_1  (\tau I+A_1)^{-1}H_2  (\tau I+A_1)^{-1} d\tau A_2^{\bar \alpha} \\
& \qquad \qquad\qquad \qquad \qquad\qquad -c_{\alpha} \int_0^{\infty} \tau^{\alpha}  (\tau I+A_1)^{-1}H_2  (\tau I+A_1)^{-1}H_1  (\tau I+A_1)^{-1} d\tau A_2^{\bar \alpha}, \\
&D^{(0,2)}[f(A_1,A_2)](H_1,H_2)= -c_{\bar \alpha} A_1^{\alpha} \int_0^{\infty} \tau^{\bar \alpha}  (\tau I+A_2)^{-1}H_1  (\tau I+A_2)^{-1}H_2  (\tau I+A_2)^{-1} d\tau  \\
& \qquad \qquad\qquad \qquad \qquad\qquad -c_{\bar \alpha} A_1^{\alpha} \int_0^{\infty} \tau^{\bar \alpha}  (\tau I+A_2)^{-1}H_2  (\tau I+A_2)^{-1}H_1  (\tau I+A_2)^{-1} d\tau,   \\
&D^{(1,1_+)}[f(A_1,A_2)](H_1,H_2)=c_{\alpha}c_{\bar \alpha} \int_{0}^{\infty} \tau^{\alpha}  (\tau I+A_1)^{-1}H_1  (\tau I+A_1)^{-1} d\tau \notag \\
&\qquad \qquad \qquad \qquad\qquad \qquad\qquad \qquad\int_{0}^{\infty} \tau^{\bar \alpha}  (\tau I+A_2)^{-1}H_2  (\tau I+A_2)^{-1} d\tau,  \\
&D^{(1_+,1)}[f(A_1,A_2)](H_1,H_2)=c_{\alpha}c_{\bar \alpha} \int_{0}^{\infty} \tau^{\alpha}  (\tau I+A_1)^{-1}H_2  (\tau I+A_1)^{-1} d\tau \notag \\
&\qquad \qquad \qquad \qquad\qquad \qquad\qquad \qquad\int_{0}^{\infty} \tau^{\bar \alpha}  (\tau I+A_2)^{-1}H_1  (\tau I+A_2)^{-1} d\tau.
\end{align*}
Then, from \eqref{eq:taylor-op} with $B_1=\rho_n$, $B_2=\sigma_n$, $A_1=\rho$, and $A_2=\sigma$ and $v(\rho_n,\rho,\tau,t):=\big(\tau I+(1-t)\rho+t\rho_n\big)^{-1}$, we obtain
\begin{align}
\rho_n^{\alpha}\sigma_n^{\bar \alpha} &=\rho^{\alpha}\sigma^{\bar \alpha}+ c_{\alpha} \int_0^{\infty} \tau^{\alpha} (\tau I+\rho)^{-1}(\rho_n-\rho) (\tau I+\rho)^{-1} d\tau \,\sigma^{\bar \alpha}  \notag \\
&\quad + c_{\bar \alpha} \rho^{\alpha} \int_0^{\infty} \tau^{\bar \alpha} (\tau I+\sigma)^{-1}(\sigma_n-\sigma)   (\tau I+\sigma)^{-1} d\tau \notag \\
& \quad -2c_{\alpha} \mspace{-4 mu} \int_0^1  \mspace{-4 mu}(1-t) \bigg[\int_0^{\infty} \tau^{\alpha} v(\rho_n,\rho,\tau,t) (\rho_n-\rho) v(\rho_n,\rho,\tau,t) (\rho_n-\rho) v(\rho_n,\rho,\tau,t) d\tau \bigg] \notag \\
& \qquad \qquad \qquad \qquad\qquad \qquad\qquad \qquad \qquad\qquad \qquad \qquad\qquad \qquad\big((1-t)\sigma+t\sigma_n\big)^{\bar \alpha} dt \notag \\
&\quad  -2c_{\bar \alpha} \int_0^1 (1-t) \big((1-t)\rho+t\rho_n\big)^{ \alpha}\bigg[\int_0^{\infty} \tau^{\bar \alpha} v(\sigma_n,\sigma,\tau,t) (\sigma_n-\sigma) v(\sigma_n,\sigma,\tau,t) (\sigma_n-\sigma)  \notag \\
& \qquad \qquad \qquad \qquad\qquad \qquad\qquad \qquad \qquad\qquad \qquad \qquad\qquad \qquad  v(\sigma_n,\sigma,\tau,t) d\tau \bigg] dt \notag \\
&\quad +2c_{\alpha}c_{\bar \alpha} \int_0^1 (1-t)\bigg[\int_0^{\infty} \tau^{\alpha} v(\rho_n,\rho,\tau,t) (\rho_n-\rho)v(\rho_n,\rho,\tau,t) d\tau \bigg]  \notag \\
&\qquad \qquad \qquad\qquad \qquad\qquad \qquad \qquad\bigg[\int_0^{\infty} \tau^{\bar \alpha} v(\sigma_n,\sigma,\tau,t) (\sigma_n-\sigma) v(\sigma_n,\sigma,\tau,t) d\tau \bigg]  dt. \notag
\end{align}
Multiplying by $r_n$, taking trace, and subsequent limits leads to \eqref{eq:qalphlim} using similar arguments as in Theorem \ref{Thm:quantrelent-limdist}, provided 
 \begin{align}
  \bar p_n(r_n,t)&:=(1-t)  \bigg[\int_0^{\infty} \tau^{\alpha} v(\rho_n,\rho,\tau,t) r_n^{\frac 12}(\rho_n-\rho) v(\rho_n,\rho,\tau,t) r_n^{\frac 12}(\rho_n-\rho) v(\rho_n,\rho,\tau,t) d\tau \bigg]  \notag \\
  & \qquad \qquad \qquad\qquad \qquad\qquad \qquad \qquad\qquad \qquad \qquad\qquad \qquad \qquad \qquad \big((1-t)\sigma+t\sigma_n\big)^{\bar \alpha}, \notag \\
\bar q_n(r_n,t)&:=(1-t) \big((1-t)\rho+t\rho_n\big)^{ \alpha}\bigg[\int_0^{\infty} \tau^{\bar \alpha} v(\sigma_n,\sigma,\tau,t) r_n^{\frac 12}(\sigma_n-\sigma) v(\sigma_n,\sigma,\tau,t) r_n^{\frac 12}(\sigma_n-\sigma) \notag \\
& \qquad \qquad \qquad\qquad \qquad\qquad \qquad \qquad\qquad \qquad \qquad\qquad \qquad \qquad   \qquad v(\sigma_n,\sigma,\tau,t)d\tau\bigg], \notag \\
\bar s_n(r_n,t)&:=(1-t)\bigg[\int_0^{\infty} \tau^{\alpha} v(\rho_n,\rho,\tau,t) r_n^{\frac 12}(\rho_n-\rho) v(\rho_n,\rho,\tau,t) d\tau \bigg] \bigg[\int_0^{\infty} \tau^{\bar \alpha} v(\sigma_n,\sigma,\tau,t) r_n^{\frac 12}(\sigma_n-\sigma) \notag \\
&\qquad \qquad \qquad\qquad \qquad\qquad \qquad \qquad\qquad \qquad \qquad\qquad \qquad \qquad \qquad v(\sigma_n,\sigma,\tau,t) d\tau \bigg], \notag
 \end{align}
are uniformly integrable sequences along a subsequence $(n_{k_j})_{j \in \NN}$ such that 
$\big((r_{n_{k_j}}(\rho_{n_{k_j}}-\rho),r_{n_{k_j}}(\sigma_{n_{k_j}}-\sigma)\big) \rightarrow (L_1,L_2)$ in $\norm{\cdot}_1$ a.s. To see the required uniform integrability, observe that 
\begin{subequations}
\begin{align}
    \norm{\bar p_n(r_n,t)}_1 &\leq \norm{r_n^{\frac 12}(\rho_n-\rho)}_1^2 \bigg[\int_{0}^{\infty} \tau^{\alpha}\norm{v(\rho_n,\rho,\tau,t)}_{\infty}^3 d\tau \bigg] \norm{\big((1-t)\sigma+t\sigma_n\big)^{\bar \alpha}}_1, \label{eq:firstintui} \\
 \norm{\bar q_n(r_n,t)}_1 \mspace{-2 mu} &\leq \mspace{-2 mu} \norm{r_n^{\frac 12}(\sigma_n-\sigma)}_1^2 \bigg[\int_{0}^{\infty} \mspace{-6 mu}\tau^{\bar \alpha}\norm{v(\sigma_n,\sigma,\tau,t)}_{\infty}^3 d\tau \bigg] \norm{\big((1-t)\rho+t\rho_n\big)^{\alpha}}_1, \\
  \norm{\bar s_n(r_n,t)}_1 &\leq\norm{r_n^{\frac 12}(\rho_n-\rho)}_1\norm{r_n^{\frac 12}(\sigma_n-\sigma)}_1 \bigg[\int_{0}^{\infty} \tau^{\alpha}\norm{v(\rho_n,\rho,\tau,t)}_{\infty}^2 d\tau \bigg] \bigg[\int_{0}^{\infty} \tau^{\bar\alpha}\norm{v(\sigma_n,\sigma,\tau,t)}_{\infty}^2 d\tau \bigg].    
\end{align}\label{eq:bndintpetzui}
\end{subequations}
For $\alpha \in (0,1)$, we have by concavity of the function $x \mapsto x^{\alpha}$ for $x \geq 0$ that
  \begin{align}
    \norm{\big((1-t)\sigma+t\sigma_n\big)^{\bar \alpha}}_1=\sum_{i=1}^d \lambda_i^{\bar \alpha} \leq  d \left(\frac{1}{d}\sum_{i=1}^d \lambda_i\right)^{\bar \alpha} =d^{\alpha}  \norm{\big((1-t)\sigma+t\sigma_n\big)}_1^{\bar \alpha}  \leq d^{\alpha}.\notag 
\end{align}
Here,  $\{\lambda_i\}_{i=1}^d $ denotes the set of eigenvalues of $(1-t)\sigma+t\sigma_n$, and we used that for $A \geq 0$, the eigenvalues of $A^{\alpha}$ are  equal to the eigenvalues of $A$ raised to the power $\alpha$.   
Similarly, 
\begin{align}
     &   \norm{\big((1-t)\rho+t\rho_n\big)^{\alpha}}_1 \leq d^{\bar \alpha} \norm{(1-t)\rho+t\rho_n}_1^{\alpha} \leq  d^{\bar \alpha}.\label{eq:cocavityeig}
\end{align}
Next, since $\rho_{n_{k_j}} \rightarrow \rho$ and $\sigma_{n_{k_j}} \rightarrow \sigma$ in $\norm{\cdot}_1$ a.s., we have  that $(1-t)\rho+t\rho_{n_{k_j}} \geq c \rho$ and $(1-t)\sigma+t\sigma_{n_{k_j}} \geq c \sigma$ a.s. for a constant $0<c<1$  and  sufficiently large $j$.  Consequently,  we obtain  that the integrals in \eqref{eq:bndintpetzui} are finite  for $\alpha \in (0,1)$. For instance, the integrand   in \eqref{eq:firstintui} is $O(\tau^{\alpha})$ for $\tau$ close to zero and  $O(\tau^{\alpha-3})$ as $\tau \rightarrow \infty$ which implies its finiteness. Hence,
\begin{subequations}\label{eq:bochnerint}
   \begin{align} 
    &\norm{\bar p_{n_{k_j}}(r_{n_{k_j}},t)}_1 \lesssim_{d,\rho,\alpha}\norm{r_{n_{k_j}}^{\frac 12}(\rho_{n_{k_j}}-\rho)}_1^2,  \\
    &\norm{\bar q_{n_{k_j}}(r_{n_{k_j}},t)}_1 \lesssim_{d,\sigma,\alpha} \norm{r_{n_{k_j}}^{\frac 12}(\sigma_{n_{k_j}}-\sigma)}_1^2,  \\
    &\norm{\bar s_{n_{k_j}}(r_{n_{k_j}},t)}_1 \lesssim_{d,\rho,\sigma,\alpha} \norm{r_{n_{k_j}}^{\frac 12}(\rho_{n_{k_j}}-\rho)}_1 \norm{r_{n_{k_j}}^{\frac 12}(\sigma_{n_{k_j}}-\sigma)}_1.
\end{align} 
\end{subequations}
Then, the desired uniform integrability follows from those of the sequences $\big((r_{n_{k_j}}^{1/2}(\rho_{n_{k_j}}-\rho)\big)_{j \in \NN}$ and $\big(r_{n_{k_j}}^{1/2}(\sigma_{n_{k_j}}-\sigma)\big)_{j \in \NN}$. This completes the proof of the claim for $\alpha \in (0,1)$.

Next, consider the case $\alpha \in (1,2)$. Using the integral representation \cite[Lemma 2.8]{Carlen-2010}
\begin{subequations}
\begin{align}
   A^{\alpha}&=c_{\alpha-1} \int_0^{\infty} \left(\tau^{\alpha-2} A+\tau^{\alpha}(\tau I+A)^{-1}-\tau^{\alpha-1}I\right)d\tau, \label{intformfrac1} \\
   A^{\bar \alpha}&=c_{\bar \alpha+1} \int_0^{\infty}  \tau^{\bar \alpha}(\tau I+A)^{-1}d\tau, \label{intformfrac2}  
\end{align}
\end{subequations}
we have
\begin{align*}
&D^{(1,0)}[f(A_1,A_2)](H)=c_{\alpha-1} \int_{0}^{\infty}  \left(\tau^{\alpha-2}H-\tau^{\alpha}(\tau I+A_1)^{-1}H(\tau I+A_1)^{-1}\right)d\tau A_2^{\bar \alpha},  \\
&D^{(0,1)}[f(A_1,A_2)](H)=-c_{\bar \alpha+1} A_1^{\alpha} \int_{0}^{\infty} \tau^{\bar \alpha} (\tau I+A_2)^{-1} H (\tau I+A_2)^{-1}  d\tau,  \\
&D^{(2,0)}[f(A_1,A_2)](H_1,H_2)= c_{\alpha-1} \int_0^{\infty} \tau^{\alpha}  (\tau I+A_1)^{-1}H_1  (\tau I+A_1)^{-1}H_2  (\tau I+A_1)^{-1} d\tau A_2^{\bar \alpha} \\
& \qquad \qquad\qquad \qquad \qquad\qquad~~ +c_{\alpha-1} \int_0^{\infty} \tau^{\alpha}  (\tau I+A_1)^{-1}H_2  (\tau I+A_1)^{-1}H_1  (\tau I+A_1)^{-1} d\tau A_2^{\bar \alpha}, \\
&D^{(0,2)}[f(A_1,A_2)](H_1,H_2)= c_{\bar \alpha+1} A_1^{\alpha} \int_0^{\infty} \tau^{\bar \alpha}  (\tau I+A_2)^{-1}H_1  (\tau I+A_2)^{-1}H_2  (\tau I+A_2)^{-1} d\tau  \\
& \qquad \qquad\qquad \qquad \qquad\qquad ~~+c_{\bar \alpha+1} A_1^{\alpha} \int_0^{\infty} \tau^{\bar \alpha}  (\tau I+A_2)^{-1}H_2  (\tau I+A_2)^{-1}H_1  (\tau I+A_2)^{-1} d\tau,   \\
&D^{(1,1_+)}[f(A_1,A_2)](H_1,H_2)=-c_{\alpha-1}c_{\bar \alpha+1} \bigg[\int_{0}^{\infty}   \left(\tau^{\alpha-2}H_1-\tau^{\alpha}(\tau I+A_1)^{-1}H_1  (\tau I+A_1)^{-1} \right)d\tau \bigg]\notag \\
&\qquad \qquad \qquad \qquad \qquad\qquad \qquad\qquad \qquad\bigg[\int_{0}^{\infty} \tau^{\bar \alpha}  (\tau I+A_2)^{-1}H_2  (\tau I+A_2)^{-1} d\tau \bigg],  \\
&D^{(1_+,1)}[f(A_1,A_2)](H_1,H_2)=-c_{\bar \alpha+1}c_{\alpha-1} \bigg[\int_{0}^{\infty}  \left(\tau^{\alpha-2}H_2- \tau^{\alpha}(\tau I+A_1)^{-1}H_2  (\tau I+A_1)^{-1} \right) d\tau \bigg] \notag \\
&\qquad \qquad \qquad \qquad\qquad \qquad\qquad  \qquad \qquad\bigg[\int_{0}^{\infty} \tau^{\bar \alpha}  (\tau I+A_2)^{-1}H_1  (\tau I+A_2)^{-1} d\tau\bigg].
\end{align*}
Substituting $B_1=\rho_n$, $B_2=\sigma_n$, $A_1=\rho$, and $A_2=\sigma$, multiplying by $r_n$, taking trace and following similar arguments as above leads to the claim provided 
 \begin{align}
 \tilde p_n(r_n,t)&:= (1-t)\bigg[\int_0^{\infty} \tau^{\alpha}  v(\rho_n,\rho,\tau,t) r_n^{\frac 12}(\rho_n-\rho) v(\rho_n,\rho,\tau,t) r_n^{\frac 12}(\rho_n-\rho)   v(\rho_n,\rho,\tau,t) d\tau \bigg] \notag \\
 & \qquad \qquad \qquad\qquad \qquad\qquad \qquad \qquad\qquad \qquad \qquad\qquad \qquad \qquad   \qquad \big((1-t)\sigma+t\sigma_n\big)^{\bar \alpha},\notag \\
 \tilde q_n(r_n,t)&:= (1-t)\big((1-t)\rho+t\rho_n\big)^{\alpha}\bigg[\int_0^{\infty} \tau^{\bar \alpha}  v(\sigma_n,\sigma,\tau,t) r_n^{\frac 12}(\sigma_n-\sigma) v(\sigma_n,\sigma,\tau,t) r_n^{\frac 12}(\sigma_n-\sigma)  \notag \\
& \qquad \qquad \qquad\qquad \qquad\qquad \qquad \qquad\qquad \qquad \qquad\qquad \qquad \qquad   \qquad  v(\sigma_n,\sigma,\tau,t) d\tau\bigg],\notag \\
 \tilde s_n(r_n,t)&:=(1-t)\bigg[\int_{0}^{\infty}   \left(\tau^{\alpha-2}r_n^{\frac 12}(\rho_n-\rho)-\tau^{\alpha}v(\rho_n,\rho,\tau,t)~r_n^{\frac 12}(\rho_n-\rho) ~v(\rho_n,\rho,\tau,t) \right)d\tau \bigg]\notag \\
&\qquad \qquad \qquad \qquad \qquad\qquad \qquad\bigg[\int_{0}^{\infty} \tau^{\bar \alpha}  v(\sigma_n,\sigma,\tau,t)r_n^{\frac 12}(\sigma_n-\sigma) v(\sigma_n,\sigma,\tau,t) d\tau \bigg], \notag
 \end{align}
 are uniformly integrable  along a subsequence $(n_{k_j})_{j \in \NN}$ satisfying 
 $\big((r_{n_{k_j}}(\rho_{n_{k_j}}-\rho),r_{n_{k_j}}(\sigma_{n_{k_j}}-\sigma)\big) \rightarrow (L_1,L_2)$ in $\norm{\cdot}_1$ a.s. Observe that 
 \begin{subequations} \label{eq:unifbochpetz}
    \begin{align}
  \norm{\tilde p_n(r_n,t)}_1 &\lesssim_{d,\alpha} \norm{r_n^{\frac 12}(\rho_n-\rho)}_1^2\bigg[\int_{0}^{\infty}  \mspace{-2 mu}\tau^{\alpha}\norm{v(\rho_n,\rho,\tau,t)}_{\infty}^3 d\tau \bigg], \label{eq:unifbochpetz1} \\
  \norm{\tilde q_n(r_n,t)}_1 &\lesssim_{d,\alpha} \norm{r_n^{\frac 12}(\sigma_n-\sigma)}_1^2\bigg[\int_{0}^{\infty}\tau^{\bar \alpha} \norm{v(\sigma_n,\sigma,\tau,t)}_{\infty}^3 d\tau \bigg], \label{eq:unifbochpetz2}\\
  \norm{\tilde
  s_n(r_n,t)}_1 &\lesssim_{d,\alpha}\norm{r_n^{\frac 12}(\sigma_n-\sigma)}_1 \int_{0}^{\infty} \tau^{\bar \alpha}\norm{v(\sigma_n,\sigma,\tau,t)}_{\infty}^2 d\tau \notag \\
  &\qquad \bigg[\int_{0}^{\infty} \big\|\tau^{\alpha-2}r_n^{\frac 12}(\rho_n-\rho)-\tau^{\alpha}v(\rho_n,\rho,\tau,t)~r_n^{\frac 12}(\rho_n-\rho)  v(\rho_n,\rho,\tau,t)\big\|_1 d\tau \bigg]  . \label{eq:unifbochpetz3} 
 \end{align}  
 \end{subequations}
It is straightforward to see  the finiteness of the integrals in \eqref{eq:unifbochpetz}, except perhaps the last integral in  \eqref{eq:unifbochpetz3}. For the latter, observe that along  a  subsequence $(n_{k_j})_{j \in \NN}$ such that $\rho_{n_{k_j}} \rightarrow \rho$ in $\norm{\cdot}_1$  (a.s.), there exists constants $c,c'>0$ $\big($depending on $\rho$ and the realizations $\rho_{n_{k_j}}\big)$ such that 
\begin{align}
     c I  \leq A_{j,t}:=(1-t)\rho+t\rho_{n_{k_j}} \leq c' I. \notag
\end{align}
Put $V_{j,t}(\tau):=(\tau I+A_{j,t})^{-1}$ and $H_j:=r_{n_{k_j}}^{1/2}(\rho_{n_{k_j}}-\rho)$. Since $A_{j,t}$ commutes with $V_{j,t}(\tau)$ and $\tau V_{j,t}(\tau)=I-A_{j,t}V_{j,t}(\tau)$,
\begin{align}
&\tau^{\alpha-2}H_j-\tau^\alpha V_{j,t}(\tau)H_jV_{j,t}(\tau) =\tau^{\alpha-2}\Big(A_{j,t}V_{j,t}(\tau)H_j+H_jV_{j,t}(\tau)A_{j,t}
-A_{j,t}V_{j,t}(\tau)H_jV_{j,t}(\tau)A_{j,t}\Big). \notag
\end{align}
Consequently,
\begin{align}
&\big\|\tau^{\alpha-2}H_j-\tau^\alpha V_{j,t}(\tau)H_jV_{j,t}(\tau)\big\|_1 \notag\\
&\quad\leq \tau^{\alpha-2}\|H_j\|_1
\left(2\|A_{j,t}V_{j,t}(\tau)\|_\infty+\|A_{j,t}V_{j,t}(\tau)\|_\infty^2\right) \notag\\
&\quad\leq \tau^{\alpha-2}\|H_j\|_1
\left(\frac{2c'}{\tau+c}+\frac{(c')^2}{(\tau+c)^2}\right). \notag 
\end{align}
The scalar function on the last line is integrable on $(0,\infty)$ for $1<\alpha<2$: it is $O(\tau^{\alpha-2})$ at zero and $O(\tau^{\alpha-3})$ at infinity. This gives the required $j$-independent Bochner-integrable domination of the last factor in \eqref{eq:unifbochpetz3}; the other factors are treated as in the case $0<\alpha<1$. Thus the remainder terms vanish and \eqref{eq:petzrelent-twosample-alt} follows for $\alpha\in(1,2)$. Together with the preceding case this proves Part~(i) for $\alpha\in(0,1)\cup(1,2)$.   The case $\alpha=2$ is simpler as $f(A_1,A_2)=A_1^2A_2^{-1}$ and the relevant derivatives can be computed using the rules $D[A^2](H)=AH+HA$ and \eqref{eq:derinvop}. Since  rest of the proof is similar to above, we omit the details.  

\noindent
 \textbf{Two-sample null:} Consider $\alpha \in (0,1)$. We will show that 
\begin{align}
  r_n^2\big(Q_{\alpha}(\rho_n,\sigma_n)-Q_{\alpha}(\rho,\rho)\big) =r_n^2\big(Q_{\alpha}(\rho_n,\sigma_n)-1\big) \trightarrow{w} & \frac 12 \operatorname{Tr}\big[\rho^{\bar \alpha} D^2[\rho^{\alpha}](L_1,L_1)+\rho^{\alpha}D^2[\rho^{\bar \alpha}](L_2,L_2) \notag \\
  & \quad +2 D[\rho^{\alpha}](L_1)D[\rho^{\bar \alpha}](L_2)\big]. \label{eq:petzrenyi-null}
\end{align}
Then, an application of the functional  delta method yields the claim in \eqref{eq:petzrelent-twosample-null} by noting that $\log 1=0$. From \eqref{eq:taylor-op} with $B_1=\rho_n$, $B_2=\sigma_n$, $A_1=A_2=\rho$,  and $v(\rho_n,\rho,\tau,t):=\big(\tau I+(1-t)\rho+t\rho_n\big)^{-1}$, we obtain
\begin{align}
&\rho_n^{\alpha}\sigma_n^{\bar \alpha} \notag \\
&=\rho+ c_{\alpha} \int_0^{\infty} \mspace{-8 mu}\tau^{\alpha} (\tau I+\rho)^{-1}(\rho_n-\rho) (\tau I+\rho)^{-1} d\tau \rho^{\bar \alpha} + c_{\bar \alpha} \rho^{\alpha} \mspace{-5 mu}\int_0^{\infty} \mspace{-8 mu} \tau^{\bar \alpha} (\tau I+\rho)^{-1}(\sigma_n-\rho) (\tau I+\rho)^{-1} d\tau\notag \\
& \quad  \notag  -2c_{\alpha} \int_0^1 (1-t) \bigg[\int_0^{\infty} \tau^{\alpha} v(\rho_n,\rho,\tau,t) (\rho_n-\rho) v(\rho_n,\rho,\tau,t) (\rho_n-\rho) v(\rho_n,\rho,\tau,t) d\tau \bigg] \notag \\
&\qquad \qquad \qquad \qquad \qquad\qquad \qquad \qquad \qquad \qquad \qquad \qquad\qquad \qquad \big((1-t)\rho+t\sigma_n\big)^{\bar \alpha} dt\notag \\
& \quad    -2c_{\bar \alpha} \int_0^1 (1-t) \big((1-t)\rho+t\rho_n\big)^{ \alpha}\bigg[\int_0^{\infty} \tau^{\bar \alpha} v(\sigma_n,\rho,\tau,t)  (\sigma_n-\rho) v(\sigma_n,\rho,\tau,t) (\sigma_n-\rho)  \notag \\
&\qquad \qquad \qquad \qquad \qquad\qquad \qquad \qquad \qquad \qquad \qquad \qquad\qquad \qquad v(\sigma_n,\rho,\tau,t) d\tau \bigg] dt\notag \\
&+2c_{\alpha}c_{\bar \alpha} \int_0^1 (1-t)\bigg[\int_0^{\infty} \tau^{\alpha} v(\rho_n,\rho,\tau,t) (\rho_n-\rho) v(\rho_n,\rho,\tau,t) d\tau \bigg]\notag \\
& \qquad \qquad \qquad \qquad \qquad\qquad \qquad \qquad \qquad \bigg[\int_0^{\infty} \tau^{\bar \alpha} v(\sigma_n,\rho,\tau,t) (\sigma_n-\rho) v(\sigma_n,\rho,\tau,t) d\tau \bigg]  dt. \notag
\end{align}
Multiplying by $r_n^2$, taking trace and following similar arguments as above leads to \eqref{eq:petzrenyi-null}, provided 
\begin{subequations} \label{eq:inttrzero}
    \begin{align}
&  \tr{ \int_0^{\infty} \tau^{\alpha} (\tau I+\rho)^{-1}(\rho_n-\rho) (\tau I+\rho)^{-1} d\tau \rho^{\bar \alpha}}=0, \label{eq:inttrzero1}\\
  & \tr{\rho^{\alpha} \int_0^{\infty} \tau^{\bar \alpha} (\tau I+\rho)^{-1}(\sigma_n-\rho) (\tau I+\rho)^{-1} d\tau}=0,\label{eq:inttrzero2}
\end{align}
\end{subequations}
and the sequences
 \begin{align}
 &\bar p_n(r_n,t)\notag \\
 &:= (1-t)\bigg[\int_0^{\infty} \mspace{-5 mu}\tau^{\alpha}  v(\rho_n,\rho,\tau,t) r_n(\rho_n-\rho) v(\rho_n,\rho,\tau,t) r_n(\rho_n-\rho)   v(\rho_n,\rho,\tau,t) d\tau \bigg]\mspace{-2 mu}\big((1-t)\rho+t\sigma_n\big)^{\bar \alpha},\notag \\
& \tilde q_n(r_n,t)\notag \\
&:= (1-t)\big((1-t)\rho+t\rho_n\big)^{\alpha}\mspace{-2 mu}\bigg[\int_0^{\infty} \mspace{-5 mu}\tau^{\bar \alpha}  v(\sigma_n,\rho,\tau,t) r_n(\sigma_n-\rho) v(\sigma_n,\rho,\tau,t) r_n(\sigma_n-\rho) v(\sigma_n,\rho,\tau,t) d\tau\bigg],\notag \\
 &\tilde s_n(r_n,t):=(1-t)\bigg[\int_0^{\infty} \tau^{\alpha} v(\rho_n,\rho,\tau,t) (\rho_n-\rho) v(\rho_n,\rho,\tau,t) d\tau \bigg]\notag \\
&\qquad \qquad \qquad \qquad \qquad\qquad \qquad \qquad \qquad \quad  \qquad\bigg[\int_{0}^{\infty} \tau^{\bar \alpha}  v(\sigma_n,\rho,\tau,t)r_n(\sigma_n-\rho) v(\sigma_n,\rho,\tau,t) d\tau \bigg],\notag
 \end{align}
 are uniformly integrable along a subsequence  $(n_{k_j})_{j \in \NN}$ satisfying 
  $\big((r_{n_{k_j}}(\rho_{n_{k_j}}-\rho),r_{n_{k_j}}(\sigma_{n_{k_j}}-\rho)\big) \rightarrow (L_1,L_2)$  in $\norm{\cdot}_1$ a.s. To show \eqref{eq:inttrzero1}, note that  we have
 \begin{align}
  \tr{ \int_0^{\infty} \tau^{\alpha} (\tau I+\rho)^{-1}(\rho_n-\rho) (\tau I+\rho)^{-1} d\tau \rho^{\bar \alpha}}&\stackrel{(a)}{=} \tr{(\rho_n-\rho)\rho^{\bar \alpha}\int_0^{\infty} \tau^{\alpha} (\tau I+\rho)^{-2} d\tau } \notag \\
  & \stackrel{(b)}{=}\alpha c_{\alpha}^{-1}\tr{\rho_n-\rho}=0, \notag
 \end{align}
 where $(a)$ follows by similar arguments  leading to \eqref{eq:firstordtmzero} and $(b)$ is due to 
 \begin{align}
     \int_0^{\infty} \tau^{\alpha} (\tau I+\rho)^{-2} d\tau &= -\tau^{\alpha}(\tau I+\rho)^{-1} \big|_0^{\infty}-\alpha    \int_0^{\infty} \tau^{\alpha-1}(\tau I+\rho)^{-1} d\tau = \alpha  c_{\alpha}^{-1}\rho^{-\bar \alpha}.\notag 
 \end{align}
 In the above, the first equality uses integration by parts and  the second equality uses \eqref{intformfrac2}. 
 Similarly, \eqref{eq:inttrzero2} also holds. The desired integrability can be shown similar to \eqref{eq:bochnerint}. This completes the proof of \eqref{eq:petzrelent-twosample-null} for the case $\alpha \in (0,1)$. The proof when $\alpha \in (1,2]$ is similar and hence omitted. 

\medskip

 If $[\rho_n,\rho]=[\sigma_n,\rho]=[\sigma_n,\rho_n]=0$ for all $n$ sufficiently large, then $[L_1,L_2]=[L_1,\rho]=[L_2,\rho]=0$. The scalar derivative identities give
 \begin{align}
   \rho^{\bar \alpha}D^2[\rho^{\alpha}](L_1,L_1)&=\alpha(\alpha-1)L_1^2\rho^{-1}, \notag\\
   \rho^{\alpha}D^2[\rho^{\bar \alpha}](L_2,L_2)&=\alpha(\alpha-1)L_2^2\rho^{-1}, \notag\\
   2D[\rho^{\alpha}](L_1)D[\rho^{\bar \alpha}](L_2)&=-2\alpha(\alpha-1)L_1L_2\rho^{-1}. \notag
 \end{align}
 Substituting this in \eqref{eq:petzrelent-twosample-null} leads to \eqref{eq:petzrelent-twosample-null-comm}.

 \subsection{Proof of Theorem \ref{Thm:sandrelent-limdist}} \label{Thm:sandrelent-limdist-proof} 
 We will prove the claim for $\rho, \sigma>0$. The general case $\rho \ll \sigma$ can be handled via similar arguments as in the proof of Theorem \ref{Thm:quantrelent-limdist}. 
 We  begin by noting that the following variational form 
holds for the sandwiched R\'{e}nyi  divergence as given in \cite{muller2013quantum} (see \cite{Berta-2018} for a further application of this form in generalizing sandwiched R\'{e}nyi  divergence to infinite dimensional quantum settings) for $\alpha \in (1,\infty)$:
 \begin{align}
 \sanddiv{\rho}{\sigma}{\alpha}&:=\frac{\alpha}{\alpha-1} \log \norm{\rho^{\frac 12 }\sigma^{\frac{\bar \alpha}{\alpha}}\rho^{\frac 12 }}_{\alpha} =\max_{\substack{\eta \geq 0, \\ \norm{\eta}_{\frac{\alpha}{\alpha-1} \leq 1}} } \sanddivvar{\rho}{\sigma}{\alpha}{\eta}, \label{eq:varexpsr}
 \end{align}
 where 
 \begin{align}
     \sanddivvar{\rho}{\sigma}{\alpha}{\eta}:=\big(\alpha/(\alpha-1)\big) \log \tr{\rho^{\frac 12 }\sigma^{\frac{\bar \alpha}{\alpha}}\rho^{\frac 12 } \eta}.\notag
 \end{align}
 For $0.5<\alpha<1$, the conjugate exponent $q=\alpha/(\alpha-1)$ is negative. In this case, $\|\eta\|_q:=(\operatorname{Tr}\eta^q)^{1/q}$ for $\eta>0$ denotes the Schatten $q$ anti-norm, and  a similar variational formula \eqref{eq:varexpsr} holds with the maximization over $\eta>0$ such that  $\norm{\eta}_{\frac{\alpha}{\alpha-1}}=1$.  The endpoint $\alpha=1/2$ is treated separately below.  
  The maximum in  \eqref{eq:varexpsr} in both cases  is achieved by 
  \begin{align}
 \eta=\eta^{\star}(\rho,\sigma,\alpha):=\big(\rho^{\frac 12 }\sigma^{\frac{\bar \alpha}{\alpha}}\rho^{\frac 12 }\big)^{\alpha-1}/\big\|\rho^{\frac 12 }\sigma^{\frac{\bar \alpha}{\alpha}}\rho^{\frac 12 }\big\|_{\alpha}^{\alpha-1}. \label{eq:optvarform}     
  \end{align}  
Defining $\eta_n^{\star}:=\eta^{\star}(\rho_n,\sigma_n,\alpha)$, we have
\begin{subequations}\label{eq:sandrenybnd}
 \begin{align}
   r_n \big( \sanddiv{\rho_n}{\sigma_n}{\alpha} -\sanddiv{\rho}{\sigma}{\alpha}\big) &\leq r_n  \big( \sanddivvar{\rho_n}{\sigma_n}{\alpha}{\eta_n^{\star}} -\sanddivvar{\rho}{\sigma}{\alpha}{\eta_n^{\star}}\big), \label{eq:sandrenyub}\\
     r_n \big( \sanddiv{\rho_n}{\sigma_n}{\alpha} -\sanddiv{\rho}{\sigma}{\alpha}\big) &\geq r_n  \big( \sanddivvar{\rho_n}{\sigma_n}{\alpha}{\eta^{\star}} -\sanddivvar{\rho}{\sigma}{\alpha}{\eta^{\star}}\big).\label{eq:sandrenylb}
 \end{align}
 \end{subequations}
We will show that the limit on the RHS of \eqref{eq:sandrenyub} and \eqref{eq:sandrenylb}  coincides with the RHS of \eqref{eq:sandrelent-twosample-alt}, thus proving the desired claim. 

To establish that the former limit equals the RHS of \eqref{eq:sandrelent-twosample-alt}, it is sufficient to show that
 \begin{align}
   & r_n \left(\tr{\rho_n^{\frac 12}\sigma_n^{\frac{\bar \alpha}{\alpha}} \rho_n^{\frac 12} \eta_n^{\star}}-\tr{\rho^{\frac 12}\sigma^{\frac{\bar \alpha}{\alpha}} \rho^{\frac 12}\eta_n^{\star}}\right)  \notag \\
   & \qquad \trightarrow{w} \frac{\tr{\Big(D[\rho^{\frac 12}](L_1) \sigma^{\frac{\bar \alpha}{\alpha}} \rho^{\frac 12}+ \rho^{\frac 12} \sigma^{\frac{\bar \alpha}{\alpha}} D[\rho^{\frac 12}](L_1)  +\rho^{\frac 12} D[\sigma^{\frac{\bar \alpha}{\alpha}}](L_2) \rho^{\frac 12}\Big) \big(\rho^{\frac 12} \sigma^{\frac{\bar \alpha}{\alpha}}\rho^{\frac 12} \big)^{\alpha-1}}}{\big\|\rho^{\frac 12 }\sigma^{\frac{\bar \alpha}{\alpha}}\rho^{\frac 12 }\big\|_{\alpha}^{\alpha-1}}. \label{eq:sandwichreninterm} 
 \end{align}
 Then,  the functional delta-method (specifically \cite[Theorem 3.9.5 and Lemma 3.9.7]{AVDV-book})  applied to the continuously differentiable function $x \mapsto \alpha \log x/(\alpha-1)$ at $x=\norm{\rho^{\frac 12 }\sigma^{\frac{\bar \alpha}{\alpha}}\rho^{\frac 12 }}_{\alpha}>0$  yields 
\begin{align}
  &  r_n  \big( \sanddivvar{\rho_n}{\sigma_n}{\alpha}{\eta_n^{\star}} -\sanddivvar{\rho}{\sigma}{\alpha}{\eta_n^{\star}}\big) \notag \\
    &\trightarrow{w} \frac{\alpha}{\alpha-1} \frac{\tr{\Big(D[\rho^{\frac 12}](L_1) \sigma^{\frac{\bar \alpha}{\alpha}} \rho^{\frac 12}+ \rho^{\frac 12} \sigma^{\frac{\bar \alpha}{\alpha}} D[\rho^{\frac 12}](L_1)  +\rho^{\frac 12} D[\sigma^{\frac{\bar \alpha}{\alpha}}](L_2) \rho^{\frac 12}\Big) \big(\rho^{\frac 12} \sigma^{\frac{\bar \alpha}{\alpha}}\rho^{\frac 12} \big)^{\alpha-1}}}{\big\|\rho^{\frac 12 }\sigma^{\frac{\bar \alpha}{\alpha}}\rho^{\frac 12 }\big\|_{\alpha}^{\alpha}}, \notag 
\end{align}
as desired. Further, invoking the  subsequence argument, it  suffices to prove \eqref{eq:sandwichreninterm} along a subsequence $(n_{k_j})_{j \in \NN}$ such that $\big((r_{n_{k_j}}(\rho_{n_{k_j}}-\rho),r_{n_{k_j}}(\sigma_{n_{k_j}}-\sigma)\big) \rightarrow (L_1,L_2)$ in trace-norm  a.s. 

 To show \eqref{eq:sandwichreninterm}, we  compute the Fr\'{e}chet derivatives of 
 \begin{align}  f(A_1,A_2)=A_1^{\frac 12}A_2^{\frac{\bar \alpha}{\alpha}} A_1^{\frac 12}. \notag
 \end{align} Consider $\alpha \in (0.5,1) \cup (1,\infty)$. Since $\bar \alpha/\alpha \in (0,1) \cup (-1,0)$,  the following integral representations given in \cite[Lemma 2.8]{Carlen-2010} are relevant for our purpose:  
\begin{align}
A^{\alpha}&=c_{\alpha}\int_{0}^{\infty} \tau ^{\alpha} \left(\frac{1}{\tau I}-\frac{1}{\tau I+A}\right)d\tau,~~ \alpha \in (0,1),\notag \\
A^{\alpha}&=c_{\alpha+1}\int_{0}^{\infty}  \frac{\tau ^{\alpha}}{\tau I+A} d\tau, ~~\alpha \in (-1,0). \notag
\end{align}
 We have via the chain and product rule for Fr\'{e}chet derivatives that
\begin{align*}
&D^{(1,0)}[f(A_1,A_2)](H)=c_{\frac 12} \int_{0}^{\infty} \tau^{\frac 12} (\tau I+A_1)^{-1} H (\tau I+A_1)^{-1}  d\tau A_2^{\frac{\bar \alpha}{\alpha}}A_1^{\frac 12} \notag \\
&\qquad \qquad \qquad \qquad \qquad \qquad \qquad +c_{\frac 12} A_1^{\frac 12} A_2^{\frac{\bar \alpha}{\alpha}}
\int_{0}^{\infty} \tau^{\frac 12} (\tau I+A_1)^{-1} H (\tau I+A_1)^{-1}  d\tau
,  \\
&D^{(0,1)}[f(A_1,A_2)](H)=-c_{\frac{\bar \alpha}{\alpha}+1} A_1^{\frac 12} \left[\int_{0}^{\infty} \tau^{\frac{\bar \alpha}{\alpha}} (\tau I+A_2)^{-1} H (\tau I+A_2)^{-1}  d\tau \right] A_1^{\frac 12},  \\
&D^{(2,0)}[f(A_1,A_2)](H_1,H_2)= -c_{\frac 12} \bigg[\int_0^{\infty} \tau^{\frac 12}  (\tau I+A_1)^{-1}H_1  (\tau I+A_1)^{-1}H_2  (\tau I+A_1)^{-1} d\tau  \\
& \qquad \qquad\qquad \qquad \qquad\qquad \quad + \int_0^{\infty} \tau^{\frac 12}  (\tau I+A_1)^{-1}H_2  (\tau I+A_1)^{-1}H_1  (\tau I+A_1)^{-1} d\tau \bigg] A_2^{\frac{\bar \alpha}{\alpha}} A_1^{\frac 12}  \\
&\qquad \qquad\qquad \qquad \qquad\qquad \quad -c_{\frac 12} A_1^{\frac 12}A_2^{\frac{\bar \alpha}{\alpha}} \bigg[\int_0^{\infty} \tau^{\frac 12}  (\tau I+A_1)^{-1}H_1  (\tau I+A_1)^{-1}H_2  (\tau I+A_1)^{-1} d\tau  \notag \\
&\qquad \qquad\qquad \qquad \qquad\qquad \quad + \int_0^{\infty} \tau^{\frac 12}  (\tau I+A_1)^{-1}H_2  (\tau I+A_1)^{-1}H_1  (\tau I+A_1)^{-1} d\tau \bigg] \notag \\
&\qquad \qquad\qquad \qquad \qquad\qquad \quad  +c_{\frac 12}^2 \bigg[\int_0^{\infty} \tau^{\frac 12}  (\tau I+A_1)^{-1}H_1  (\tau I+A_1)^{-1} d\tau \bigg] A_2^{\frac{\bar \alpha}{\alpha}} \notag \\
& \qquad \qquad\qquad \qquad \qquad\qquad \qquad \quad  \bigg[ \int_0^{\infty} \tau^{\frac 12}  (\tau I+A_1)^{-1}H_2  (\tau I+A_1)^{-1} d\tau\bigg] \notag \\
&\qquad \qquad\qquad \qquad \qquad\qquad \quad  +c_{\frac 12}^2 \bigg[\int_0^{\infty} \tau^{\frac 12}  (\tau I+A_1)^{-1}H_2  (\tau I+A_1)^{-1} d\tau \bigg] A_2^{\frac{\bar \alpha}{\alpha}} \notag \\
& \qquad \qquad\qquad \qquad \qquad\qquad \qquad \quad  \bigg[ \int_0^{\infty} \tau^{\frac 12}  (\tau I+A_1)^{-1}H_1  (\tau I+A_1)^{-1} d\tau\bigg], \notag \\
&D^{(0,2)}[f(A_1,A_2)](H_1,H_2)= c_{\frac{\bar \alpha}{\alpha}+1} A_1^{\frac 12} \bigg[\int_0^{\infty} \tau^{\frac{\bar \alpha}{\alpha}}  (\tau I+A_2)^{-1}H_1  (\tau I+A_2)^{-1}H_2  (\tau I+A_2)^{-1} d\tau  \\
& \qquad \qquad\qquad \qquad \qquad\qquad \quad   + \int_0^{\infty} \tau^{\frac{\bar \alpha}{\alpha}}  (\tau I+A_2)^{-1}H_2  (\tau I+A_2)^{-1}H_1  (\tau I+A_2)^{-1} d\tau \bigg] A_1^{\frac 12} ,   \\
&D^{(1,1_+)}[f(A_1,A_2)](H_1,H_2)=-c_{\frac 12}  c_{\frac{\bar \alpha}{\alpha}+1}\bigg[\int_{0}^{\infty} \tau^{\frac 12}  (\tau I+A_1)^{-1}H_1  (\tau I+A_1)^{-1} d\tau \bigg]  \notag \\
&\qquad \qquad \qquad \qquad\qquad \qquad\qquad \qquad \bigg[\int_{0}^{\infty} \tau^{\frac{\bar \alpha}{\alpha}}  (\tau I+A_2)^{-1}H_2  (\tau I+A_2)^{-1} d\tau \bigg] A_1^{\frac 12}  \notag \\
&\qquad \qquad \qquad \qquad\qquad \qquad\qquad  -c_{\frac 12}  c_{\frac{\bar \alpha}{\alpha}+1} A_1^{\frac 12} \bigg[\int_{0}^{\infty} \tau^{\frac{\bar \alpha}{\alpha}}  (\tau I+A_2)^{-1}H_2  (\tau I+A_2)^{-1} d\tau \bigg]  \notag \\
&\qquad \qquad \qquad \qquad\qquad \qquad\qquad \qquad \bigg[\int_{0}^{\infty} \tau^{\frac 12}  (\tau I+A_1)^{-1}H_1  (\tau I+A_1)^{-1} d\tau \bigg] ,  \\
&D^{(1_+,1)}[f(A_1,A_2)](H_1,H_2)=-c_{\frac 12}  c_{\frac{\bar \alpha}{\alpha}+1}\bigg[\int_{0}^{\infty} \tau^{\frac 12}  (\tau I+A_1)^{-1}H_2  (\tau I+A_1)^{-1} d\tau \bigg]  \notag \\
&\qquad \qquad \qquad \qquad\qquad \qquad\qquad \qquad \bigg[\int_{0}^{\infty} \tau^{\frac{\bar \alpha}{\alpha}}  (\tau I+A_2)^{-1}H_1  (\tau I+A_2)^{-1} d\tau \bigg] A_1^{\frac 12}  \notag \\
&\qquad \qquad \qquad \qquad\qquad \qquad\qquad  -c_{\frac 12}  c_{\frac{\bar \alpha}{\alpha}+1} A_1^{\frac 12} \bigg[\int_{0}^{\infty} \tau^{\frac{\bar \alpha}{\alpha}}  (\tau I+A_2)^{-1}H_1  (\tau I+A_2)^{-1} d\tau \bigg]  \notag \\
&\qquad \qquad \qquad \qquad\qquad \qquad\qquad \qquad \bigg[\int_{0}^{\infty} \tau^{\frac 12}  (\tau I+A_1)^{-1}H_2  (\tau I+A_1)^{-1} d\tau \bigg] .
\end{align*}

\medskip
 
Then,  substituting the  expressions for Fr\'{e}chet derivatives derived above in \eqref{eq:taylor-op} with $B_1=\rho_n$, $B_2=\sigma_n$, $A_1=\rho$, and $A_2=\sigma$ leads to
\begin{align}
\rho_n^{\frac 12}\sigma_n^{\frac{\bar \alpha}{\alpha}} \rho_n^{\frac 12} \eta_n^{\star}&=\rho^{\frac 12}\sigma^{\frac{\bar \alpha}{\alpha}} \rho^{\frac 12}\eta_n^{\star} + c_{\frac 12} \bigg[\int_{0}^{\infty} \tau^{\frac 12} (\tau I+\rho)^{-1} (\rho_n-\rho) (\tau I+\rho)^{-1}  d\tau \bigg]\sigma^{\frac{\bar \alpha}{\alpha}}\rho^{\frac 12} \eta_n^{\star} \notag \\
&\quad+c_{\frac 12} \rho^{\frac 12} \sigma^{\frac{\bar \alpha}{\alpha}} \bigg[\int_{0}^{\infty} \tau^{\frac 12}  (\tau I+\rho)^{-1} (\rho_n-\rho) (\tau I+\rho)^{-1}  d\tau \bigg]  \eta_n^{\star} \notag \\
&\quad-c_{\frac{1}{\alpha}} \rho^{\frac 12} \bigg[\int_{0}^{\infty} \tau^{\frac{\bar \alpha}{\alpha}} (\tau I+\sigma)^{-1} (\sigma_n-\sigma)   (\tau I+\sigma)^{-1}  d\tau \bigg] \rho^{\frac 12} \eta_n^{\star} +R_n\eta_n^{\star}, \label{eq:taylorexpsandren}
\end{align}
where $R_n:=R_{1,n}+R_{2,n}+R_{3,n}$, and with $v(\rho_n,\rho,\tau,t):=\big(\tau I+(1-t)\rho+t\rho_n\big)^{-1}$, 
\begin{align}
  R_{1,n} 
 &:=-2c_{\frac 12} \int_0^1 (1-t)
   \int_0^{\infty} \tau^{\frac 12}  v(\rho_n,\rho,\tau,t)(\rho_n-\rho)  v(\rho_n,\rho,\tau,t)(\rho_n-\rho)  v(\rho_n,\rho,\tau,t) d\tau 
    \notag \\
&\qquad \qquad \qquad \qquad\qquad \qquad\qquad \qquad \big( (1-t)\sigma+t\sigma_n\big)^{\frac{\bar \alpha}{\alpha}} \big( (1-t)\rho+t\rho_n\big)^{\frac 12}   dt \notag \\
& -2c_{\frac 12} \mspace{-2 mu}\int_0^1 \mspace{-2 mu}(1-t)
   \big( (1-t)\rho+t\rho_n\big)^{\frac 12}\mspace{-2 mu} \big( (1-t)\sigma+t\sigma_n\big)^{\frac{\bar \alpha}{\alpha}}  \notag \\
   &\qquad \qquad\qquad \qquad\qquad \qquad\int_0^{\infty} \mspace{-2 mu}\tau^{\frac 12}  v(\rho_n,\rho,\tau,t)(\rho_n-\rho)  v(\rho_n,\rho,\tau,t)(\rho_n-\rho) v(\rho_n,\rho,\tau,t) d\tau dt  \notag \\
   & +4c_{\frac 12}^2 \int_0^1 (1-t) \bigg[\int_0^{\infty} \tau^{\frac 12}  v(\rho_n,\rho,\tau,t) (\rho_n-\rho) v(\rho_n,\rho,\tau,t) d\tau \bigg] \big( (1-t)\sigma+t\sigma_n\big)^{\frac{\bar \alpha}{\alpha}} \notag \\
   &\qquad \qquad\qquad \qquad\qquad \qquad \qquad \qquad\qquad \bigg[ \int_0^{\infty} \tau^{\frac 12}  v(\rho_n,\rho,\tau,t)(\rho_n-\rho) v(\rho_n,\rho,\tau,t) d\tau\bigg] dt,\notag \\
 R_{2,n}  &:=2c_{\frac{1}{\alpha}} \int_0^1 (1-t) \big( (1-t)\rho+t\rho_n\big)^{\frac 12}  \bigg[\int_0^{\infty} \tau^{\frac{\bar \alpha}{\alpha}}   v(\sigma_n,\sigma,\tau,t)(\sigma_n-\sigma)  v(\sigma_n,\sigma,\tau,t)(\sigma_n-\sigma)    \notag \\
 & \qquad \qquad\qquad \qquad\qquad \qquad  \qquad \qquad\qquad \qquad\qquad \qquad v(\sigma_n,\sigma,\tau,t) d\tau \bigg] \big( (1-t)\rho+t\rho_n\big)^{\frac 12} dt, \notag \\
  R_{3,n}  &:=-2c_{\frac{1}{\alpha}} c_{\frac 12} \int_0^1 (1-t) \bigg[\int_{0}^{\infty} \tau^{\frac 12}  v(\rho_n,\rho,\tau,t)(\rho_n-\rho)  v(\rho_n,\rho,\tau,t) d\tau \bigg] \notag \\
  &\qquad \qquad\qquad \qquad\qquad \quad\bigg[\int_{0}^{\infty} \tau^{\frac{\bar \alpha}{\alpha}} v(\sigma_n,\sigma,\tau,t)  (\sigma_n-\sigma) v(\sigma_n,\sigma,\tau,t) d\tau \bigg] \big( (1-t)\rho+t\rho_n\big)^{\frac 12} dt \notag \\
& \quad -2c_{\frac{1}{\alpha}} c_{\frac 12} \int_0^1 \mspace{-4 mu}(1-t) \big((1-t)\rho+t\rho_n\big)^{\frac 12} \bigg[\int_{0}^{\infty} \mspace{-4 mu}\tau^{\frac 12}  v(\rho_n,\rho,\tau,t)(\rho_n-\rho)  v(\rho_n,\rho,\tau,t) d\tau \bigg] \notag \\
& \qquad \qquad\qquad \qquad\qquad \qquad \qquad\qquad \quad\qquad\bigg[\int_{0}^{\infty} \tau^{\frac{\bar \alpha}{\alpha}} v(\sigma_n,\sigma,\tau,t)  (\sigma_n-\sigma) v(\sigma_n,\sigma,\tau,t) d\tau \bigg] dt.  \notag 
\end{align} 
Multiplying by $r_n$ and taking limits, the desired claim will follow provided $\tr{r_nR_n \eta_n^{\star}} \rightarrow 0$ along the subsequence $(n_{k_j})_{j \in \NN}$ mentioned above. To show this, we require an interchange of limits and integral which can be justified via the uniform integrability conditions stated below:
\begin{align}
&\norm{r_nR_{1,n} \eta_n^{\star}}_1 \notag \\
&\lesssim_{d,\alpha}  \norm{r_n^{\frac 12}(\rho_n-\rho)}_1^2
   \bigg[\int_0^{\infty} \tau^{\frac 12}  \norm{v(\rho_n,\rho,\tau,t)}_{\infty}^3  d\tau \bigg]
  \norm{\big((1-t)\sigma+t\sigma_n\big)^{\frac{\bar \alpha}{\alpha}}}_1  \norm{(1-t)\rho+t\rho_n}_1^{\frac 12} \norm{\eta_n^{\star}}_1 \notag \\
  &\qquad \qquad\qquad\qquad+ \norm{r_n^{\frac 12}(\rho_n-\rho)}_1^2 \norm{\big((1-t)\sigma+t\sigma_n\big)^{\frac{\bar \alpha}{\alpha}}}_1    \bigg[\int_0^{\infty} \tau^{\frac 12}  \norm{v(\rho_n,\rho,\tau,t)}_{\infty}^2  d\tau \bigg]^2 \norm{\eta_n^{\star}}_1, \notag \\
 & \norm{r_nR_{2,n} \eta_n^{\star}}_1 \lesssim_{d,\alpha}\mspace{-4 mu} \norm{(1-t)\rho+t\rho_n}_1  \norm{r_n^{\frac 12}(\sigma_n-\sigma)}_1^2  \bigg[\int_0^{\infty} \mspace{-4 mu}\tau^{^{\frac{\bar \alpha}{\alpha}}}  \norm{v(\sigma_n,\sigma,\tau,t)}_{\infty}^3  d\tau \bigg] \norm{\eta_n^{\star}}_1,\notag \\
 &\norm{r_nR_{3,n} \eta_n^{\star}}_1 \lesssim_{d,\alpha} \norm{r_n^{\frac 12}(\rho_n-\rho)}_1 \norm{r_n^{\frac 12}(\sigma_n-\sigma)}_1 \norm{(1-t)\rho+t\rho_n}_1^{\frac 12}  \bigg[\int_0^{\infty} \tau^{\frac 12}  \norm{v(\rho_n,\rho,\tau,t)}_{\infty}^2  d\tau \bigg] \notag \\
 & \qquad \qquad \qquad\qquad\qquad\qquad \qquad\qquad\qquad\qquad \qquad\qquad\qquad\bigg[\int_0^{\infty} \tau^{\frac{\bar \alpha}{\alpha}}  \norm{v(\sigma_n,\sigma,\tau,t)}_{\infty}^2  d\tau \bigg] \norm{\eta_n^{\star}}_1.\notag
\end{align}
To obtain these bounds, we use H\"older's inequality as in \eqref{eq:applyHolderineq}. If $\alpha \in  (0.5,1)$, then $\bar\alpha/\alpha\in(0,1)$ and operator concavity of the map $x \mapsto x^{\beta}$ for $\beta \in (0,1]$ gives a uniform bound on $\|((1-t)\sigma+t\sigma_n)^{\bar\alpha/\alpha}\|_1$. If $\alpha \in (1,\infty)$, $\bar\alpha/\alpha$ is negative.  Since $\sigma>0$ and $\sigma_n\to\sigma$ in operator norm along the subsequence, for all sufficiently large $n$ one has $((1-t)\sigma+t\sigma_n)\geq c_\sigma I$ uniformly in $t\in[0,1]$ for some $c_\sigma>0$. Consequently its negative powers have uniformly bounded trace norm. Also, $\rho_{n_{k_j}}\to\rho$ and $\sigma_{n_{k_j}}\to\sigma$ a.s. in trace norm and $\rho,\sigma>0$  implies $\eta^{\star}_{n_{k_j}}\to\eta^{\star}$ a.s. in trace norm. 
Hence, analogous to \eqref{eq:bochnerint},  the desired  $j$-independent integrable domination holds a.s. along the subsequence $(n_{k_j})_{j \in \NN}$:
\begin{align}
\norm{r_{n_{k_j}}R_{1,n_{k_j}} \eta_{n_{k_j}}^{\star}}_1 &\lesssim_{d, \rho,\sigma,\alpha} \norm{r_{n_{k_j}}^{\frac 12}(\rho_{n_{k_j}}-\rho)}_1^2, \notag \\
\norm{r_{n_{k_j}}R_{2,n_{k_j}} \eta_{n_{k_j}}^{\star}}_1 &\lesssim_{d, \rho,\sigma,\alpha} \norm{r_{n_{k_j}}^{\frac 12}(\rho_{n_{k_j}}-\rho)}_1^2, \notag \\
\norm{r_{n_{k_j}}R_{3,n_{k_j}} \eta_{n_{k_j}}^{\star}}_1 &\lesssim_{d, \rho,\sigma,\alpha} \norm{r_{n_{k_j}}^{\frac 12}(\rho_{n_{k_j}}-\rho)}_1 \norm{r_{n_{k_j}}^{\frac 12}(\sigma_{n_{k_j}}-\sigma)}_1, \notag
\end{align}
Moreover, we have  
$r_n \big(\operatorname{Tr}\big[\rho_n^{\frac 12}\sigma_n^{\frac{\bar \alpha}{\alpha}} \rho_n^{\frac 12} \eta_n^{\star}\big]-\operatorname{Tr}\big[\rho^{\frac 12}\sigma^{\frac{\bar \alpha}{\alpha}} \rho^{\frac 12} \eta_n^{\star}\big]\big)$ converges to the RHS of \eqref{eq:sandwichreninterm} along  $(n_{k_j})_{j \in \NN}$. In a similar vein,  $r_n \big(\operatorname{Tr}\big[\rho_n^{\frac 12}\sigma_n^{\frac{\bar \alpha}{\alpha}} \rho_n^{\frac 12} \eta^{\star}\big]-\operatorname{Tr}\big[\rho^{\frac 12}\sigma^{\frac{\bar \alpha}{\alpha}} \rho^{\frac 12}\eta^{\star}\big]\big) $ converges a.s. to RHS of \eqref{eq:sandwichreninterm} along the same subsequence, thus proving \eqref{eq:sandwichreninterm} and completing the proof of \eqref{eq:sandrelent-twosample-alt} when $\alpha \in (0.5,1) \cup (1,\infty)$. The proof for $\alpha=0.5$ can be shown similar to above by applying Taylor's theorem to $f(A_1,A_2)=A_1^{\frac 12}A_2A_1^{\frac 12}$, for which the Fr\'{e}chet derivatives are relatively easier to compute.

\medskip

To see  that \eqref{eq:sandrelent-twosample-alt} simplifies to \eqref{eq:petzrelent-twosample-alt-comm} when all relevant operators commute, observe that $D[\rho^{\frac 12}](L_1)=L_1 \rho^{\frac 12}/2$ and $D[\sigma^{\frac{\bar \alpha}{\alpha}}](L_2)=(\bar \alpha/\alpha) L_2 \sigma^{\frac{\bar \alpha}{\alpha}-1}$, and hence
\begin{align}
    \Big(D[\rho^{\frac 12}](L_1) \sigma^{\frac{\bar \alpha}{\alpha}} \rho^{\frac 12}+ \rho^{\frac 12} \sigma^{\frac{\bar \alpha}{\alpha}} D[\rho^{\frac 12}](L_1)  +\rho^{\frac 12} D[\sigma^{\frac{\bar \alpha}{\alpha}}](L_2) \rho^{\frac 12}\Big) \big(\rho^{\frac 12} \sigma^{\frac{\bar \alpha}{\alpha}}\rho^{\frac 12} \big)^{\alpha-1}=L_1 \sigma^{\bar \alpha} \rho^{-\bar \alpha}+L_2 \frac{\bar \alpha}{\alpha} \sigma^{-\alpha} \rho^{\alpha}. \notag 
\end{align}
Substituting this in \eqref{eq:sandrelent-twosample-alt}  and noting that $\big\|\rho^{\frac 12 }\sigma^{\frac{\bar \alpha}{\alpha}}\rho^{\frac 12 }\big\|_{\alpha}^{\alpha}=\tr{\rho^{\alpha}\sigma^{\bar \alpha}}$ proves the claim. 
\subsection{Proof of Corollary \ref{cor:limdistfidel}}
By setting $\alpha=0.5$ in \eqref{eq:sandrelent-twosample-alt}, we obtain
\begin{align}
   & r_n \Big(\sanddiv{\rho_n}{\sigma_n}{\frac 12}-\sanddiv{\rho}{\sigma}{\frac 12}\Big)\trightarrow{w}  \frac{-\tr{\Big(D[\rho^{\frac 12}](L_1) \sigma \rho^{\frac 12}+ \rho^{\frac 12} \sigma D[\rho^{\frac 12}](L_1)  +\rho^{\frac 12} L_2 \rho^{\frac 12}\Big) \big(\rho^{\frac 12} \sigma\rho^{\frac 12} \big)^{-\frac 12}}}{\sqrt{F(\rho,\sigma)}}.\notag
\end{align}
The claim for fidelity then follows by noting that $F(\rho,\sigma)=e^{-\sanddiv{\rho}{\sigma}{1/2}}$, and applying the functional delta method to the above equation for the map $x \mapsto e^{-x}$ at $x=\sanddiv{\rho}{\sigma}{1/2}$.

Next, consider  max-divergence given in \eqref{eq:rendivinft} which corresponds to infinite-order sandwiched R\'{e}nyi divergence. 
Set  $\lambda_{\max}(A):=\lambda_{\max,I}(A)$,
\[
X_n:=\rho_n^{\frac12}\sigma_n^{-1}\rho_n^{\frac12} 
\qquad \text{and} \qquad
X:=\rho^{\frac12}\sigma^{-1}\rho^{\frac12}.
\]
Since $X_n,X \geq 0$, we have  
\[
e^{\mathsf{D}_{\max}(\rho_n\|\sigma_n)}=\lambda_{\max}(X_n) \qquad \text{and} \qquad
e^{\mathsf{D}_{\max}(\rho\|\sigma)}=\lambda_{\max}(X).
\]
We will use the fact that the map $A\mapsto\lambda_{\max}(A)$  is Hadamard directionally differentiable on Hermitian matrices. Its directional derivative at $X$ in the direction $H$ is
\[
D[\lambda_{\max}(X)](H)
=
\lambda_{\max,P_{max}(X)}\!\left(H\right),
\]
where $P_{\max}(X)$ denotes the eigenprojection of $X$ corresponding to $\lambda_{\max}(X)$.  
On the other hand,  the operator-valued delta method and the product rule give
\[
r_n(X_n-X)\trightarrow{w} \bar H \coloneqq D[\rho^{\frac12}](L_1)\sigma^{-1}\rho^{\frac12}
+\rho^{\frac12}\sigma^{-1}D[\rho^{\frac12}](L_1)
+\rho^{\frac12}D[\sigma^{-1}](L_2)\rho^{\frac12}. \notag
\]
Then, by the Hadamard directional differentiability of $\lambda_{\max}(X)$ in the direction  $\bar H$, we have via an application of the  functional delta method that
\[
r_n\!\left(
e^{\mathsf{D}_{\max}(\rho_n\|\sigma_n)}
-e^{\mathsf{D}_{\max}(\rho\|\sigma)}
\right)
\trightarrow{w}
\lambda_{\max,P_{\max}(X)}\!\left(
\bar H
\right).
\]
Finally, applying the ordinary delta method to $x\mapsto\log x$ at
$x=e^{\mathsf{D}_{\max}(\rho\|\sigma)}>0$ gives
\[
r_n\big(
\mathsf{D}_{\max}(\rho_n\|\sigma_n)
-\mathsf{D}_{\max}(\rho\|\sigma)
\big)
\trightarrow{w}
e^{-\mathsf{D}_{\max}(\rho\|\sigma)}
\lambda_{\max,P_{\max}(X)}\!\left(
\bar H
\right).
\]
\subsection{Proof of Theorem \ref{Thm:quantrelent-limdist-inf}}\label{Thm:quantrelent-limdist-inf-proof}
We will prove the claim assuming $\rho,\sigma>0$. 
In infinite dimensions since $\rho^{-1}$ and $\sigma^{-1}$ are unbounded, the products with $\rho^{-1}$ and $\sigma^{-1}$ below are therefore understood through the trace-class extensions in the definition of $\mathfrak S_{1,\rho^{-1}}$ and $\mathfrak S_{1,\sigma^{-1}}$. The general support case is obtained by applying the same estimates on the fixed support subspaces $P_\rho\HH$ and $P_\sigma\HH$; the assumptions $\rho_n \ll \sigma_n  \ll \sigma$ and $\rho_n \ll \rho \ll  \sigma$ ensure that no components outside these subspaces occur.  Following the proof of Theorem \ref{Thm:quantrelent-limdist}, it suffices to show that the terms in a Taylor's expansion of the quantum relative entropy are well-defined and the uniform integrability of the remainder terms are ensured.   Note that $\qrel{\rho_n}{\sigma_n}<\infty$ for all $n \in \NN$ and $\qrel{\rho}{\sigma}<\infty$ by assumption.

The key difference from the finite dimensional case is that the argument that there exists a constant $0<c'<1$ such that  $(1-t)\rho+t\sigma_{n_{k_j}} \geq c' \rho$   for sufficiently large $j$ does not hold. Hence, we need a different argument to ensure uniform integrability of the relevant terms. For the two-sample null case the weighted-norm convergence is sufficient for this purpose; the additional stochastic boundedness condition $\PP(\|\rho_n\sigma_n^{-1}\|_\infty>c)\to0$ is used only in the two-sample alternative.

\noindent
 \textbf{Two-sample null:} 
Consider a  subsequence  $(n_{k_j})_{j \in \NN}$ such that 
$\big((r_{n_{k_j}}(\rho_{n_{k_j}}-\rho),r_{n_{k_j}}(\sigma_{n_{k_j}}-\rho)\big) \rightarrow (L_1,L_2)$ in $\norm{\,\cdot\,}_{1,\rho^{-1}} \vee \norm{\,\cdot\,}_{1,\rho^{-1}} $ a.s. This is again possible  by Skorokhods representation theorem  by separability of the space of trace-class operators, i.e., set of operators with finite trace norm, which further implies separability of $\mathfrak{S}_{1,\rho^{-1}}$ (indeed, $X\mapsto X\rho$ is an isometric isomorphism from the closed subspace $\{X\in\mathfrak{S}_1:X=XP_\rho\}$ onto $\mathfrak{S}_{1,\rho^{-1}}$).

Since the high-level proof is similar to that of Theorem \ref{Thm:quantrelent-limdist}, we will only highlight the differences.  
To show the  integrability of the first term in the RHS of  \eqref{eq:Tayexp-twosamp-null}, we have
\begin{align}
    \int_{0}^{\infty}  \norm{\rho\left(\tau I+\rho\right)^{-1} (\rho_n-\rho)  \left(\tau I+\rho\right)^{-1}}_1 d\tau &=  \int_{0}^{\infty}  \norm{\rho\left(\tau I+\rho\right)^{-1} (\rho_n-\rho) \left(\tau I+\rho\right)^{-1}}_1 d\tau \notag \\
    &\leq    \norm{(\rho_n-\rho)\rho^{-1}}_1 \int_{0}^{\infty}  \norm{\rho\left(\tau I+\rho\right)^{-1}}_{\infty}^2  d\tau  \notag \\
       & \leq   \norm{(\rho_n-\rho)\rho^{-1}}_1, \label{eq:interchterms}
\end{align}
where  the first inequality is due to H\"{o}lder's inequality for Schatten-norms and   the second inequality is because
\begin{align}
    \int_{0}^{\infty}  \norm{\rho\left(\tau I+\rho\right)^{-1}}_{\infty}^2  d\tau =\int_{0}^{\infty}  \frac{\lambda_{\max}(\rho)^2}{(\tau+\lambda_{\max}(\rho))^2}  d\tau = \lambda_{\max}(\rho) \leq 1,\label{eq:bndinteg}
\end{align}
where $\lambda_{\max}(\rho)=\|\rho\|_\infty$. 
Similarly,
\begin{align}
   \int_{0}^{\infty} \norm{  \rho\left(\tau I+\rho\right)^{-1} (\sigma_n-\rho) \left(\tau I+\rho\right)^{-1}}_1 d\tau &\leq  \norm{(\sigma_n-\rho)\rho^{-1}}_1 \int_{0}^{\infty}  \norm{\rho\left(\tau I+\rho\right)^{-1}}_{\infty}^2  d\tau \leq \norm{(\sigma_n-\rho)\rho^{-1}}_1.\notag
\end{align}
Next, we show uniform integrability of the terms  defined in \eqref{eq:unifbochinttermqrel}. Recalling $v(\rho_n,\rho,\tau,t):=\big(\tau I+(1-t)\rho+t\rho_n\big)^{-1}$, 
the first term can be bounded as 
\begin{align}
\norm{p_n(r_n,t)}_1 &\leq (1-t) \int_{0}^{\infty} \norm{r_n(\rho_n-\rho) v(\rho_n,\rho,\tau,t) r_n(\rho_n-\rho) v(\rho_n,\rho,\tau,t)}_1 d\tau \notag \\
& \leq (1-t)\norm{r_n(\rho_n-\rho) \rho^{-1}}_1^2 \int_{0}^{\infty} \norm{\rho \,v(\rho_n,\rho,\tau,t)}_{\infty}^2 d\tau, \notag 
\end{align}
where the final inequality follows due to H\"{o}lder's inequality for Schatten-norms along with $\norm{A}_p \geq \norm{A}_q$ for any linear operator $A$ and $1 \leq p \leq q \leq \infty$. 
Next, note that since $r_n(\rho_{n_{k_j}}-\rho)\rho^{-1} \rightarrow L_1\rho^{-1}$ in trace norm a.s., we have $(\rho_{n_{k_j}} -\rho)\rho^{-1} \rightarrow 0$ in $\norm{\,\cdot\,}_{\infty}$ a.s. Using this, we have a.s. that there exists sufficiently large $j_0 \in \NN$ (depending on the random realization and $\rho$) such that for all $n=n_{k_j}$ with $j \geq j_0$,
\begin{align}
\norm{\rho \,v(\rho_n,\rho,\tau,t)}_{\infty}&=   \norm{\rho(\tau I+\rho)^{-1}  \big(I+t(\rho_n-\rho)\rho^{-1}\rho (\tau I+\rho)^{-1}\big)^{-1}}_{\infty} \notag \\
& \leq \norm{\rho(\tau I+\rho)^{-1}}_{\infty}\norm{\big(I+t(\rho_n-\rho)\rho^{-1}\rho (\tau I+\rho)^{-1}\big)^{-1}}_{\infty} \notag \\
& \leq \frac{\norm{\rho(\tau I+\rho)^{-1}}_{\infty}}{1-\norm{(\rho_n-\rho)\rho^{-1}}_{\infty}\norm{\rho (\tau I+\rho)^{-1}}_{\infty}} \notag \\
& \leq 2\norm{\rho(\tau I+\rho)^{-1}}_{\infty}.\label{eq:estschatinfnorm} 
\end{align} 
For all $n=n_{k_j}$ with $j \geq j_0$ and $t \in [0,1]$, 
\begin{align}
 \norm{p_n(r_n,t)}_1 &\leq 4(1-t)\norm{r_n(\rho_n-\rho) \rho^{-1}}_1^2  \int_{0}^{\infty} \norm{\rho\left(\tau I+\rho\right)^{-1}}_{\infty}^2 d\tau \notag \\
& \leq 4(1-t)\norm{r_n(\rho_n-\rho)\rho^{-1}}_1^2  \notag \\
& \leq 4\norm{r_n(\rho_n-\rho)\rho^{-1}}_1^2, \notag   
\end{align}
where the penultimate inequality uses \eqref{eq:bndinteg}. 
Similarly, for the last term in \eqref{eq:unifbochinttermqrel}, we have for  all $n=n_{k_j}$ with $j \geq j_0$, 
\begin{align}
    \norm{\tilde q_n(r_n,t)}_1 &\leq (1-t)\norm{r_n(\rho_n-\rho)\rho^{-1}}_1 \norm{r_n(\sigma_n-\rho)\rho^{-1}}_1 \int_{0}^{\infty} \norm{\rho \,v(\sigma_n,\rho,\tau,t)}_{\infty}^2 d\tau \notag \\
& \leq 4\norm{r_n(\rho_n-\rho)\rho^{-1}}_1 \norm{r_n(\sigma_n-\rho)\rho^{-1}}_1. \notag
\end{align}
For the second and third terms in \eqref{eq:unifbochinttermqrel}, we have  for all $n=n_{k_j}$ with $j \geq j_0$,
\begin{align}
 \norm{q_n(r_n,t)}_1 &\leq   (1-t) \int_{0}^{\infty} \norm{\big((1-t) \rho+t\rho_n \big)v(\rho_n,\rho,\tau,t) r_n (\rho_n-\rho) v(\rho_n,\rho,\tau,t) r_n(\rho_n-\rho)v(\rho_n,\rho,\tau,t)}_1 d\tau \notag \\
 &\leq (1-t)\norm{r_n(\rho_n-\rho)\rho^{-1}}_1^2  \int_{0}^{\infty} \norm{I+t(\rho_n-\rho)\rho^{-1} }_{\infty}\norm{\rho \,v(\rho_n,\rho,\tau,t)}_{\infty}^3 d\tau \notag \\
 & \stackrel{(c)}{\leq}12 (1-t)\norm{r_n(\rho_n-\rho)\rho^{-1}}_1^2  \int_{0}^{\infty} \norm{\rho \left(\tau I+\rho\right)^{-1}}_{\infty}^3 d\tau \notag \\
 &\leq 6\norm{r_n(\rho_n-\rho)\rho^{-1}}_1^2, \notag \\
 \norm{\tilde p_n(r_n,t)}_1 &\leq (1-t)\int_{0}^{\infty} \norm{\big((1-t)\rho+t\rho_n\big) v(\sigma_n,\rho,\tau,t) r_n(\sigma_n-\rho) v(\sigma_n,\rho,\tau,t) r_n(\sigma_n-\rho)  v(\sigma_n,\rho,\tau,t)}_1 d\tau \notag \\
 &\leq (1-t)\int_{0}^{\infty} \norm{\big((1-t)\rho+t\rho_n\big) v(\sigma_n,\rho,\tau,t)}_{\infty}\norm{r_n(\sigma_n-\rho) v(\sigma_n,\rho,\tau,t) }_1^2 d\tau \notag \\
 &\leq (1-t)\norm{r_n(\sigma_n-\rho)\rho^{-1}}_1^2  \norm{\big((1-t) \rho+t\rho_n \big)\rho^{-1}}_{\infty}\int_{0}^{\infty} \norm{\rho \, v(\sigma_n,\rho,\tau,t)}_{\infty}^3 d\tau \notag \\
 & \leq (1-t)\norm{r_n(\sigma_n-\rho)\rho^{-1}}_1^2 \big(\norm{I+t(\rho_n -\rho)\rho^{-1}}_{\infty} \big)\int_{0}^{\infty} \norm{\rho\,v(\sigma_n,\rho,\tau,t)}_{\infty}^3 d\tau \notag \\
 & \leq 6 \norm{r_n(\sigma_n-\rho)\rho^{-1}}_1^2.   \notag
\end{align}
 From the above bounds, along the subsequence $(n_{k_j})_{j \in \NN}$, we have
\begin{subequations}\label{unifbochintinf}
    \begin{align}
 &\norm{ p_{n_{k_j}}(r_{n_{k_j}},t)}_1 
      \lesssim \norm{  r_{n_{k_j}}(\rho_{n_{k_j}} -\rho)\rho^{-1}}_1^2 , \\
      &\norm{\tilde q_{n_{k_j}}(r_{n_{k_j}},t)}_1  \lesssim \norm{r_{n_{k_j}}(\rho_{n_{k_j}}-\rho)\rho^{-1}}_1 \norm{r_{n_{k_j}}(\sigma_{n_{k_j}}-\rho)\rho^{-1}}_1, \\
     & \norm{q_{n_{k_j}}(r_{n_{k_j}},t)}_1 
\lesssim \norm{  r_{n_{k_j}}(\rho_{n_{k_j}} -\rho )\rho^{-1}}_1^2, \\
&\norm{\tilde p_{n_{k_j}}(r_{n_{k_j}},t)}_1 \lesssim  \norm{r_{n_{k_j}}(\sigma_{n_{k_j}}-\rho)\rho^{-1}}_1^2 .     
\end{align}
\end{subequations}
Since the RHS of the equations in \eqref{unifbochintinf} are pathwise bounded independently of $j$ for all sufficiently large $j$,  this yields the required uniform Bochner-integrable domination. Moreover,  since $v(\rho_n,\rho,\tau,t) \rightarrow (\tau I+\rho)^{-1}$ in $\norm{\cdot}_{\infty}$ norm for $\tau \in (0,\infty)$ and $t\in[0,1]$ a.s. along the subsequence $(n_{k_j})_{j \in \NN}$, we have   $p_{n_{k_j}}(r_{n_{k_j}},t)$, $q_{n_{k_j}}(r_{n_{k_j}})$, $\tilde p_{n_{k_j}}(r_{n_{k_j}},t)$ and $\tilde q_{n_{k_j}}(r_{n_{k_j}},t)$ converging in trace norm to their respective limits a.s. Hence, Bochner dominated convergence theorem applies a.s. The rest of the proof is same as that of Theorem \ref{Thm:quantrelent-limdist}, and hence omitted.

\medskip

\noindent
 \textbf{Two-sample alternative:} Consider the expansion in \eqref{eq:taylexpqrelalt}. By taking trace, we obtain \eqref{eq:bndtrqrelalt}  using \eqref{eq:tracevanishqrel} and $\qrel{\rho_n}{\sigma_n} \vee \qrel{\rho}{\sigma}<\infty$. We need to verify that the remaining terms in the expansion are well-defined and that the second-order terms satisfy a uniform integrability condition along a subsequence  $(n_{k_j})_{j \in \NN}$ such that 
$\big((r_{n_{k_j}}(\rho_{n_{k_j}}-\rho),r_{n_{k_j}}(\sigma_{n_{k_j}}-\sigma)\big) \rightarrow (L_1,L_2)$ in $\norm{\,\cdot\,}_{1,\rho^{-1}}\vee  \norm{\,\cdot\,}_{1,\sigma^{-1}}$ and $\big\|\rho_{n_{k_j}} \sigma_{n_{k_j}}^{-1}\big\|_{\infty} \leq c$ a.s. Such a subsequence is guaranteed by Skorokhod's representation theorem and an application of the Borel-Cantelli lemma using $\PP\big(\big\|\rho_n \sigma_n^{-1}\big\|_{\infty}>c\big) \rightarrow 0$.

Along the subsequence   $(n_{k_j})_{j \in \NN}$, $\big\|\rho_{n_{k_j}} \sigma_{n_{k_j}}^{-1}\big\|_{\infty} \leq c$ implies that $\sigma_{n_{k_j}}^{-1}\rho_{n_{k_j}}^2\sigma_{n_{k_j}}^{-1} \leq c^2 I$, or equivalently, $\rho_{n_{k_j}}^2 - c^2 \,\sigma_{n_{k_j}}^2 \leq 0$.  Taking limits $ j \rightarrow \infty$ and noting that the  convergence in $\norm{\,\cdot\,}_{1,\rho^{-1}}\vee  \norm{\,\cdot\,}_{1,\sigma^{-1}}$ norm implies  $\rho_{n_{k_j}} \rightarrow \rho$ and $\sigma_{n_{k_j}} \rightarrow \sigma$ in trace norm, we have $\rho-c \sigma \leq 0$, yielding  $\big\|\rho \sigma^{-1}\big\|_{\infty} \leq c$. Hence
\begin{align}
    \norm{\rho\,(\log \rho-\log \sigma)}_{\infty} \leq \norm{\rho\,\log \rho}_{\infty}+\norm{\rho \log \sigma}_{\infty} \leq (1/e)+\norm{\rho  \sigma^{-1}}_{\infty} \norm{\sigma \log \sigma}_{\infty}\leq (1+c)/e.\notag
\end{align}

Using these and $\rho,\sigma>0$, we have
 \begin{align}
     \norm{r_n(\rho_n-\rho)\left(\log \rho-\log \sigma\right)}_1 \leq \norm{r_n(\rho_n-\rho)\rho^{-1}}_1 &\norm{\rho\, (\log \rho-\log \sigma)}_{\infty} \leq \frac{1+c}{e} \, \norm{r_n(\rho_n-\rho)\rho^{-1}}_1, \notag \\
     \norm{\rho \int_{0}^{\infty} \left(\tau I+\sigma\right)^{-1} r_n(\sigma_n-\sigma) \left(\tau I+\sigma\right)^{-1} d\tau }_1 & \leq \int_{0}^{\infty} \norm{\rho \left(\tau I+\sigma\right)^{-1} r_n(\sigma_n-\sigma) \left(\tau I+\sigma\right)^{-1}}_1 d\tau \notag \\
      & \leq  \norm{\rho \sigma^{-1}}_{\infty} \norm{r_n(\sigma_n-\sigma)\sigma^{-1}}_1 \int_{0}^{\infty} \big\|\sigma \left(\tau I+\sigma\right)^{-1}\big\|_{\infty}^2  d\tau \notag \\
     &\leq \big\|r_n(\sigma_n-\sigma)\sigma^{- 1}\big\|_1 \big\|\rho \sigma^{-1}\big\|_{\infty} \notag \\
     & \leq c\,\big\|r_n(\sigma_n-\sigma)\sigma^{- 1}\big\|_1. \notag
 \end{align}
Note that a.s., the RHS of the above equations are finite and can be uniformly bounded  along the subsequence $(n_{k_j})_{j \in \NN}$ for sufficiently large $j$, thanks to $\big((r_{n_{k_j}}(\rho_{n_{k_j}}-\rho),r_{n_{k_j}}(\sigma_{n_{k_j}}-\sigma)\big) \rightarrow (L_1,L_2)$ in $\norm{\,\cdot\,}_{1,\rho^{-1}}\vee  \norm{\,\cdot\,}_{1,\sigma^{-1}}$. Also, we have  $\big((r_{n_{k_j}}(\rho_{n_{k_j}}-\rho),r_{n_{k_j}}(\sigma_{n_{k_j}}-\sigma)\big) \rightarrow (L_1,L_2)$ in trace norm.  These imply that the first two terms in \eqref{eq:bndtrqrelalt} are well defined and also enable taking limit $j \rightarrow \infty$, yielding 
 \begin{align}
  & \tr{r_{n_{k_j}}(\rho_{n_{k_j}}-\rho)\left(\log \rho-\log \sigma\right)-\rho \mspace{-2 mu}\int_{0}^{\infty} \mspace{-6 mu}\left(\tau I+\sigma\right)^{-1} r_{n_{k_j}}(\sigma_{n_{k_j}}-\sigma) \left(\tau I+\sigma\right)^{-1} d\tau   } \notag \\
   & \trightarrow{w} \tr{L_1\left(\log \rho-\log \sigma\right)}-\tr{\rho \int_{0}^{\infty} \left(\tau I+\sigma\right)^{-1} L_2 \left(\tau I+\sigma\right)^{-1} d\tau}. \notag 
 \end{align}
It remains to show that $p_n(r_n,t),q_n(r_n,t),\bar p_n(r_n,t), \bar q_n(r_n,t)$ as defined in \eqref{eq:unifbochinttermqrel} and \eqref{eq:bochintaltqrel} are uniformly integrable along the subsequence  $(n_{k_j})_{j \in \NN}$. For the first two terms, this follows similar to \eqref{unifbochintinf}. 
For the third and fourth terms, we have for all $n=n_{k_j}$ with $j \geq j_0$ as in the proof above that
\begin{align}
  \norm{\bar p_n(r_n,t)}_1& \leq (1-t) \norm{r_n^{\frac 12}(\sigma_n-\sigma)\sigma^{-1}}_1^2\int_{0}^{\infty}  \norm{\big((1-t)\rho+t\rho_n\big) \sigma^{-1}}_{\infty} \norm{\sigma \,v(\sigma_n,\sigma,\tau,t)}_{\infty}^3  d\tau \notag    \\ 
 & \leq (1-t) \norm{r_n^{\frac 12}(\sigma_n-\sigma)\sigma^{-1}}_1^2 \norm{\big((1-t)\rho+t\rho_n\big) \rho^{-1}}_{\infty}\norm{\rho \sigma^{-1}}_{\infty}\int_{0}^{\infty}   \norm{\sigma \,v(\sigma_n,\sigma,\tau,t)}_{\infty}^3  d\tau \notag    \\
  & \leq 12c(1-t) \norm{r_n^{\frac 12}(\sigma_n-\sigma)\sigma^{-1}}_1^2 \int_{0}^{\infty}
 \norm{\sigma\left(\tau I+\sigma\right)^{-1}}_{\infty}^3  d\tau \notag\\
  & \leq 6c \norm{r_n^{\frac 12}(\sigma_n-\sigma)\sigma^{- 1}}_1^2, \notag
 \end{align}  
where the third inequality uses \eqref{eq:estschatinfnorm} along with  $\norm{\rho \sigma^{-1}}_{\infty} \leq c$ and $\norm{\big((1-t)\rho+t\rho_n\big) \rho^{-1}}_{\infty} \leq 3/2$ (assuming $j_0$ is sufficiently large such that $\norm{(\rho_n-\rho)\rho^{-1}}_{\infty} \leq 0.5$). 
Similarly,
\begin{align}
    \norm{\bar q_n(r_n,t)}_1& \leq (1-t)  \int_{0}^{\infty} \norm{r_n^{\frac 12}(\rho_n-\rho)\rho^{-1}}_1\norm{\rho \sigma^{-1}}_{\infty}\norm{r_n^{\frac 12}(\sigma_n-\sigma) \sigma^{-1}}_1\norm{\sigma \,v(\sigma_n,\sigma,\tau,t)}_{\infty}^2 d\tau\notag \\
& \leq \norm{\rho \sigma^{-1}}_{\infty}\norm{r_n^{\frac 12}(\rho_n-\rho)\rho^{-1}}_1 \norm{r_n^{\frac 12}(\sigma_n-\sigma) \sigma^{- 1}}_1, \notag \\
& \leq c \,\norm{r_n^{\frac 12}(\rho_n-\rho)\rho^{-1}}_1 \norm{r_n^{\frac 12}(\sigma_n-\sigma) \sigma^{- 1}}_1. \notag 
\end{align}
Note that the RHS of the above equations are uniformly bounded along the subsequence  $(n_{k_j})_{j \in \NN}$ for sufficiently large $j$ a.s.  This establishes the required uniform integrability conditions  (the required pathwise Bochner-integrable domination). Moreover, since $v(\rho_n,\rho,\tau,t)\to(\tau I+\rho)^{-1}$ and
$v(\sigma_n,\sigma,\tau,t)\to(\tau I+\sigma)^{-1}$ in $\norm{\cdot}_{\infty}$ norm for $\tau \in (0,\infty)$ and $t \in [0,1]$ a.s. along the subsequence $(n_{k_j})_{j \in \NN}$, we have   $p_{n_{k_j}}(r_{n_{k_j}},t)$, $q_{n_{k_j}}(r_{n_{k_j}})$, $\bar p_{n_{k_j}}(r_{n_{k_j}},t)$ and $\bar q_{n_{k_j}}(r_{n_{k_j}},t)$ converging in trace norm to their respective limits a.s. Hence, Bochner dominated convergence theorem applies. The rest of the  proof is omitted, being similar to that of Theorem \ref{Thm:quantrelent-limdist}. 
\subsection{Proof of Proposition \ref{prop:limdisttomest}} \label{Sec:prop:limdisttomest-proof}
Recall that $\hat {\bm{s}}^{(n)}(\rho):=\big(\hat s^{(n)}_1(\rho),\ldots, \hat s^{(n)}_{d^2-1}(\rho)\big)$, where
 \begin{align}
      \hat s^{(n)}_{j}(\rho) &:= \frac 1n \sum_{k=1}^n \ind_{O_{k}(j,\rho) =+1}-\ind_{O_{k}(j,\rho) =-1 }, ~1 \leq j \leq d^2-1. \notag 
 \end{align}
 With $\gamma_0=I$ and $\hat s^{(n)}_{0}(\rho)=1$, we have 
 \begin{align}
  \sqrt{n}\big(\hat \rho_n-\rho\big) &=   \sqrt{n}\left(\hat \rho_n-\bar \rho_n\right)+\sqrt{n}\left( \bar \rho_n-\rho\right), \label{eq:trianglineqest}
 \end{align}
 where recall that
 \begin{align}
       \bar \rho_n &=  \frac {1}{d} \sum_{j=0}^{d^2-1} \hat s_{j}^{(n)}(\rho) \gamma_j. \notag 
 \end{align}
Note that the first term above can be written as
\begin{align}
    \sqrt{n}\left(\hat \rho_n-\bar \rho_n\right)= \ind_{\bar \rho_n \ngeq 0 }\sqrt{n} \left(P_{\cS_d}(\bar \rho_n)-\frac {1}{d} \sum_{j=0}^{d^2-1} \hat s_{j}^{(n)}(\rho) \gamma_j\right). \notag
\end{align}
We have $\big|\hat s_{j}^{(n)}(\rho)\big| \leq 1$ for all $j,n$. Moreover,  the entries of $\gamma_j$ are bounded, and so is  $P_{\cS_d}(\bar \rho_n)$ being a projection onto the space of density operators. From this,   it follows that  
\begin{align}
  \sqrt{n}\norm{\hat \rho_n-\bar \rho_n}_1 \lesssim_d \sqrt{n}. \notag
\end{align}
Consequently,
 \begin{align}
\EE\left[ \sqrt{n}\norm{\hat \rho_n-\frac {1}{d} \sum_{j=0}^{d^2-1} \hat s_{j}^{(n)} (\rho)\gamma_j}_1\right] \lesssim_d \PP\left(\bar \rho_n \ngeq 0 \right) \sqrt{n}. \label{eq:l1normexp}  
 \end{align} 
For $\bm{s}=(s_1,\ldots,s_{d^2-1}) \in \RR^{d^2-1}$, let $\omega(\bm{s}):=\frac 1d \left(I+\sum_{j=1}^{d^2-1} s_{j} \gamma_j\right)$ Note that the set $\cC:=\{\bm{s} \in \RR^{d^2-1}: \omega(\bm{s}) \geq 0\}$ is a convex set$\mspace{2 mu}$\footnote{This set has a simple characterization when $d=2$, given by  $\cC=\{\bm{s} \in \RR^{d^2-1}: \norm{\bm{s}}_2 \leq 1\}$. This is due to the fact that the eigenvalues of $\omega(\bm{s})$ are equal to $(1 \pm \norm{\bm{s}}_2)/2$. }. Recalling that $\bm{s}(\rho)=\big(s_1(\rho),\ldots, s_{d^2-1}(\rho)\big)$ with  $s_{j}(\rho)=\tr{\rho \gamma_j}$ and $\rho,\sigma >0$, we have  $\bm{s}(\rho), \bm{s}(\sigma) \in \mathrm{int}(\cC)$, where $\mathrm{int}(\cC)$ denotes the interior of $\cC$.  The set $\mathrm{int}(\cC)$ corresponds to $\bm{s}$ such that the eigenvalues $\omega(\bm{s})$  are strictly positive. Let $\{\lambda_i(\bm{s})\}_{i=1}^{d^2} $ denote the eigenvalues of $\omega(\bm{s})$ such that $\lambda_i(\bm{s}) \leq \lambda_j(\bm{s})$ when $j \leq i$. Note that these eigenvalues satisfy $\sum_{i=1}^{d^2}\lambda_i(\bm{s})=1$ and  are the roots of a characteristic polynomial of $\omega(\bm{s})$. 
Let 
\begin{align}
    \delta(\rho):=\min_{1 \leq i \leq d^2} \{\lambda_i(\bm{s}(\rho))/2\} \notag
\end{align}
and consider the event 
\begin{align}
    \cE_i:= \{\abs{\lambda_i\big(\hat {\bm{s}}^n(\rho)\big)- \lambda_i\big(\bm{s}(\rho)\big)} \geq \delta (\rho) \}. \notag
\end{align}
For $\rho,\sigma>0$, we have $ \delta(\rho) \wedge  \delta(\sigma)>0$. 
By the continuity of the roots of a polynomial as specified in \cite[Theorem 1.4]{Marden-1966}, there exists an $\epsilon_{\delta(\rho)} >0$ such that $\norm{\bm{s}-\tilde {\bm{s}}}_1 \leq \epsilon_{\delta(\rho)} $ implies that $\abs{\lambda_i(\bm{s})-\lambda_i(\tilde {\bm{s}})} \leq \delta(\rho) $ for all $1 \leq i \leq d^2$.     
Hence, 
\begin{align}
\{\bar \rho_n \ngeq 0\} \subseteq \cup_{i=1}^{d^2} \cE_i \subseteq \{\norm{\hat{ \bm{s}}^n(\rho)- \bm{s}(\rho)}_2 > \epsilon_{\delta(\rho)}\}. \notag    
\end{align}
Now, since $\EE\big[\hat {\bm{s}}^{(n)}(\rho)\big]=\bm{s}(\rho)$, we have  by an application of  Hoeffding's inequality that
 \begin{align}
\PP\left(\bar \rho_n \ngeq 0 \right) \leq  \PP\left( \norm{\hat{ \bm{s}}^n(\rho)- \bm{s}(\rho)}_2 \geq \epsilon_{\delta(\rho)} \right) \leq e^{-n c_i(\rho,d)}, \notag  
 \end{align}
where  $c_i(\rho,d)>0$ is some constant that depends on $\rho,d$. 
Consequently, the LHS of \eqref{eq:l1normexp} converges to zero which implies that  the first term in the RHS of \eqref{eq:trianglineqest} converges weakly to zero. 

Next, consider the second term in the RHS of \eqref{eq:trianglineqest}. Since the measurements for different Pauli operators are done on independent copies of $\rho$,   $\hat s^{(n)}_{j}(\rho)$ are independent across different $j$. Setting $X_k(j,\rho):= \ind_{O_{k}(j,\rho) =+1}-\ind_{O_{k}(j,\rho) =-1 }$ and  noting that $\EE[X_k(j,\rho)]=s_{j}(\rho)$,
we have  by the classical central limit theorem (CLT) that
\begin{align}
 \sqrt{n}\left( \frac {1}{d} \sum_{j=0}^{d^2-1} \hat s_{j}^{(n)}(\rho) \gamma_j-\rho\right)
  &=\sum_{j=1}^{d^2-1}   \frac{\gamma_j n^{-\frac 12}}{d}\sum_{k=1}^n \big( X_k(j,\rho)-s_{j}(\rho)  \big)  \trightarrow{w} \underbrace{\sum_{j=1}^{d^2-1}  \gamma_j Z_{j}(\rho)}_{L_{\rho}},\notag
 \end{align}
where $Z_{j}(\rho) \sim N\big(0,4s_{j}^+(\rho)s_{j}^-(\rho)/d^2\big) $ are independent for different $j$.  
From this and the fact that the first term in the RHS of \eqref{eq:trianglineqest} converges weakly to zero, an application of Slutsky's theorem  yields $\sqrt{n}(\hat \rho_n-\rho) \trightarrow{w} L_{\rho}$. Repeating the same arguments with $\rho,\bar \rho_n$ replaced by $\sigma,\bar \sigma_n$,  and noting that $\norm{I}_1/n \rightarrow 0$ leads to $  \sqrt{n}(\hat \sigma_n-\sigma) \trightarrow{w} L_{\sigma}$. Consequently,  $\big(\sqrt{n}(\hat \rho_n-\rho\big),\sqrt{n}(\hat \sigma_n-\sigma) \big)\trightarrow{w} (L_{\rho},L_{\sigma})$, where $L_{\rho}$ and $L_{\sigma}$ are independent  due to the independence of the measurements on $\rho$ and $\sigma$.  Moreover,  $\hat \rho_n \ll \hat \sigma_n  \ll \sigma$ and $\hat \rho_n \ll \rho \ll  \sigma$ since $\hat \rho_n \ll \hat \sigma_n$ and $\rho, \sigma >0$. 
Hence, Theorem \ref{Thm:quantrelent-limdist} applies and \eqref{eq:qrel-twosample-alt}  yields 
 \begin{align}
     \sqrt{n} \big(\qrel{\hat \rho_n}{\hat \sigma_n}-\qrel{\rho}{\sigma}\big) &\trightarrow{w}  \tr{L_{\rho} (\log \rho-\log \sigma )-\rho D[\log \sigma](L_{\sigma})} \notag \\
     &=\sum_{j=1}^{d^2-1}Z_{j}(\rho) \tr{\gamma_j(\log \rho-\log \sigma )} -Z_j(\sigma)\tr{\rho D[\log \sigma](\gamma_j) } \notag \\
     & \sim N(0, v_2^2(\rho,\sigma)), \notag
    \end{align}
    where $v_2^2(\rho,\sigma)$ is defined in Proposition \ref{prop:limdisttomest}. 
    This completes the proof.

    \subsection{Proof of Proposition \ref{prop:HTperf}}\label{prop:HTperf-proof}
  We will use Proposition \ref{prop:limdisttomest} to prove the claim. To that end, we first bound the relevant variances  $v_1^2(\rho_k,\sigma)$ for $k \in \cI$ given by 
  \begin{align}
    v_1^2(\rho_k,\sigma)&:=  \sum_{j=1}^{d^2-1}\frac{4s_{j}^+(\rho_k)s_{j}^-(\rho_k)}{d^2}\tr{\gamma_j(\log \rho_k -\log \sigma  )}^2.  \notag
  \end{align}
   Observe that $s_{j}^+(\rho_k)s_{j}^-(\rho_k) \leq 1/4$ since $0 \leq s_{j}^+(\rho_k)=1- s_{j}^-(\rho_k) \leq 1$. Hence, we have
    \begin{align}
     v_1^2(\rho_k,\sigma) &\leq    \frac{1}{d^2} \sum_{j=1}^{d^2-1}\tr{\gamma_j(\log\rho_k -\log \sigma)}^2 \notag \\
     & \leq     \frac{2}{d^2} \left(\sum_{j=1}^{d^2-1}\tr{\gamma_j\log \rho_k }^2+\tr{\gamma_j\log \sigma}^2\right)  \notag \\
    &\leq  \frac{2}{d^2} \left(\sum_{j=1}^{d^2-1}\norm{\gamma_j}_2^2\norm{\log\rho_k }_2^2+\norm{\gamma_j}_2^2\norm{\log \sigma}_2^2 \right)\notag \\
    &  \leq 4d^2 (\log b)^2 , \notag
    \end{align}
    where the second inequality uses $(a-b)^2 \leq 2(a^2+b^2)$ for $a,b \in \RR$, the penultimate inequality  follows by Cauchy-Schwarz, and the final inequality uses $\|\gamma_j\|_2^2 \leq d $ and $\|\log \sigma\|_2^2 \vee \|\log \rho_k \|_2^2 \leq d (\log b)^2$ for all $k \in \cI$. Then,
    \begin{align}
        \alpha_{i,n}\big(\cT_n^{\mathrm{tom}},\rho_i^{\otimes n}, \sigma_i^{\otimes n}\big)&=\PP\big(\hat D_n \notin (\epsilon_i+cn^{-\frac 12}, \epsilon_{i+1}+c n^{-\frac 12})|H=i\big) \notag \\
        &=\PP\big(\hat D_n-\qrel{\rho_i}{\sigma} \notin (\epsilon_i-\qrel{\rho_i}{\sigma}+c n^{-\frac 12}, \epsilon_{i+1}-\qrel{\rho_i}{\sigma}+c n^{-\frac 12})|H=i\big) \notag \\
        & \leq \PP\big(n^{\frac 12}\big(\hat D_n-\qrel{\rho_i}{\sigma}\big) \notin (n^{\frac 12}(\epsilon_i-\qrel{\rho_i}{\sigma})+c, c)|H=i\big), \notag
    \end{align}
    where the final inequality follows because $\qrel{\rho_i}{\sigma} \leq \epsilon_{i+1}$. 
Note that $n^{\frac 12}\big(\hat D_n-\qrel{\rho_i}{\sigma}\big) \trightarrow{w} W_i \sim N(0,v_1^2(\rho_i,\sigma))$ given $H_i$ is the true hypothesis by \eqref{eq:qreltom-onesample-alt} in Proposition \ref{prop:limdisttomest}.  Taking limits in the equation above and applying Portmanteaus theorem \cite[Theorem 2.1]{Billingsley-99}   yields
    \begin{align}
    \limsup_{n \rightarrow \infty} \alpha_{i,n}\big(\cT_n^{\mathrm{tom}},\rho_i^{\otimes n}, \sigma_i^{\otimes n}\big) &\leq \limsup_{n \rightarrow \infty} \PP\big(n^{\frac 12}\big(\hat D_n-\qrel{\rho_i}{\sigma}\big) \notin (n^{\frac 12}(\epsilon_i-\qrel{\rho_i}{\sigma})+c, c)|H=i\big) \notag \\
    & = \PP\big(W_i \notin (-\infty, c)\big). \notag 
    \end{align}
   For $i \in \cI$ such that $v_1^2(\rho_i,\sigma)>0$,   the above probability equals $Q\left(c/v_1(\rho_i,\sigma)\right)$, where $Q(x)$ is the upper-tail function of a standard normal random variable given by $Q(x):=(2 \pi)^{-1/2}\int_{x}^{\infty}e^{-u^2/2}du$.   If $Q^{-1}(\tau)\geq0$, then choosing $c\geq2dQ^{-1}(\tau)|\log b|$ gives $Q\big(c/v_1(\rho_i,\sigma)\big)\leq\tau$, due to $v_1(\rho_i,\sigma)\leq2d|\log b|$ for all $i \in \cI$. If $Q^{-1}(\tau)<0$ (equivalently $\tau>1/2$), taking $c\geq0$ gives $Q(c/v_1)\leq1/2<\tau$. For $i \in \cI $ such that  $v_1^2(\rho_i,\sigma)=0$, then the same choices of $c$ yields $\PP\big(W_i \notin (-\infty, c)\big)=0 \leq \tau$. Hence, setting $c\geq2d(Q^{-1}(\tau)\vee0)|\log b|$, the RHS above is bounded by $\tau$ for all $i \in \cI$.

\begin{appendices}
\section{Proof of Lemma \ref{lem:multprojtr}} \label{sec:app-lem:multprojtr-proof}
To prove $(i)$, note that since $A \ll P$ with $P$ being a projection, we have $PAP=A$. Hence, 
\begin{align}
\tr{AB}=\tr{PAPB}=\tr{PAPBP},\notag
\end{align}
where the final equality follows since $PAPB$ is trace-class and $P$ is bounded (being a projection), and using the fact that $\tr{MN}=\tr{NM}$ when $M$ and  $N$ are trace-class and bounded linear operators, respectively. 
\medskip

To prove $(ii)$, first consider the case that $\HH$ is finite dimensional, i.e., $\HH=\HH_d$. By the spectral theorem, we have $A=\sum_{i=1}^d \lambda_{i,A} \ket{e_{i,A}}\bra{e_{i,A}}$, where  $\{\lambda_{i,A}\}_{i=1}^d$ and $\{e_{i,A}\}_{i=1}^d$ are the set of eigenvalues and corresponding set of orthonormal eigenvectors of $A$, respectively. Similarly, let $B=\sum_{i=1}^d \lambda_{i,B} \ket{e_{i,B}}\bra{e_{i,B}}$ and $C=\sum_{i=1}^d \lambda_{i,C} \ket{e_{i,C}}\bra{e_{i,C}}$ be the spectral decomposition of $B$ and $C$, respectively.  Then, $B \leq A \leq C$ implies that for  all $ 1 \leq j \leq d$, 
\begin{align}
 \sum_{i=1}^d \lambda_{i,B} \abs{\langle e_{j,A}, e_{i,B}\rangle}^2 \leq  \bra{e_{j,A}} A \ket{e_{j,A}}= \lambda_{j,A}\leq \sum_{i=1}^d \lambda_{i,C} \abs{\langle e_{j,A}, e_{i,C}\rangle}^2. \notag 
\end{align}
Hence, we have
\begin{align}
  \norm{A}_1=  \sum_{j=1}^d \abs{\lambda_{j,A}} &\leq \sum_{j=1}^d \max \left\{\sum_{i=1}^d \abs{\lambda_{i,B}} \abs{\langle e_{j,A}, e_{i,B}\rangle}^2,\sum_{i=1}^d \abs{\lambda_{i,C}} \abs{\langle e_{j,A}, e_{i,C}\rangle}^2\right\} \notag \\
    & \leq \sum_{j=1}^d\sum_{i=1}^d \abs{\lambda_{i,B}} \abs{\langle e_{j,A}, e_{i,B}\rangle}^2 +\abs{\lambda_{i,C}}\abs{\langle e_{j,A}, e_{i,C}\rangle}^2 \notag \\
    &=\sum_{i=1}^d \abs{\lambda_{i,B}} +\abs{\lambda_{i,C}}=\norm{B}_1+\norm{C}_1, \label{eq:tracenormfinite}
\end{align}
where in the penultimate inequality, we upper bounded maximum of positive numbers by its sum, and in the penultimate equality, we  used that the  eigenvectors form an orthonormal basis implying  that 
\begin{align}
  \sum_{j=1}^d  \abs{\langle e_{j,A}, e_{i,B}\rangle}^2 =\sum_{j=1}^d \abs{\langle e_{j,A}, e_{i,C}\rangle}^2 =1,~ \forall~ 1 \leq i \leq d. \notag
\end{align}
For the case of separable $\HH$, let $(e_n)_{n \in \NN}$ be a sequence of orthonormal basis elements, and $(P_n)_{n \in \NN}$, with $P_n :=\sum_{i=1}^n \ket{e_i}\bra{e_i}$,  be an increasing sequence of orthogonal projections. Then, setting $A_n=P_nAP_n$, $B_n=P_nBP_n$,  $C_n=P_nCP_n$ and noting that $B_n \leq A_n \leq C_n$ , \eqref{eq:tracenormfinite} implies that $\norm{A_n}_1 \leq \norm{B_n}_1+\norm{C_n}_1$. Taking limits $n \rightarrow \infty$ and observing that $A_n \rightarrow A$, $B_n \rightarrow B$ and $C_n \rightarrow C$ in trace norm, completes the proof.    
\section{Limit distributions for Measured Relative Entropy} \label{Sec:limdist-mrel}
Here, we derive limit distributions for estimators of measured relative entropy with respect to a general class of measurements (see \cite{Donald1986,Hiai-Petz-1991,Bennett-1993,Rains-2001} for certain important classes).  Before stating our result, we need to introduce some terminology. Let $\cM$ denote a set of POVMs (see the books \cite{wilde2017quantum, Hayashi-book-2016, tomamichel2015quantum}), where a POVM here  refers to a set $M=\{M_i\}_{i \in \cI}$ of operators indexed by a discrete set  $\cI$  satisfying $0 \leq M_i \leq I$ and $\sum_{i \in \cI}M_i=I$. 
  The measured relative entropy between density operators $\rho$ and $\sigma$ with respect to  $\cM$ is 
\begin{align}
\mrel{\rho}{\sigma}:= \sup_{M \in \cM} \qrel{\mathsf{P}_{\rho,M}}{\mathsf{P}_{\sigma,M}}, \label{eq:mrelent}
\end{align}
where  $\mathsf{P}_{\rho,M}$ denotes the probability measure defined via $\mathsf{P}_{\rho,M}(i):=\tr{M_i\rho}$. 
In what follows, we will also use the notation 
$\mathsf{P}_{A,M}(\cdot)$ for a general operator $A$ and POVM $M$, in which case it  is no longer necessarily a probability distribution. 

Each POVM  $M$ defines a measurement  (quantum to classical) channel given by $M(\rho):=\sum_{i \in \cI} \tr{M_i\rho} \ket{i}\bra{i}$, where  $\ket{i}$ denotes the $i^{th}$ computational basis element. We will use the same notation $M$ for the POVM and the measurement channel (linear superoperator)  induced by it, with the usage being evident from the context.  With this notation,  $\mrel{\rho}{\sigma}=\sup_{M \in \cM}\qrel{M(\rho)}{M(\sigma)}$. We  identify two POVMs $M$ and $\tilde M$ in $\cM$ if for every $\rho,\sigma \in \cS_d$, $\mathsf{P}_{\rho,M}(\cdot)$ and $\mathsf{P}_{\sigma,M}(\cdot)$ are equivalent (up to the same permutation) to $\mathsf{P}_{\rho,\tilde M}(\cdot)$ and $\mathsf{P}_{\sigma, \tilde M}(\cdot)$, respectively, since 
measured relative entropy remains the same with this identification. Also, we will restrict to $\cM$ such that $\abs{\cI} \leq m$ for all $M \in\M$ and some  $m \in \NN$. This is not a restriction when $\cM$ contains the set of all projective measurements (see \cite{Berta2015OnEntropies}). By adding zero matrices as necessary,  we may assume without loss of generality that $\cI=\{0,\ldots,m-1\}$. 
We will view $\cM$ as a subset of the space of linear superoperators equipped with operator norm topology.

Let $M^{\star}(\rho,\sigma,\cM)$ denote an optimizer achieving the supremum (if it exists) in \eqref{eq:mrelent}. The next result  characterizes  limit distribution for measured relative entropy estimation in terms of quantities induced by optimal measurement. 
\begin{theorem}
[Limit distribution for measured relative entropy]\label{Thm:measrelent-limdist}
 Let $\rho_n \ll \sigma_n \ll \sigma$, $\rho_n \ll \rho \ll \gg \sigma$, and $\cM$ be compact  such that $M^{\star}:=M^{\star}(\rho,\sigma,\cM)$ is unique with the identification of POVMs as mentioned above. 
 If $\big(r_n(\rho_n-\rho),r_n(\sigma_n-\sigma)\big) \trightarrow{w} (L_1,L_2)$, then
 \begin{align}
    r_n \big(\mrel{\rho_n}{\sigma_n}-\mrel{\rho}{\sigma}\big) \trightarrow{w}   \sum_{i \in \cI}\left[\mathsf{P}_{L_1,M^{\star}}(i) \log \frac{\mathsf{P}_{\rho,M^{\star}}(i)}{\mathsf{P}_{\sigma,M^{\star}}(i)}-\frac{\mathsf{P}_{L_2,M^{\star}}(i)\mathsf{P}_{\rho,M^{\star}}(i)}{\mathsf{P}_{\sigma,M^{\star}}(i)}\right].\label{eq:mrel-twosample-alt}
\end{align}
\end{theorem}
Before we proceed with the proof of Proposition \ref{Thm:measrelent-limdist}, a few remarks are in order. 
Since the definition of $\mathsf{D}_{\cM}$ itself involves a supremum, a method of proof similar to that of  Theorem \ref{Thm:sandrelent-limdist} applies. However, a key technical difference arises due to the fact that without additional assumptions, the maximizer in \eqref{eq:mrelent} is not unique. Moreover, the maximizer does not have a closed form expression in general. Hence, to establish the above claim, our proof involves showing that the optimal measurement POVM of the empirical measured relative entropy, $\mrel{\rho_n}{\sigma_n}$, converges to  $M^{\star}$ in operator norm. This convergence necessitates the requirement of a  unique $M^{\star}$ as stated above.
 \begin{proof}
It suffices to show that for every subsequence of $\NN$, there exists a further subsequence  along which the convergence in \eqref{eq:mrel-twosample-alt} holds.  Let $n_{k_j}$ be a subsequence such that  $\big((r_{n_{k_j}}(\rho_{n_{k_j}}-\rho),r_{n_{k_j}}(\sigma_{n_{k_j}}-\sigma)\big) \rightarrow (L_1,L_2)$ 
 in trace-norm a.s.,  
which exists by the same reasons stated in the proof of Theorem \ref{Thm:quantrelent-limdist}.  
Let $M_n^{\star}$ be such that  
\begin{align}
    \mrel{\rho_n}{\sigma_n} = \qrel{\mathsf{P}_{\rho_n,M_n^{\star}}}{\mathsf{P}_{\sigma_n,M_n^{\star}}}.\notag
\end{align}
Such an $M_n^{\star}$ exists  since $\cM$ is compact and $\qrel{P_{\rho_n,M}}{P_{\sigma_n,M}}$ is a continuous functional of $M$ for $\rho_n \ll \sigma_n$.

We will show that $M_n^{\star} \rightarrow M^{\star}$ a.s. in operator norm. Note that since $\rho \ll \gg \sigma$, there exists constants $c_1,c_2>0$ such that  $c_1 \sigma \leq \rho \leq c_2 \sigma$. Since a quantum channel is a completely positive linear map, we also have  $c_1 M(\sigma) \leq M(\rho) \leq c_2 M(\sigma)$. Moreover, since $\rho_{n_{k_j}} \rightarrow \rho$ and $\sigma_{n_{k_j}} \rightarrow \sigma$ in trace norm a.s., there exists constants $\tilde c_1,\tilde c_2$ (depending on the realization $\rho_{n_{k_j}}$ and $\sigma_{n_{k_j}}$) such that for all $j$ sufficiently large, $0<\tilde c_1 \leq c_1$,  $0<c_2 \leq \tilde c_2$,  $\tilde c_1 \sigma_{n_{k_j}} \leq \rho_{n_{k_j}} \leq \tilde c_2 \sigma_{n_{k_j}} $ and $\tilde c_1 M\big(\sigma_{n_{k_j}}\big) \leq M \big(\rho_{n_{k_j}}\big) \leq \tilde c_2 M \big( \sigma_{n_{k_j}}\big) $.   
Then,  denoting by $A_-$, the negative part of $A=A_+-A_-$, we have
\begin{align}
\qrel{M(\rho)}{M(\sigma)}&\stackrel{(a)}{=}\int_{\tilde c_1}^{\tilde c_2} \frac{ds}{s} \tr{\big(M(\rho)-sM(\sigma)\big)_-}+ \log \tilde c_2+1-\tilde c_2 \notag \\
&\stackrel{(b)}{=}\int_{\tilde c_1}^{\tilde c_2} \frac{ds}{2s} \left(s-1+\norm{\big(M(\rho)-sM(\sigma)\big)}_1 \right)+ \log \tilde c_2+1-\tilde c_2, \notag
\end{align}
where $(a)$ follows from the integral representation of quantum relative entropy (see \cite[Theorem 6]{frenkel2023integral} and \cite[Corollary 1]{jencova-23}), while $(b)$ uses $\norm{A}_1=A_++A_-$ and $\tr{M(\rho)-sM(\sigma)}=1-s$. Using similar representation for $\qrel{M(\rho_{n_{k_j}})}{M(\sigma_{n_{k_j}})}$ and subtracting from previous equation, we have
\begin{align}
 & \abs{\qrel{M(\rho)}{M(\sigma)}-\qrel{M(\rho_{n_{k_j}})}{M(\sigma_{n_{k_j}})}}\notag \\
  &=\abs{\int_{\tilde c_1}^{\tilde c_2} \frac{ds}{2s} \norm{M(\rho)-sM(\sigma)}_1-\norm{M(\rho_{n_{k_j}})-sM(\sigma_{n_{k_j}})}_1} \notag \\
  & \leq \int_{\tilde c_1}^{\tilde c_2} \frac{ds}{2s} \abs{\big\|M(\rho)-sM(\sigma)\big\|_1-\big\|M(\rho_{n_{k_j}})-sM(\sigma_{n_{k_j}})\big\|_1} \notag \\
  & \stackrel{(a)}{\leq} \int_{\tilde c_1}^{\tilde c_2} \frac{ds}{2s} \big\|M(\rho)-M(\rho_{n_{k_j}})-sM(\sigma)+sM(\sigma_{n_{k_j}})\big\|_1 \notag \\
  & \stackrel{(b)}{\leq}  \int_{\tilde c_1}^{\tilde c_2} \frac{ds}{2s} \left(\big\|M(\rho)-M(\rho_{n_{k_j}})\big\|_1+s\big\|M(\sigma)-M(\sigma_{n_{k_j}})\big\|_1 \right) \notag \\
  & \stackrel{(c)}{\leq}  \int_{\tilde c_1}^{\tilde c_2} \frac{ds}{2s} \left(\big\|\rho-\rho_{n_{k_j}}\big\|_1+s\big\|\sigma-\sigma_{n_{k_j}}\big\|_1\right), \notag 
\end{align}
where $(a)$ and $(b)$ follows from triangle inequality, while $(c)$ is due to data processing inequality for trace norm. By setting $h_n(M)=\qrel{M(\rho_n)}{M(\sigma_n)}$, it follows from the above inequality that $h_{n_{k_j}}(M)$ converges uniformly to $\qrel{M(\rho)}{M(\sigma)}$  (as a function of $M$) given $\rho_{n_{k_j}} \rightarrow \rho$ and $\sigma_{n_{k_j}} \rightarrow \sigma$. Also, note that for any $M,\tilde M \in \cM$, we have via similar steps as above that
\begin{align}
  \abs{\qrel{M(\rho)}{M(\sigma)}-\mathsf{D}\big(\tilde M(\rho)\|\tilde M(\sigma)\big)}  
  & \leq  \int_{c_1}^{c_2} \frac{ds}{2s} \left(\big\|M(\rho)-\tilde M(\rho)\big\|_1+s\big\|M(\sigma)-\tilde M(\sigma)\big\|_1 \right) \notag \\
  & \leq  \int_{c_1}^{c_2} \frac{ds}{2s} \big\|M-\tilde M\big\| (s+1)  \notag \\
  &\leq \frac{\big\|M-\tilde M\big\|}{2} \left(\ln \frac{c_2}{c_1}+c_2-c_1\right), \notag
\end{align}
where $\norm{\cdot}$ in the last two inequalities denotes the operator norm. Hence, $\qrel{M(\rho)}{M(\sigma)}$ is a uniformly continuous functional of $M$ in operator norm.  Since $\cM$ is compact in the operator norm topology  and $M^{\star}=\argmax \mrel{\rho}{\sigma}$ is unique, the aforementioned uniform convergence implies that $M^{\star}_{n_{k_j}} \rightarrow M^{\star}$ a.s. in operator norm.

\medskip

Equipped with the above, we next prove \eqref{eq:mrel-twosample-alt}. Observe  that by definition of $M_n^{\star}$ and $M^{\star}$, 
we have 
\begin{subequations}\label{eq:measrelent-bnd}
\begin{align}
 &  r_n\left( \mrel{\rho_n}{\sigma_n} -\mrel{\rho}{\sigma} \right) \leq  r_n\left(\qrel{\mathsf{P}_{\rho_n,M_n^{\star}}}{\mathsf{P}_{\sigma_n,M_n^{\star}}}-\qrel{\mathsf{P}_{\rho,M_n^{\star}}}{\mathsf{P}_{\sigma,M_n^{\star}}}\right), \label{eq:measrelent-ub} \\
 &  r_n\left( \mrel{\rho_n}{\sigma_n} -\mrel{\rho}{\sigma} \right) \geq  r_n\left(\qrel{P_{\rho_n,M^{\star}}}{P_{\sigma_n,M^{\star}}}-\qrel{\mathsf{P}_{\rho,M^{\star}}}{\mathsf{P}_{\sigma,M^{\star}}}\right). \label{eq:measrelent-lb}
\end{align}
\end{subequations}
Denoting the LHS and RHS  of  \eqref{eq:mrel-twosample-alt} by $g_n(\rho_n,\sigma_n)$ and $g(L_1,L_2)$, respectively, we will show that along the subsequence $(n_{k_j})_{j \in \NN}$, the RHS of \eqref{eq:measrelent-ub} and \eqref{eq:measrelent-lb} converge a.s. to  $g(L_1,L_2)$. This then implies that  $g_{n_{k_j}} \rightarrow g(L_1,L_2)$ a.s. and  also weakly.

From Taylor expansion applied to the function $f(x,y)=x\log \frac{x}{y}$, we have\footnote{The conditions required for validity of Taylor's expansion requires $u,v,\bar u,\bar v >0$, but as noted in the proof of \cite[Proposition 1]{SGK-IT-2023}, this expansion is also valid  for all $u,v,\bar u,\bar v$ as mentioned above. } for $u,v,\bar u,\bar v$ such that $u \geq 0, v, \bar u,\bar v>0$ or $u=0, v \geq 0, \bar u,\bar v>0$ that
\begin{align}
  u\log \frac{u}{v}&=\bar u \log \frac{\bar u}{\bar v}+  \left(1+\log \frac{\bar u}{\bar v}\right)(u-\bar u)-\frac{\bar u}{\bar v}(v-\bar v)+\int_0^1 \frac{(1-\tau)(u-\bar u)^2 }{(1-\tau) \bar u+\tau u}d\tau \notag \\
  & +\int_0^1 \frac{(1-\tau)((1-\tau)\bar u+\tau u)(v-\bar v)^2 }{\big((1-\tau)\bar v+\tau v\big)^2}d\tau -2\int_0^1 \frac{(1-\tau)(u-\bar u)(v-\bar v) }{\big((1-\tau)\bar v+\tau v\big)}d\tau. \label{eq:taylorexp-mre}
\end{align}
Let  $u=\mathsf{P}_{\rho_n,M_n^{\star}}(i)$, $v=\mathsf{P}_{\sigma_n,M_n^{\star}}(i)$, $\bar u=\mathsf{P}_{\rho,M_n^{\star}}(i)$, $\bar v=\mathsf{P}_{\sigma, M_n^{\star}}(i)$ for  $i \in \cI$ and note that the aforementioned constraints on $u,v,\bar u,\bar v$ are satisfied since $\mathsf{P}_{\rho_n,M_n^{\star}} \ll \mathsf{P}_{\rho,M_n^{\star}} \ll \mathsf{P}_{\sigma,M_n^{\star}}$ and $\mathsf{P}_{\rho_n,M_n^{\star}} \ll \mathsf{P}_{\sigma_n,M_n^{\star}} \ll \mathsf{P}_{\sigma,M_n^{\star}}$  due to $\rho_n \ll \sigma_n \ll \sigma$ and  $\rho_n \ll \rho \ll  \sigma$.  
Substituting the above and summing the resulting expression over $i$  in the above equation leads to
\begin{align}
f_n&:=r_n\big(\qrel{M_n^{\star}(\rho_n)}{M_n^{\star}(\sigma_n)}-\qrel{M_n^{\star}(\rho)}{M_n^{\star}(\sigma)} \big)\notag \\
&= \sum_{i \in \cI} r_n (\mathsf{P}_{\rho_n,M_n^{\star}}(i)-\mathsf{P}_{\rho,M_n^{\star}}(i))\log \frac{\mathsf{P}_{\rho,M_n^{\star}}(i)}{\mathsf{P}_{\sigma, M_n^{\star}}(i)}-\sum_{i \in \cI} r_n(\mathsf{P}_{\sigma_n,M_n^{\star}}(i)-\mathsf{P}_{\sigma, M_n^{\star}}(i))\frac{\mathsf{P}_{\rho,M_n^{\star}}(i)}{\mathsf{P}_{\sigma, M_n^{\star}}(i)} +R_n,\notag
\end{align}
where $R_n=R_{n,1}+R_{n,2}-2R_{n,3}$,  and
\begin{align}
    R_{n,1}&:=\sum_{i \in \cI}\int_{0}^1 \frac{(1-\tau)\Big(r_n^{\frac 12}\big(\mathsf{P}_{\rho_n,M_n^{\star}}(i)-\mathsf{P}_{\rho,M_n^{\star}}(i)\big)\Big)^2 }{(1-\tau) \mathsf{P}_{\rho,M_n^{\star}}(i)+\tau \mathsf{P}_{\rho_n,M_n^{\star}}(i)}d\tau, \notag \\
    R_{n,2}&:=\sum_{i \in \cI}\int_0^1 \frac{(1-\tau)((1-\tau)\mathsf{P}_{\rho,M_n^{\star}}(i)+\tau \mathsf{P}_{\rho_n,M_n^{\star}}(i))\Big(r_n^{\frac 12}\big((\mathsf{P}_{\sigma_n,M_n^{\star}}(i)-\mathsf{P}_{\sigma, M_n^{\star}}(i)\big)\Big)^2 }{\big((1-\tau)\mathsf{P}_{\sigma_n,M_n^{\star}}(i)+\tau \mathsf{P}_{\sigma, M_n^{\star}}(i)\big)^2}d\tau, \notag \\ 
    R_{n,3}&:=\sum_{i \in \cI}\int_{0}^1 \frac{(1-\tau)r_n^{\frac 12}\big(\mathsf{P}_{\rho_n,M_n^{\star}}(i)-\mathsf{P}_{\rho,M_n^{\star}}(i)\big)r_n^{\frac 12}\big(\mathsf{P}_{\sigma_n,M_n^{\star}}(i)-\mathsf{P}_{\sigma, M_n^{\star}}(i)\big) }{\big((1-\tau)\mathsf{P}_{\sigma, M_n^{\star}}(i)+\tau \mathsf{P}_{\sigma_n,M_n^{\star}}(i)\big)}d\tau. \notag  
\end{align}
Note that since $M^{\star}_{n_{k_j}} \rightarrow M^{\star}$ in operator norm a.s. and $\big((r_{n_{k_j}}(\rho_{n_{k_j}}-\rho),r_{n_{k_j}}(\sigma_{n_{k_j}}-\sigma)\big) \rightarrow (L_1,L_2)$  in trace-norm a.s., we have $\mathsf{P}_{\rho_{n_{k_j}},M_{n_{k_j}}^*}(i) \rightarrow \mathsf{P}_{\rho,M^{\star}}(i)$, $\mathsf{P}_{\sigma_{n_{k_j}},M_{n_{k_j}}^*}(i) \rightarrow \mathsf{P}_{\sigma,M^{\star}}(i)$, and 
\begin{align}
&r_{n_{k_j}}\Big(\mathsf{P}_{\rho_{n_{k_j}},M_{n_{k_j}}^*}(i)-\mathsf{P}_{\rho,M_{n_{k_j}}^*}(i)\Big)=\operatorname{Tr}\Big[M_{n_{k_j}}^*(i)r_{n_{k_j}}\big(\rho_{n_{k_j}}-\rho)\Big] \rightarrow \tr{M^{\star}(i)L_1}, \notag \\
&r_{n_{k_j}}\Big(\mathsf{P}_{\sigma_{n_{k_j}},M_{n_{k_j}}^*}(i)-\mathsf{P}_{\sigma,M_{n_{k_j}}^*}(i)\Big)=\tr{M_{n_{k_j}}^*(i)r_{n_{k_j}}\big(\sigma_{n_{k_j}}-\sigma)} \rightarrow \tr{M^{\star}(i)L_2}, \notag \\
&r_{n_{k_j}}^{\frac 12}\Big(\mathsf{P}_{\rho_{n_{k_j}},M_{n_{k_j}}^*}(i)-\mathsf{P}_{\rho,M_{n_{k_j}}^*}(i)\Big) \rightarrow 0, \quad \mbox{and} \quad r_{n_{k_j}}^{\frac 12}\Big(\mathsf{P}_{\sigma_{n_{k_j}},M_{n_{k_j}}^*}(i)-\mathsf{P}_{\sigma,M_{n_{k_j}}^*}(i)\Big)\rightarrow 0. \notag 
\end{align}
Next, to see that  the remainder terms vanish in the limit, let $P=P_\rho=P_\sigma$ denote the common support of $\rho$ and $\sigma$. Since $\rho$ and $\sigma$ are strictly positive on the support of $P$, there is $c>0$ such that $\rho,\sigma\geq c P$. Along the chosen subsequence, for all sufficiently large $j$, also $\rho_{n_{k_j}},\sigma_{n_{k_j}}\geq(a/2)P$. For $m_i:=\tr{PM_{n_{k_j}}^\star(i)P}$, this gives denominators bounded below by $(c/2)m_i$, while
\[
\Big|\mathsf P_{\rho_{n_{k_j}},M_{n_{k_j}}^\star}(i)-\mathsf P_{\rho,M_{n_{k_j}}^\star}(i)\Big|
\leq \norm{\rho_{n_{k_j}}-\rho}_\infty m_i,
\]
and analogously for $\sigma$. Since the scaled perturbations are bounded and $r_n\to\infty$, each remainder is $O(m_i/r_n)$, uniformly in $i$; moreover $\sum_i m_i=\tr P<\infty$. Consequently $R_{n_{k_j}}\to0$ and $f_{n_{k_j}} \rightarrow g(L_1,L_2)$.  This implies that the RHS of \eqref{eq:measrelent-ub} converges to $g(L_1,L_2)$. Via analogous arguments, the RHS of \eqref{eq:measrelent-lb} also converges a.s. to $g(L_1,L_2)$ along the sequence $(n_{k_j})_{j \in \NN}$, which implies that $g_{n_{k_j}} \rightarrow g(L_1,L_2)$ a.s, and hence also weakly as claimed. This completes the proof of the theorem.
 \end{proof}

\end{appendices}

\bibliographystyle{IEEEtran}
\bibliography{ref,ref-quant}
\end{document}